\documentclass[10pt]{llncs}

\usepackage{amssymb}

\usepackage{xspace}
\usepackage[dvipsnames]{xcolor}
\usepackage{tikz}
\usetikzlibrary{calc,patterns,arrows.meta}

\usepackage{hyperref}
\usepackage{mathtools}
\usepackage{enumitem}

\setlist{noitemsep}
\setlist[2]{noitemsep}

\spnewtheorem{observation}[theorem]{Observation}{\bfseries}{\itshape}
\usepackage{thmtools, thm-restate}

\newcommand{\cl}[1]{\ensuremath{\mathsf{#1}}}
\newcommand{\ds}{\textsc{ds}\xspace}
\newcommand{\cds}{\textsc{cds}\xspace}
\newcommand{\dsudg}{\textsc{ds-udg}\xspace}
\newcommand{\cdsudg}{\textsc{cds-udg}\xspace}
\newcommand{\nph}{\cl{NP}-hard\xspace}

\newcommand{\apxh}{\cl{APX}-hard\xspace}
\newcommand{\p}{\cl{P}\xspace}
\newcommand{\np}{\cl{NP}\xspace}
\newcommand{\xp}{\cl{XP}\xspace}

\newcommand{\Wone}{\cl{W[1]}\xspace}

\newcommand{\disk}{\delta}
\newcommand{\mypath}{\pi}

\renewcommand{\leq}{\leqslant}
\renewcommand{\geq}{\geqslant}
\renewcommand{\le}{\leqslant}
\renewcommand{\ge}{\geqslant}
\newcommand{\mypara}[1]{\vspace{10pt} \noindent \textbf{\sffamily #1}}

\newcommand{\Reals}{\mathbb{R}}

\renewcommand{\to}{\rightarrow}

\newcommand{\cd}{\mathcal{D}}
\newcommand{\cS}{\mathcal{S}}
\newcommand{\cu}{\mathcal{U}}
\newcommand{\cR}{\mathcal{R}}

\newcommand{\graph}{\mathcal{G}}
\renewcommand{\deg}{^{\circ}}

\newcommand{\Zleftsmall}{Z^-_{\leq}}
\newcommand{\Zrightsmall}{Z^+_{\leq}}
\newcommand{\Zleftlarge}{Z^-_{>}}
\newcommand{\Zrightlarge}{Z^+_{>}}
\newcommand{\tree}{\mathcal{T}}
\newcommand{\opt}{\mbox{{\sc opt}}}
\newcommand{\lchild}{\mbox{left-child}}
\newcommand{\rchild}{\mbox{right-child}}

\newcommand{\myclaim}[2]
{\begin{quotation} \noindent {\emph{Claim.} }{#1} \\[2mm]
{\emph{Proof of claim.}} #2 \hfill {\footnotesize $\Box$} \end{quotation}}

\tikzstyle{cp}=[draw=black,thick]
\tikzstyle{cl}=[draw=red, fill=red, fill opacity=0.1]
\tikzstyle{xb}=[draw=black, fill=black, fill opacity=0.4, very thick]
\tikzstyle{yb}=[draw=black, fill=none, very thick]

\DeclarePairedDelimiter{\civ}{[}{]}

\DeclarePairedDelimiter{\loiv}{(}{]}
\DeclarePairedDelimiter{\roiv}{[}{)}

\DeclareMathOperator{\sep}{sep}
\DeclareMathOperator{\cost}{cost}
\DeclareMathOperator{\core}{core}

\newcommand{\skb}[1]{{#1}}

\newenvironment{proofsketch}{\noindent\emph{Proof sketch.}}{\hfill $\Box$ \medskip\\}
\DeclareMathOperator{\mypred}{\mathit{pred}}
\DeclareMathOperator{\myray}{\mathit{ray}}
\DeclareMathOperator{\bundle}{\mathit{bundle}}
\newcommand{\mytrue}{{\sc true}\xspace}
\newcommand{\myfalse}{{\sc false}\xspace}
\newcommand{\mynext}[1]{\mathit{next}(#1)}

\newcommand{\BeginMyItemize}{\begin{itemize}\setlength{\itemsep}{-\parskip}}
\newcommand{\EndMyItemize}{\end{itemize}}

\newcommand{\BeginMyEnumerate}{\begin{enumerate}\setlength{\itemsep}{-\parskip}}
\newcommand{\EndMyEnumerate}{\end{enumerate}}

\tikzstyle{xb}=[draw=black, fill=black, fill opacity=0.2, very thick]
\tikzstyle{yb}=[draw=black, fill=none, very thick]
\tikzstyle{xbi}=[draw=black, fill=black, fill opacity=0.2, thick]
\tikzstyle{ybi}=[draw=black, fill=none, thick]

\newcommand{\cc}[2]{
\draw [draw=blue, fill=blue, fill opacity=0.15] (#1) circle (1);
\draw [draw=red, fill=red, fill opacity=0.08] (#2) circle (1);
}

\newcommand{\tapeblock}[2]{
\begin{scope}[shift={(#1)},rotate around={{#2}:(0,0)}]
\draw [RedViolet] (0,0) circle (1);
\draw [RedViolet] (0.1,0) circle (1);
\draw [RedViolet] (0.2,0) circle (1);
\end{scope}
}

\newcommand{\tf}[2]{
\draw [draw=ForestGreen, fill=ForestGreen, fill opacity=0.15] (#1) circle
(1);
\draw [draw=YellowOrange, fill=YellowOrange, fill opacity=0.15](#2) circle
(1); }

\newcommand{\turnleft}[2]{
\begin{scope}[shift={(#1)},rotate around={{#2}:(0,0)}]
\tf{0,0.7}{0,0}
\tf{1.15,1.7}{1.95,0.0}
\tf{1.6,2.15}{3.3,1.35}
\tf{2.6,3.3}{3.3,3.3}
\end{scope}
}

\newcommand{\turnright}[2]{
\begin{scope}[shift={(#1)},rotate around={{#2}:(0,0)}]
\tf{0,0}{0,0.7}
\tf{1.95,0.0}{1.15,1.7}
\tf{3.3,1.35}{1.6,2.15}
\tf{3.3,3.3}{2.6,3.3}
\end{scope}
}

\newcommand{\ccc}[2]{
\draw [draw=blue, fill=blue, fill opacity=0.15] (#1) circle (1);
\draw [draw=red, pattern=north west lines, pattern color=red] (#2) circle
(1); }

\newcommand{\connectors}[1]{
\begin{scope}[shift={(#1)}]
\cc{17.3,10.4}{18,11.55}
\cc{18.7,9.6}{18,8.45}
\end{scope}
}

\newcommand{\sevencircles}[1]
{
\begin{scope}[rotate around={{#1}:(8,8)}]
\cc{\base,2-\de}{\base,3-\de}
\cc{\base + \shift,2+\de}{\base + \shift,2+2*\de}
\cc{\base + 2*\shift,2-\de}{\base + 2*\shift,2+6*\de}
\cc{\base + 3*\shift,2+\de}{\base + 3*\shift,2+2*\de}
\cc{\base + 4*\shift,2-\de}{11,2-\de}
\cc{\base + 4*\shift,4-2*\de}{11,4}
\cc{8,4+3*\de}{7,4+3*\de}
\end{scope}
}

\newcommand{\gadgetgraph}[1]
{
\begin{scope}[rotate around={{#1}:(8,8)}]
\draw [dashed, thick] (0,2) node{} -- (2,1.6) node{} -- (2,0) node{};
\draw [dashed, thick] (2,1.6) -- (4,2) node{} -- (6,2) node{} -- (6,0)
node{};
\draw [dashed, thick] (6,2) -- (8,2) node{} -- (9.75,2) node{} -- (10,0)
node{};
\draw [dashed, thick] (9.75,2) -- (9.7,4) node{} -- (8.2,4) node{} -- (8,6)
node{};
\draw [dashed, thick] (8,6) -- (10,6) node{} -- (10,8) node{};
\end{scope}
}

\newcommand{\graphconnector}[1]{
\begin{scope}[shift={(#1)}]
\draw [dashed, thick] (16,10) node{} -- (17.3,10.4) node{}
-- (18.7,9.6) node{} -- (20,10) node{};
\end{scope}
}

\newcommand{\blocks}
{
\node at (2,0) {$X_6$}; \node at (4,0) {$Y_6$};
\node at (6,0) {$X'_6$}; \node at (8,0) {$Y'_6$};
\node at (10,0) {$X'_5$}; \node at (12,0) {$Y'_5$};
\node at (14,0) {$X_5$}; \node at (16,0) {$Y_5$};
\node at (16,2) {$X_4$}; \node at (16,4) {$Y_4$};
\node at (16,6) {$X'_4$}; \node at (16,8) {$Y'_4$};
\node at (16,10) {$X'_3$}; \node at (16,12) {$Y'_3$};
\node at (16,14) {$X_3$}; \node at (16,16) {$Y_3$};
\node at (14,16) {$X_2$}; \node at (12,16) {$Y_2$};
\node at (10,16) {$X'_2$}; \node at (8,16) {$Y'_2$};
\node at (6,16) {$X'_1$}; \node at (4,16) {$Y'_1$};
\node at (2,16) {$X_1$}; \node at (0,16) {$Y_1$};
\node at (0,14) {$X_8$}; \node at (0,12) {$Y_8$};
\node at (0,10) {$X'_8$}; \node at (0,8) {$Y'_8$};
\node at (0,6) {$X'_7$}; \node at (0,4) {$Y'_7$};
\node at (0,2) {$X_7$}; \node at (0,0) {$Y_7$};
}

\newcommand{\blocksswap}
{
\node at (0,0) {$Y_7$}; \node at (0,2) {$X_7$};
\node at (2,0) {$X_6$}; \node at (4,0) {$Y_6$};
\node at (6,0) {$X'_6$}; \node at (8,0) {$Y'_6$};
\node at (10,0) {$X'_5$}; \node at (12,0) {$Y'_5$};
\node at (14,0) {$X_5$}; \node at (16,0) {$Y_5$};
\node at (16,2) {$X_4$}; \node at (16,4) {$Y_4$};
\node at (16,6) {$X'_4$}; \node at (16,8) {$Y'_4$};
\node at (16,10) {$X'_3$}; \node at (16,12) {$Y'_3$};
\node at (16,14) {$X_3$}; \node at (16,16) {$Y_3$};
\node at (14,16) {$X_1$}; \node at (12,16) {$Y_2$};
\node at (10,16) {$X'_2$}; \node at (8,16) {$Y'_2$};
\node at (6,16) {$X'_1$}; \node at (4,16) {$Y'_1$};
\node at (2,16) {$X_2$}; \node at (0,16) {$Y_1$};
\node at (0,14) {$X_8$}; \node at (0,12) {$Y_8$};
\node at (0,10) {$X'_8$}; \node at (0,8) {$Y'_8$};
\node at (0,6) {$X'_7$}; \node at (0,4) {$Y'_7$};
}

\begin{document}

\title{The Homogeneous Broadcast Problem in Narrow and Wide Strips\thanks{This research was
supported by the Netherlands Organization for Scientific Research (NWO) under
project no.~024.002.003.}}

\author{Mark de Berg\inst{1}
            \and
        Hans L. Bodlaender\inst{1}\inst{2}
            \and
        S\'andor Kisfaludi-Bak\inst{1}
}

\institute{Department of Mathematics and Computer Science, TU Eindhoven, The Netherlands\\
\email{m.t.d.berg@tue.nl; h.l.bodlaender@tue.nl; s.kisfaludi.bak@tue.nl}
\and
Department of Computer Science, Utrecht University, The Netherlands\\
}

\maketitle

\begin{abstract}
Let $P$ be a set of nodes in a wireless network, where each node is modeled
as a point in the plane, and let $s\in P$ be a given source node. Each
node~$p$ can transmit information to all other nodes within unit distance,
provided~$p$ is activated. The (homogeneous) broadcast problem is to activate a minimum
number of nodes such that in the resulting directed communication graph, the
source~$s$ can reach any other node. \skb{We study the complexity of the regular
and the hop-bounded version of the problem (in the latter, $s$ must be able
to reach every node within a specified number of hops), with the restriction
that all points lie inside a strip of width $w$. We almost completely
characterize the complexity of both the regular and the hop-bounded versions
as a function of the strip width $w$.}
\end{abstract}

\section{Introduction}
Wireless networks give rise to a host of interesting algorithmic problems. In
the traditional model of a wireless network each node is modeled as a
point~$p\in \Reals^2$, which is the center of a disk~$\disk(p)$ whose radius
equals the transmission range of~$p$. Thus~$p$ can send a message to another
node~$q$ if and only if $q\in \disk(p)$. Using a larger transmission radius
may allow a node to transmit to more nodes, but it requires more power and is
more expensive. This leads to so-called range-assignment problems, where the
goal is to assign a transmission range to each node such that the resulting
communication graph has desirable properties, while minimizing the cost of
the assignment. We are interested in broadcast problems, where the desired
property is that a given source node can reach any other node in the
communication graph. Next, we define the problem more formally.

Let $P$ be a set of $n$ points in~$\Reals^d$ and let $s\in P$ be a source
node. A \emph{range assignment} is a function $\rho: P \to \Reals_{\geq 0}$
that assigns a transmission range $\rho(p)$ to each point~$p\in P$. Let
$\graph_\rho=(P,E_\rho)$ be the directed graph where $(p,q)\in E_\rho$ iff
$|pq|\leq \rho(p)$. The function~$\rho$ is a \emph{broadcast assignment} if
every point $p \in P$ is reachable from~$s$ in $\graph_\rho$. If
every~$p\in P$ is reachable within $h$ hops, for a given parameter~$h$, then
$\rho$ is an \emph{$h$-hop broadcast assignment}. The ($h$-hop) broadcast
problem is to find an ($h$-hop) broadcast assignment whose cost $\sum_{p \in
P} \cost(\rho(p))$ is minimized. Often the cost of assigning transmission
radius $x$ is defined as $\cost(x)= x^{\alpha}$ for some constant~$\alpha$.
In $\Reals^1$, both the basic broadcast problem and the $h$-hop version are
solvable in $O(n^2)$ time \cite{Das06linear}. In $\Reals^2$ the problem is
\nph for any $\alpha>1$~\cite{Clementi2001worst,Fuchs08}, and in $\Reals^3$
it is even \apxh~\cite{Fuchs08}. There are also several approximation
algorithms~\cite{Ambuhl05,Clementi2001worst}. For the 2-hop broadcast problem
in $\Reals^2$ an $O(n^7)$ algorithm is known~\cite{Ambuhl04} and for any
constant~$h$ there is a PTAS~\cite{Ambuhl04}. Interestingly, the complexity
of the 3-hop broadcast problem is unknown.

An important special case of the broadcast problem is where we allow only two
possible transmission ranges for the points, $\rho(p)=1$ or $\rho(p)=0$. In
this case the exact cost function is irrelevant and the problem becomes to
minimize the number of active points. This is called the \emph{homogeneous
broadcast problem} and it is the version we focus on. From now on, all
mentions of broadcast and $h$-hop broadcast refer to the homogeneous setting.
Observe that if $\rho(p)=1$ then $(p,q)$ is an edge in $\graph_{\rho}$ if and
only if the disks of radius~$1/2$ centered at $p$ and $q$ intersect. Hence,
if all points are active then $\graph_{\rho}$ in the intersection graph of a
set of congruent disks or, in other words, a \emph{unit-disk graph (UDG)}.
Because of their relation to wireless networks, UDGs have been studied
extensively.

Let $\cd$ be a set of congruent disks in the plane, and let $\graph_{\cd}$ be
the UDG induced by~$\cd$. A \emph{broadcast tree} on $\graph_{\cd}$ is a
rooted spanning tree of~$\graph_{\cd}$. To send a message from the root to all other nodes, each internal node of the tree has to
send the message to its children. Hence, the cost of broadcasting is related
to the internal nodes in the broadcast tree. A cheapest broadcast
tree corresponds to a minimum-size \emph{connected dominating set} on
$\graph_{\cd}$, that is, a minimum-size subset $\Delta \subset \cd$ such that
the subgraph induced by $\Delta$ is connected and each node in $\graph_{\cd}$
is either in $\Delta$ or a neighbor of a node in~$\Delta$. The
broadcast problem is thus equivalent to the following: given a UDG
$\graph_{\cd}$ with a designated source node~$s$, compute a minimum-size
connected dominated set $\Delta \subset \cd$ such that $s\in \Delta$.

In the following we denote the dominating set problem by \ds, the connected
dominating set problem by \cds, and we denote these problems on UDGs by
\dsudg and \cdsudg, respectively. Given an algorithm for the broadcast
problem, one can solve \cdsudg by running the algorithm $n$ times, once for
each possible source point. Consequently, hardness results for \cdsudg can be
transferred to the broadcast problem, and algorithms for the broadcast
problem can be transferred to \cdsudg at the cost of an extra linear factor
in the running time. It is well known that \ds and \cds are \nph, even for
planar graphs~\cite{Garey79}. \dsudg and \cdsudg are also
\nph~\cite{lichtenstein82,masuyama81}. The parameterized complexity of \dsudg
has also been investigated: Marx~\cite{Marx06} proved that \dsudg is
\Wone-hard when parameterized by the size of the dominating set. (The
definition of \Wone and other parameterized complexity classes can be found in
the book by Flum and Grohe~\cite{Flum06}.)

\mypara{Our contributions.} Knowing the existing hardness results for the
broadcast problem, we set out to investigate the following questions. Is
there a natural special case or parameterization admitting an efficient
algorithm? Since the broadcast
problem is polynomially solvable in $\Reals^1$, we study how the complexity of
the problem changes as we go from the $1$-dimensional problem to the
$2$-dimensional problem. To do this, we assume the points (that is, the disk
centers) lie in a strip of width $w$, and we study how the problem complexity
changes as we increase~$w$. We give an almost complete characterization of 
the complexity, both for the general and for the hop-bounded version of the
problem. More precisely, our results are as follows.


We first study strips of width at most~$\sqrt{3}/2$. Unit disk graphs restricted to  such \emph{narrow strips} are a subclass of co-comparability graphs \cite{Matsui98}, for which an $O(nm)$  time \cds algorithm is known~\cite{Kratsch93,Breu96}. (Here $m$ denotes the number of edges in the graph.) The broadcast problem is slightly different because it requires $s$ to be in the dominating set; still, one would expect better running times in this restricted graph class. Indeed, we show that for narrow strips the broadcast problem can be solved in $O(n \log n)$ time. The hop condition in the $h$-hop broadcast problem has not been studied yet for co-comparability graphs to our knowledge. This condition complicates the problem considerably. Nevertheless, we show that the $h$-hop broadcast problem in narrow strips is solvable in polynomial time. Our algorithm runs in $O(n^6)$ and uses a subroutine for $2$-hop broadcast, which may be of independent interest: we show that the $2$-hop broadcast problem  is solvable in $O(n^4)$ time. Our subroutine is based on an algorithm by Amb\"{u}hl et al. \cite{Ambuhl04} for the non-homogeneous case, which runs in $O(n^7)$ time. This result is deferred to Appendix~\ref{sec:planaralgs}.

Second, we investigate what happens for wider strips. We show that the
broadcast problem has an $n^{O(w)}$ dynamic-programming algorithm for strips
of width $w$. We prove a matching lower bound of $n^{\Omega(w)}$, conditional
on the Exponential Time Hypothesis (ETH). Interestingly, the $h$-hop
broadcast problem has no such algorithm (unless $\p=\np$): we show this
problem is already \nph on a strip of width~$40$. One of the gadgets in this
intricate construction can also be used to prove that a \cdsudg and the
broadcast problem are \Wone-hard parameterized by the solution size $k$. The 
\Wone-hardness proof is discussed in Section~\ref{sec:AppHardness}. It is a 
reduction from \textsc{Grid Tiling} based on ideas by Marx~\cite{Marx06}, 
and it implies that there is no $f(k)n^{o(\sqrt{k})}$ algorithm for \cdsudg 
unless ETH fails.

\section{Algorithms for broadcasting inside a narrow strip}
\label{sec:narrowstrip}
In this section we present polynomial algorithms (both for broadcast and for
$h$-hop broadcast) for inputs that lie inside a strip~$\cS :=
\mathbb{R}\times \civ{0,w}$, where $0<w\leq \sqrt{3}/2$ is the width of the
strip. Without loss of generality, we assume that the source lies on the
$y$-axis. Define $\cS_{\geq 0} := [0,\infty) \times \civ{0,w}$ and $\cS_{\leq
0} := (-\infty,0] \times \civ{0,w}$.

Let $P$ be the set of input points. We define $x(p)$ and $y(p)$ to be the
$x$- and $y$-coordinate of a point~$p\in P$, respectively, and $\disk(p)$
to be the unit-radius disk centered at~$p$. Let $\graph=(P,E)$ be the graph
with $(p,q)\in E$ iff $q\in\disk(p)$, and let ${P}' := P \setminus
\disk(s)$ be the set of input points outside the source disk. We say that a
point $p\in P$ is \emph{left-covering} if $pp'
\in E$ for all $p' \in {P}'$ with $x(p') < x(p)$; $p$ is
\emph{right-covering} if $p'p \in E$ for all $p' \in {P}'$ with $x(p') >
x(p)$. We denote the set of left-covering and right-covering points by
$Q^-$ and $Q^+$ respectively. Finally, the \emph{core area} of a point
$p$, denoted by $\core(p)$, is
$\civ{x(p)-\frac12,x(p)+\frac12}\times\civ{0,w}$. Note that $\core(p) \subset
\disk(p)$ because $w \le \sqrt{3}/2$, i.e., the disk of $p$ covers a part of
the strip that has horizontal length at least one. This is a key property of
strips of width at most $\sqrt{3}/2$, and will be used repeatedly.

We partition $P$ into levels $L_0,L_1,\dots L_t$, based on hop distance
from $s$ in~$\graph$. Thus $L_i := \{p\in P: d_{\graph}(s,p)=i\}$, where
$d_{\graph}(s,p)$ denotes the hop-distance. Let $L^-_i$ and $L^+_i$ denote
the points of $L_i$ with negative and nonnegative coordinates, respectively.
We will use the following observation multiple times.

\begin{restatable}{observation}{obscorecover}
\label{obs:corecover}
Let $\graph=(P,E)$ be a unit disk graph on a narrow strip $\cS$.\vspace*{-2mm}
\begin{enumerate}
\item[(i)] Let $\mypath$ be a path in $\graph$ from a point $p\in P$ to a
	point $q\in P$. Then the region $\civ{x(p)-\frac12, x(q)+\frac12}\times
	\civ{0,w}$ is fully covered by the disks  of the points in~$\mypath$.
\item[(ii)] The overlap of neighboring levels is at most $\frac{1}{2}$ in
    $x$-coordinates: $\max\{x(p)|p\in L^+_{i-1}\} \leq \min\{x(q)|q\in
    L^+_i\} + \frac{1}{2}$ for any $i>0$ with $L^+_i\neq \emptyset$; similarly,
    $\min\{x(p)|p\in L^-_{i-1}\} \geq \max\{x(q)|q\in L^-_i\} - \frac{1}{2}$
    for any $i>0$ with $L^-_i\neq \emptyset$.
\item[(iii)] Let $p$ be an arbitrary point in $L^+_i$ for some $i>0$. Then
	the disks of any path $\mypath(s,p)$ cover all points in all levels $L_0
	\cup L_1 \cup L^+_2 \cup \dots \cup L^+_{i-1}$. A similar statement holds for points in $L^-_i$.
\end{enumerate}
\end{restatable}
\begin{proof}
For (i), note that for any edge $(u,v)\in E$, we have that $\core(u)$ and
$\core(v)$ intersect. Thus the union of the cores of the points of $\mypath$
is connected, and contains $\core(p)$ and $\core(q)$. Consequently, it covers
$\civ{x(p)-\frac12, x(q)+\frac12}\times \civ{0,w}$.

We prove (ii) by contradiction. Suppose that there are $p\in L_{i-1}$ and
$q\in L_i$ with $x(p) > x(q)+\frac12$. Any shortest path $\mypath(s,p)$ must
have a point $p'$ inside~$\civ{x(q)-\frac12,x(q)+\frac12}\times \civ{0,w}$,
because no edge of the path can jump over this part of the strip. This
point~$p'$ has level at most~$i-2$ and $q\in \disk(p')$, contradicting that
$q$ is at level~$i$.

Statement~(iii) follows from (i) and (ii): the disks of $\mypath(s,p)$ cover
$\disk(s) \cup \civ{-\frac12, x(p)+\frac12}\times \civ{0,w}$, and $L_0
\cup L_1 \cup L^+_2 \cup \dots \cup L^+_{i-1}$ is contained in this set. \qed
\end{proof}

\subsection{Minimum broadcast set in a narrow strip} \label{subse:non-hop-bounded}
A \emph{broadcast set} is a point set $D \subseteq P$ that gives a
feasible broadcast, i.e., a connected dominating set of $\graph$ that
contains $s$. Our task is to find a minimum broadcast set inside a narrow
strip. Let $p,p'\in P$ be points with maximum and minimum $x$-coordinate, respectively.
Obviously there must be paths from $s$ to $p$ and $p'$ in $\graph$ such that all points
on these paths are active, except possibly $p$ and $p'$.
If $p$ and $p'$ are also active, then these
paths alone give us a feasible broadcast set: by
Observation~\ref{obs:corecover}(i), these paths cover all our input points.
Instead of activating $p$ and $p'$, it is also enough to activate the
points of a path that reaches $Q^-$ and a path that reaches~$Q^+$.
In most cases it is sufficient to look for broadcast sets with this structure.
\begin{restatable}{lemma}{lemtwopaths}
\label{lem:twopaths}
If there is a minimum broadcast set with an active point on $L_2$,
then there is a minimum broadcast set consisting of the disks of a
shortest path $\mypath^-$ from $s$ to $Q^-$ and a shortest path $\mypath^+$
from $s$ to $Q^+$.  These two paths share $s$ and they may or may not
share their first point after $s$. 
\end{restatable}

\begin{figure}
\begin{center}
\begin{tikzpicture}[x=0.8 cm,y=0.8 cm,
	strip/.style={black, line width=3pt}
]
\clip (1,0) rectangle (7,1.73);
\draw [very thick] (0,0) -- (7,0);
\draw [very thick] (0,1.73) -- (7,1.73);
\tikzstyle{every node}=[draw,circle,fill=black,minimum size=3pt,
inner sep=0pt]
	\node [ForestGreen,label={[ForestGreen]above:{$s$}}] (s) at (2,1) {};
\node (r1) at (3.3,1.2) {};
\node [label=above:$a$] (r2) at (4.3,0.1) {};
\node [label=below:$b$] (r3) at (4.35,1.6) {};
\node [label=above left:$\bar{a}$] (r4) at (6,0.3) {};
\node [label=below right:$\bar{b}$] (r5) at (5.9,1.5) {};
\draw (s) -- (r1) --(r2);
\draw [dashed] (r2) --(r4);
\draw (r1) -- (r3);
\draw [dashed] (r3) -- (r5);
\filldraw [ForestGreen, fill opacity=0.2] (s.center) circle (2);
\filldraw [black, fill opacity=0.05, dotted] (r1.center) circle (2);
\filldraw [black, fill opacity=0.2] (r2.center) circle (2);
\filldraw [black, fill opacity=0.2] (r3.center) circle (2);
\end{tikzpicture}
\hspace{0.3cm}
\begin{tikzpicture}[x=0.8 cm,y=0.8 cm,
	strip/.style={black, line width=3pt}
]
\clip (1,0) rectangle (8,1.73);
\draw [very thick] (0,0) -- (8,0);
\draw [very thick] (0,1.73) -- (8,1.73);
\tikzstyle{every node}=[draw,circle,fill=black,minimum size=3pt,
inner sep=0pt]
\node [ForestGreen,label={[ForestGreen]above:{$s$}}] (s) at (2,1) {};
\node (r1) at (3.3,1.2) {};
\node [label=above right:$a$] (r2) at (4.3,0.1) {};
\node [label=below:$b$] (r3) at (4.35,1.6) {};
\node [label=right:$\bar{a}$] (r4) at (6,0.3) {};
\node [red,label=below right:$\bar{b}$] (r5) at (5.9,1.5) {};
\draw (s) -- (r1) --(r2);
\draw (r1) -- (r3) -- (r5);
\draw [dashed] (r5) -- (r4);
\filldraw [ForestGreen, fill opacity=0.2] (s.center) circle (2);
\filldraw [black, fill opacity=0.05, dotted] (r1.center) circle (2);
\filldraw [red, fill opacity=0.2] (r5.center) circle (2);
\filldraw [black, fill opacity=0.2] (r3.center) circle (2);
\end{tikzpicture}
\end{center}
\caption{A swap operation. The edges of the broadcast tree are solid lines.}\label{fig:swap}
\end{figure}
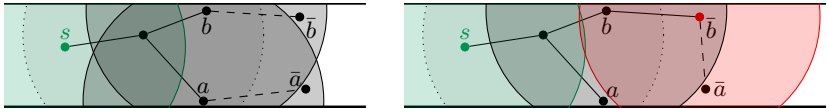

\begin{proof} We begin by showing that there is a minimum broadcast that
intersects both $Q^+$ and $Q^-$.

\myclaim{There is a minimum broadcast set~$D'$ containing a point in $Q^+$.}
{Let $D$ be a minimum broadcast set. Without loss of generality, we may
assume that $L^+_2$ has an active point. It follows that this active point in
$L^+_2$ has a descendant leaf $a \in L^+_{\geq 2}$ in the broadcast tree (the
tree one gets by performing breadth first search from $s$ in the graph
spanned by $D$). Note that $\disk(a)$ does not cover any points in $\cS_{\leq
0}\setminus \disk(s)$, since $a \not\in\core(s)$ and $\core(s)$ has width
$1$.

Suppose that $D\cap Q^+=\emptyset$. Since $a \not\in Q^+$, there is a point
$\bar{b}$ with a larger $x$-coordinate than $a$ which is not covered by
$\disk(a)$, but covered by another disk $\disk(b)$ for some $b \in D$.
Similarly, there must be a point $\bar{a}\in \disk(a)\setminus \disk(b)$ with
$x(\bar{a})>x(b)$ (see Fig.~\ref{fig:swap} for an example). Since
$\disk(b)$ covers $\core(b)$, we have $x(\bar{a})>x(b)+\frac12$, and
similarly $x(\bar{b}) > x(a)+\frac12$.

Note that $x(\bar{b})\le x(b)+1$, so $x(\bar{b})-x(\bar{a})<\frac12$. The
other direction yields $x(\bar{b})-x(\bar{a})>-\frac12$, thus $\bar{a}\in
\disk(\bar{b})$, or in other words, any point covered by $\disk(a)$ to the
right of $\disk(b)$ can be covered by replacing $\disk(a)$ with
$\disk(\bar{b})$. We call such a replacement a \emph{swap operation}. This
operation results in a new minimum broadcast set, because the size of the set
remains the same, and no vertex can become disconnected from the source on
either side: the right side remains connected along the broadcast tree, and
the left is untouched since $\disk(a)\cap \cS_{\leq 0} \subseteq
\disk(s)$. Repeated swap operations lead to a minimum-size broadcast set~$D'$
that contains at least one point from $Q^+$. (The procedure terminates since
the sum of the $x$-coordinates of the active points increases.)}
%
%
The resulting minimum broadcast set $D'$ contains a path $\mypath^+$ from $s$
to $Q^+$. Let $a^+$ be the last point on $\mypath^+$ that falls in $L_1$.
Without loss of generality, we can assume that the first two points of
$\mypath^+$ are $s$ and $a^+$. Let $q^+=Q^+\cap \mypath^+$. By part (iii) of
Observation~\ref{obs:corecover}, the disks around the points of $\mypath^+$
cover all points with $x$ coordinates between $0$ and $x(q^+)+\frac12$; and
$q^+\in Q^+$ implies that it covers all input points with $x$-coordinate
higher than $x(q^+)+\frac12$. Consequently, there are no active points in the
right part outside this path---that is, no active points in $\cS_{\geq 0} \setminus
\big(\disk(s) \cup \mypath^+ \big)$)---since those could be removed while
maintaining the feasibility of the solution.

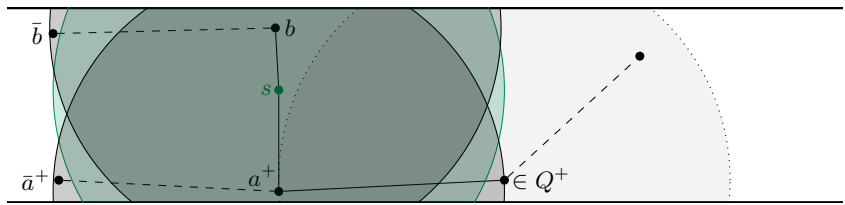
\begin{figure}
\begin{center}
\begin{tikzpicture}[x=1.5 cm,y=1.5 cm,
	strip/.style={black, line width=3pt}
]
\clip (-0.4,0) rectangle (7,1.73);
\draw [very thick] (-1,0) -- (8,0);
\draw [very thick] (-1,1.73) -- (8,1.73);
\tikzstyle{every node}=[draw,circle,fill=black,minimum size=3pt,
inner sep=0pt]
\node [ForestGreen, label=,label={[ForestGreen]left:{$s$}}] (s) at (2,1) {};
\node [label=above left:$a^+$] (r2) at (2,0.1) {};
\node [label=right:$b$] (r3) at (1.97,1.55) {};
\node [label=left:$\bar{a}^+$] (r4) at (0.05,0.2) {};
\node [label=left:$\bar{b}$] (r5) at (0,1.5) {};
\node [label=right:$\in Q^+$] (r6) at (4,0.2) {};
\node (r7) at (5.2,1.3) {};
\draw (r6) -- (r2) -- (s) -- (r3);
\draw [dashed] (r4) -- (r2) (r6) -- (r7);
\draw [dashed] (r3) -- (r5);
\filldraw [ForestGreen, fill opacity=0.2] (s.center) circle (2);
\filldraw [black, fill opacity=0.05, dotted] (r6.center) circle (2);
\filldraw [black, fill opacity=0.2] (r2.center) circle (2);
\filldraw [black, fill opacity=0.2] (r3.center) circle (2);
\end{tikzpicture}
\end{center}
\caption{If $D' \cap L^-_2 = \emptyset$, we can still do
swaps.}\label{fig:apluscoversleft}
\end{figure}

\myclaim{There is a minimum broadcast set~$D'$ containing a point in $Q^+$ and one in~$Q^-$.} {
If there is a disk in $D'\cap L^-_2$ as well, then we can reuse the previous
argument for the other side, and get a broadcast set that contains a path
$\mypath^-$ from $s$ to $Q^-$. Otherwise, we need to be slightly more careful
with our swap operations: we need to make sure not to remove $a^+$. If $a^+
\not \in Q^-$, then we can again use the previous argument: it is possible to
find another disk $b$, and corresponding uniquely covered points $\bar{a}^+$
and $\bar{b}$ (see Fig.~\ref{fig:apluscoversleft}). Note that $b\in
\disk(s)$ since we are in the case $L^-_2=\emptyset$. We argue that $b$ can
be replaced with $\bar{a}^+$: removing $b$ can not disconnect anything from
$s$ on either side, and $\disk(\bar{a}^+)$ covers all points covered by
$\disk(b)$. Repeated swap operations lead to a minimum broadcast set $D''$
that contains points from both $Q^+$ and $Q^-$.}
Let $\mypath^-$ and $a^-$ be defined analogously to how $\mypath^+$ and $a^+$
were defined above. Note that $a^+$ and $a^-$ might coincide. Since
$\mypath^+ \cup \mypath^-$ is connected and covers all points, we have
$D''=\mypath^+ \cup \mypath^-$. To finish the proof, it remains to argue that
we can take $\mypath^+$ and $\mypath^-$ to be shortest paths to $Q^+$ and
$Q^-$. Suppose  $\pi^+$ is not a shortest path to $Q^+$. (The argument for
$\pi^-$ is similar.) Then we can replace $\pi^+ \cup \pi^-$ by $\bar{\pi}^+
\cup \pi^-$ $\overline{\pi}^+$ is a shortest path from $s$ to $Q^+$. Since
$\pi^+$ and $\pi^-$ share at most one point besides~$s$, this replacement
does not increase the size of the solution. \qed
\end{proof}


Lemma~\ref{lem:narrowstructure} below fully characterizes optimal broadcast
sets. To deal with the case where Lemma~\ref{lem:twopaths} does not apply, we
need some more terminology. We say that the disk $\disk(q)$ of an active
point $q$ in a feasible broadcast set is \emph{bidirectional} if there are
two input points $p^- \in L^-_2$ and $p^+ \in L^+_2$ that are covered only by
$\disk(q)$. See points $p$ and $p'$ in Fig.~\ref{fig:strangesolution} for an
example. Note that $q \in \core(s)$, because $\core(s) = \civ{-\frac12,
\frac12} \times \civ{0,w}$ is covered by $\disk(s)$, and our bidirectional
disk has to cover points both in $\loiv{-\infty, -\frac12} \times \civ{0,w}$
and $\roiv{\frac12,\infty}\times\civ{0,w}$. Active disks that are not the
source disk and not bidirectional are called \emph{monodirectional}.
\begin{restatable}{lemma}{lemnarrowstructure}
\label{lem:narrowstructure}
For any input $P$ that has a feasible broadcast set, there is a minimum
broadcast set $D$ that has one of the following structures. \vspace*{-2mm}
\begin{enumerate}
\item[(i)] \emph{Small}: $|D|\leq 2$.
\item[(ii)] \emph{Path-like}: $|D|\ge 3$, and $D$ consists of a shortest path $\mypath^-$
    from $s$ to $Q^-$ and a shortest path $\mypath^+$ from $s$ to $Q^+$;
    $\mypath^+$ and $\mypath^-$ share $s$ and may or may not share their first point after $s$.
\item[(iii)] \emph{Bidirectional}: $|D|=3$, and $D$ contains two
bidirectional disk centers and $s$.
\end{enumerate}
\end{restatable}

\begin{proof}
Let $\opt$ be the size of a minimum broadcast set. First consider the case
$\opt\geq 4$. By Lemma~\ref{lem:twopaths} it suffices to prove that there is
an active point in~$L_2$. If $L_3\neq \emptyset$ this is trivially true, so
assume that $L_3=\emptyset$. Since $\opt \geq 4$, it follows that $L^+_2 \neq
\emptyset$ otherwise activating the shortest path from $s$ to the point with
minimum $x$-coordinate is a feasible broadcast set of size at most $3$.
Similarly, $L^-_2 \neq \emptyset$.

If $Q^+ \cap L_1 \neq \emptyset$, then there is a minimum broadcast set with
an active point in $L_2$: we take $s$, a point from $Q^+ \cap L_1$, and a
shortest path from $s$ to the leftmost point (at most two more points). Thus
we may assume that $Q^+$, and similarly, $Q^-$ are disjoint from $L_1$.

Let $\{s,p_1,p_2,p_3\}$ be a subset of a minimum broadcast set. If
$\disk(p_i)$ is monodirectional, then let $\bar{p}_i\in L_2$ be a point
uniquely covered by $p_i$; suppose that $\bar{p}_i \in \cS^+$ (the proof is
the same for the left side). Since $p_i \not \in Q^+$, there is a point $q\in
L_1$ that uniquely covers another point $\bar{q} \in L_2$. We can swap $p_i$
for $\bar{q}$ and get the desired outcome.

If all of $\disk(p_i)$ are bidirectional, then we can do a \emph{double swap
operation}: deactivate both $\disk(p_1)$ and $\disk(p_2)$, and activate
$\disk(a^-)$ and $\disk(a^+)$, where $a^-$ and $a^+$ are points uniquely
covered by $\disk(p_3)$ on the left and right part of the strip. Note that
$\disk(a^+)$ covers both $\cS_{\geq 0} \cap (\disk(p_1) \setminus \disk(s))$
and $\cS_{\geq 0} \cap (\disk(p_2) \setminus \disk(s))$, as we have seen this
happen for regular swap operations in Lemma~\ref{lem:twopaths} -- similarly,
$\disk(a^-)$ covers both $\cS_{\leq 0} \cap (\disk(p_1) \setminus \disk(s))$
and $\cS_{\leq 0} \cap (\disk(p_2) \setminus \disk(s))$.

Therefore, the new broadcast set obtained after the double swap is feasible,
and the size remains unchanged, so it is a minimum broadcast set. Notice that
a single swap or double swap results in a minimum broadcast set that has an
active point in $L_2$.

If the minimum broadcast set has size three, containing
$\{\disk(s),\disk(p_1),\disk(p_2)\}$, then either both $\disk(p_1)$ and
$\disk(p_2)$ are bidirectional, or at least one of them is monodirectional,
so a single swap operation results in a minimum broadcast set with an active
disk in $L_2$, so there is a path-like minimum broadcast set by
Lemma~\ref{lem:twopaths}. \qed
\end{proof}

As it turns out, the bidirectional case is the most difficult one to compute efficiently. (It is similar to \cdsudg in co-comparability graphs, where the case of a connected dominating set of size at most 3 dominates the running time.)

\begin{lemma}\label{lem:strange2center}
In $O(n\log n)$ time we can find a bidirectional broadcast if it exists.
\end{lemma}

\begin{figure}
\begin{center}
\includegraphics[scale=0.8]{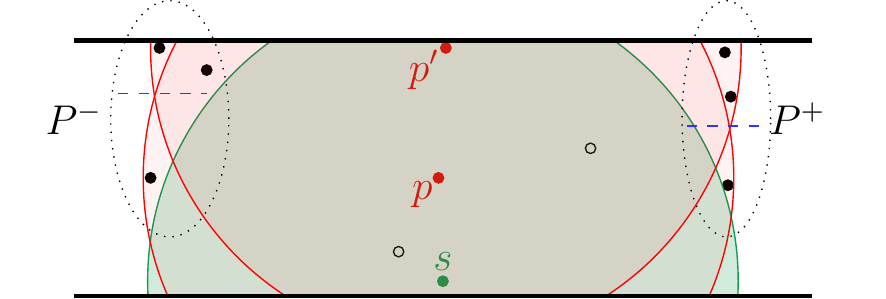}
\end{center}
\caption{A bidirectional broadcast.}\label{fig:strangesolution}
\end{figure}

\begin{proof}
Let $P^- := \{ u_1,u_2,\ldots,u_k\}$ be the set of points to the left of the
source disk~$\disk(s)$, where the points are sorted in increasing $y$-order
with ties broken arbitrarily. Similarly, let $P^+ := \{
v_1,v_2,\ldots,v_l\}$ be the set of points to the right of~$\disk(s)$, again
sorted in order of increasing $y$-coordinate. Define $P^-_{\leq i} := \{
u_1,\ldots,u_i\}$, and define $P^-_{> i}$, and $P^+_{\leq i}$ and $P^+_{> i}$
analogously. Our algorithm is based on the following observation: There is a
bidirectional solution if and only if there are indices $i,j$ and points $p,p'\in
\core(s)$ such that $\disk(p)$ covers $P^-_{\leq i}\cup P^+_{\leq j}$ and
$\disk(p')$ covers $P^-_{> i}\cup P^+_{> j}$; see Fig.~\ref{fig:strangesolution}.

Now for a point $p\in\core(s)$, define
$\Zleftsmall(p) := \max \{ i : P^-_{\leq i} \subset \disk(p) \}$ and
$\Zleftlarge(p) := \min \{ i : P^-_{> i} \subset \disk(p) \}$, and
$\Zrightsmall(p) := \max \{ i : P^+_{\leq i} \subset \disk(p) \}$, and
$\Zrightlarge(p) := \min \{ i : P^+_{> i} \subset \disk(p) \}$.
Then the observation above can be restated as:
\begin{quotation}
\noindent There is a bidirectional solution if
and only if there are points $p,p'\in \core(s)$ such that $\Zleftsmall(p)
\geq \Zleftlarge(p')$ and $\Zrightsmall(p) \geq \Zrightlarge(p')$.
\end{quotation}
It is easy to find such a pair---if it exists---in
$O(n\log n)$ time once we have computed the values $\Zleftsmall(p)$,
$\Zleftlarge(p)$, $\Zrightsmall(p)$, and $\Zrightlarge(p)$ for all points
$p\in\disk(s)$. It remains to show that these values can be computed in
$O(n\log n)$ time.

Consider the computation of~$\Zleftsmall(p)$; the other values can be
computed similarly. Let $\tree$ be a balanced binary tree whose leaves store
the points from $P^-$ in order of their $y$-coordinate. For a node $\nu$ in
$\tree$, let $F(\nu) := \{ \disk(u_i) : \mbox{ $u_i$ is stored in the subtree
rooted at $\nu$} \}$. We start by computing at each node~$\nu$ the
intersection of the disks in $F(\nu)$. More precisely, for each $\nu$ we
compute the region~$I(\nu) := \core(s) \cap \bigcap F(\nu)$. Notice that
$I(\nu)$ is $y$-monotone and convex, and each disk $\disk(u_i)$ contributes
at most one arc to~$\partial I(\nu)$. (Here $\partial I(\nu)$ refers to
the boundary of $I(\nu)$ that falls inside $\cS$.) Moreover,
$I(\nu) = I(\lchild(\nu)) \cap I(\rchild(\nu))$.
Hence, we can compute the regions $I(\nu)$ of all nodes $\nu$ in $\tree$ in
$O(n\log n)$ time in total, in a bottom-up manner. Using the tree~$\tree$ we
can now compute $\Zleftsmall(p)$ for any given~$p\in\core(s)$ by searching
in~$\tree$, as follows. Suppose we arrive at a node~$\nu$.
If $p\in I(\lchild(\nu))$, then descend to $\rchild(\nu)$,
otherwise descend to $\lchild(\nu)$. The search stops when we reach a leaf, storing a point~$u_i$. One easily verifies
that if $p\in \disk(u_i)$ then $\Zleftsmall(p) = i$, otherwise
$\Zleftsmall(p) = i-1$.

Since $I(\nu)$ is a convex region, we can check if $p\in I(\nu)$
in $O(1)$ time if we can locate the position of $p_y$ in the sorted list of
$y$-coordinates of the vertices of $\partial I(\nu)$. We can locate~$p_y$ in
this list in $O(\log n)$ time, leading to an overall query time of $O(\log^2
n)$. This can be improved to $O(\log n)$ using fractional
cascading~\cite{FracCas}. Note that the application of fractional cascading
does not increase the preprocessing time of the data structure. We conclude
that we can compute all values $\Zleftsmall(p)$ in $O(n\log n)$ time in
total. \qed
\end{proof}


In order to compute a minimum broadcast, we can first check for small and bidirectional solutions.  To find path-like solutions, we first compute the sets $Q^-$ and $Q^+$, and compute shortest paths starting from these sets back to the source disk. The path computation is very similar to the shortest path algorithm in UDGs by Cabello and Jej{\v c}i{\v c}~\cite{Cabello15}.

\begin{lemma}
\label{lem:union_intersection}
Let $P$ and $Q$ be two point sets in $\Reals^2$.
Then both $Q \cap \big( \bigcup_{p\in P} \disk(p) \big)$ and
$Q \cap \big( \bigcap_{p\in P} \disk(p) \big)$ can be computed in
$O((|P|+|Q|)\log |P|)$ time.
\end{lemma}
\begin{proof}
A point $q\in Q$ lies in $\bigcup_{p\in P} \disk(p)$ if and only if the
distance from $q$ to its nearest neighbor in~$P$ is at most~1. Hence we can
compute $Q \cap \big( \bigcup_{p\in P} \disk(p) \big)$ by computing the
Voronoi diagram of~$P$, preprocessing it for point location, and performing a
query with each~$q\in Q$. This can be done in $O((|P|+|Q|)\log |P|)$ time in
total~\cite{Berg08,optPL86}. To compute $Q \cap \big( \bigcup_{p\in P}
\disk(p) \big)$ we proceed similarly, except that we use the farthest-point
Voronoi diagram~\cite{Berg08}. \qed
\end{proof}

\begin{lemma}
\label{lem:computeQ}
We can compute the sets $Q^+$ and $Q^-$ in $O(n\log n)$ time.
\end{lemma}
\begin{proof}
We show how we can compute $Q^+$, the algorithm for $Q^-$ is analogous. Let
$p$ be an input point with the highest $x$-coordinate. Notice that all input
points in $\civ{x(p)-\frac12,x(p)}\times \civ{0,w}$ belong to $Q^+$ since their
core contains all points with higher coordinates. Points in
$\roiv{x(p)-\frac32,x(p)-1}\times \civ{0,w}$ cannot belong to $Q^+$, since they
cannot cover $p$.
It remains to find the points inside the region $\cR = \roiv{x(p)-1, x(p) -
\frac12}\times \civ{0,w}$ that belong to $Q^+$. The core of a
point in $\cR$ covers $\cR$, so it is sufficient to check whether any given
point covers all points in $\cR' = \civ{x(p)-\frac12,x(p)} \times \civ{0,w}$.
Thus we need to find the set $(\cR\cap P) \cap \big(\bigcap_{p \in
\cR'\cap P} \disk(p)\big)$, which can be computed in $O(n\log n)$ time by
Lemma~\ref{lem:union_intersection}. \qed
\end{proof}

\begin{restatable}{theorem}{thmnarrowbroadcast}\label{thm:narrowbroadcast}
The broadcast problem inside a strip of width at
most~$\sqrt{3}/2$ can be solved in $O(n\log n)$ time.
\end{restatable}

\begin{proof}
The algorithm can be stated as follows. It is best to read this pseudocode  in parallel with the explanation and analysis below.\\[5mm]
\textsc{Broadcast-In-Narrow-Strip}$(s,P)$ \\[-1.5em]
\begin{enumerate}
\item \label{step:size-check} Check if there is a small or bidirectional solution.
    If yes, report the solution and terminate.
\item \label{step:init+} Compute $Q^+$ using Lemma~\ref{lem:computeQ}. Set $i:=1$, $Q^+_1 := Q^+$, and $P' := P\setminus Q^+_1$.
\item \label{step:compute-levels+} Repeat the following until~$Q^+_i \cap
\disk(s) \neq \emptyset$ or $Q^+_i = \emptyset$.
    \begin{enumerate}
        \item Set $i:=i+1$ and determine $T_i := \{ t\in P': x(t) \geq
              \min_{p\in Q^+_{i -1}} x(p)-1 \}$.
        \item Compute $Q^+_i := T_i \cap \big( \bigcup_{p\in Q^+_{i -1}}
              \disk(p) \big)$ using Lemma~\ref{lem:union_intersection}, and
              set $P' := P'\setminus Q^+_i$.
    \end{enumerate}
\item \label{step:infeasible+} If $Q^+_i = \emptyset$, return failure.
\item \label{step:init-} Compute $Q^-$ using Lemma~\ref{lem:computeQ}. Set $j:=1$, $Q^-_1 := Q^-$, and $P' := P\setminus Q^-_1$.
\item \label{step:compute-levels-} Repeat the following until~$Q^-_j \cap
\disk(s) \neq \emptyset$ or $Q^-_j = \emptyset$.
    \begin{enumerate}
        \item Set $j:=j +1$ and determine $T_j := \{ t\in P': x(t) \leq
        \max_{p\in Q^-_{j -1}} x(p)+1 \}$.
        \item Compute $Q^-_j := T_i \cap \big( \bigcup_{p\in Q^-_{j -1}}
              \disk(p) \big)$ using Lemma~\ref{lem:union_intersection}, and
              set $P' := P'\setminus Q^-_j$.
    \end{enumerate}
\item \label{step:infeasible-} If $Q^-_j = \emptyset$, return failure.
\item \label{step:final} If $Q^+_i \cap Q^-_j = \emptyset$ then
report a solution of size $i+j+1$, namely the points of a shortest path from
$s$ to $Q^+_i$ and a shortest path from $s$ to $Q^-_j$.
Otherwise report a solution of size $i+j$: take an arbitrary point~$p$ in
$Q^+_i \cap Q^-_j$, and report $s$ plus a shortest path from
$p$ to $Q^+_i$ and a shortest path from $p$ to $Q^-_j$.
\end{enumerate}

In order to execute step~\ref{step:size-check}, we first check whether there
is a minimum broadcast set of size one or two. This is very easy for size
one: we just need to check whether the source disk covers every point or not
in $O(n)$ time. For size two, we can compute the intersection of all disks
centered outside $\disk(s)$, and check whether any input point in $\disk(s)$
falls in this intersection. This requires $O(n\log n)$ time by
Lemma~\ref{lem:union_intersection}. Finally, we need to check whether there
is a feasible minimum broadcast with the bidirectional structure.
Lemma~\ref{lem:strange2center} shows that this is also possible in $O(n\log
n)$ time.

In steps~\ref{step:init+} and~\ref{step:compute-levels+}, we compute a
shortest $s \to Q^+$ path backwards. We start from $Q^+$, and put the points
into different sets $Q_i^+$ according to their hop distance to $Q^+$:
we put $p$ into $Q^+_i$ if and only if
the shortest path from $p$ to $Q^+$ contains $i-1$ hops. Notice that in
step~\ref{step:compute-levels+} it is indeed sufficient to consider points
from $T_i$, since a point from the level $Q^+_i$ must be at distance at most
$1$ from points of $Q^+_{i-1}$, so it has $x$- coordinate at least
$\min_{p\in Q^+_{i-1}} x(p)-1$.

If $Q^+_i=\emptyset$, then there is no path
from $Q^+$ to $s$---the graph is disconnected---so there is no feasible
broadcast set. Otherwise, after the loop in step~\ref{step:compute-levels+} terminates the
shortest $s\to Q^+$ path has length exactly equal to the loop variable, $i$.
Moreover, the set of possible second vertices on an $s\to Q^+$ path is
$\disk(s)\cap Q^+_i$. The same can be said for the next two steps: the
shortest $s \to Q^-$ path has length $j$, and the set of possible second
vertices is $\disk(s)\cap Q^+_i$. In the final step, we check if $Q^+_i \cap
Q^-_j$ is empty or not. If it is empty, then by our previous observation,
there are no shortest $s\to Q^+$ and $s\to Q^-$ paths that share their second
vertex, so the two paths can only share $s$, resulting in a minimum broadcast
set of size $i+j+1$; otherwise, any point in $Q^+_i \cap Q^-_j$ is suitable
as a shared second point, resulting in a minimum broadcast set of size $i+j$.

It remains to argue that steps~\ref{step:init+}--\ref{step:final} require
$O(n\log n)$ time. We know that a single iteration of the loop in
step~\ref{step:compute-levels+} takes $O\big((|Q^+_{i-1}|+|T_i|)\log
|Q^+_{i-1}|\big)$ time by Lemma~\ref{lem:union_intersection}. We claim that
$T_i \subseteq Q^+_{i} \cup Q^+_{i+1} \cup Q^+_{i+2}$, from which the bound
on the running time follows. To prove the claim, let $p\in Q^+_{i-1}$ be a
point with minimal $x$-coordinate (see Fig.~\ref{fig:Qlevels}). All points
$p'$ with $x(p') \geq x(p)-\frac12$ are in $Q^+_{\leq i}$. Thus any point
$p''\in Q^+_{i+1}$ has $x(p'') < x(p)-\frac12$. But then any point with
$x$-coordinate at least $x(p)-1$ also has $x$-coordinate at least $x(p'') -
\frac12$, which means it is in $Q^+_{\leq i+2}$. Thus both loops require
$O(n\log n)$ time. Finally, we note that we can easily maintain some extra
information in steps~\ref{step:init+}--\ref{step:infeasible-} so the shortest
paths we need in step~\ref{step:final} can be reported in linear time.
\qed \end{proof}

\begin{figure}
\begin{center}
\begin{tikzpicture}[x=1.3cm,y=1.3cm]
\clip (1.4,-0.65) rectangle (5,2.15);
\draw [very thick] (0,0) -- (8,0);
\draw [very thick] (0,1.73) -- (8,1.73);
\begin{large}
\tikzstyle{p3n}=[draw,circle,fill=ForestGreen,minimum size=5pt, inner
sep=0pt]
\tikzstyle{p2n}=[draw,circle,fill=red,minimum size=5pt, inner sep=0pt]
\tikzstyle{p1n}=[draw,circle,fill=blue,minimum size=5pt, inner sep=0pt]
\tikzstyle{en}=[draw,circle,fill=white,minimum size=5pt, inner sep=0pt]

\node [en,label=left:$p$] (p) at (4,0.8) {};
\draw [dashed] (4,0) -- (4,1.73);
\draw [->,very thick,dotted] (2,1.73) -- (2,-0.5) -- (2.7,-0.5);
\node at (2.2,-0.3) {$T_i$};

\node [en] at (4,0.8) {};
\node [en] at (4.5,1.3) {};
\node [en] at (4.7,0.2) {};
\node [en] at (4.2,1.6) {};
\node [p1n] at (3.8,1.1) {};
\node [p1n] at (3.7,0.4) {};
\node [p1n] at (3.92,1.5) {};
\node [p1n] at (4.1,0.1) {};
\node [p2n] at (3,0.13) {};
\node [p2n] at (2.5,0.3) {};
\node [p2n] at (2.2,0.7) {};
\node [p2n] at (2.4,0.2) {};
\node [p3n] at (1.8,1.3) {};
\node [p3n] at (2.1,1.65) {};
\node [p3n] at (1.6,0.6) {};

\node [black] at (4.7,2) {$Q^+_{i-1}$};
\node [blue] at (3.7,2) {$Q^+_{i}$};
\node [red] at (2.7,2) {$Q^+_{i+1}$};
\node [ForestGreen] at (1.7,2) {$Q^+_{i+2}$};
\end{large}
\end{tikzpicture}
\end{center}
\caption{The levels $Q^+_i$ computed by the algorithm.}
\label{fig:Qlevels}
\end{figure}
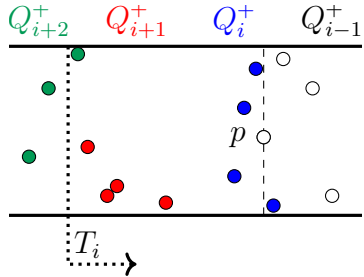

\begin{remark}
If we apply this algorithm to every disk as source, we get an $O(n^2\log n)$ algorithm for \cds in narrow strip UDGs. We can compare this to $O(mn)$, the running time that we get by applying the algorithm for co-comparability graphs~\cite{Breu96}. Note that in the most difficult case, when the size of the minimum connected dominating set is at most $3$, the unit disk graph has constant diameter, which implies that the graph is dense, i.e., the number of edges is $m=\Omega(n^2)$. Hence, we get an (almost) linear speedup for the worst-case running time.
\end{remark}

\section{Minimum-size \texorpdfstring{$h$}{h}-hop broadcast in a narrow strip}
\label{se:hop-bounded}
In the hop-bounded version of the problem we are given $P$ and a
parameter~$h$, and we want to compute a broadcast set~$D$ such that every
point~$p\in P$ can be reached in at most $h$~hops from~$s$. In other words,
for any $p\in P$, there must be a path in $\graph$ from $s$ to $p$ of
length at most~$h$, all of whose vertices, except possibly~$p$ itself, are
in~$D$. We start by investigating the structure of optimal solutions in this
setting, which can be very different from the non-hop-bounded setting.

As before, we partition $P$ into levels $L_i$ according to the hop distance
from $s$ in the graph~$\graph$, and we define $L_i^+$ and $L_i^-$ to be the subsets of
points at level~$i$ with positive and nonnegative $x$-coordinates,
respectively. Let $L_t$ be the highest non-empty level. If $t>h$ then clearly
there is no feasible solution.

If $t<h$ then we can safely use our solution for the non-hop-bounded case,
because the non-hop-bounded algorithm gives a solution which contains a path
with at most $t+1$ hops to any point in~$P$. This follows from the structure
of the solution; see Lemma~\ref{lem:narrowstructure}. (Note that it is
possible that the solution given by this algorithm requires $t+1$ hops to
some point, namely, if $Q^+\cup Q^- \subseteq L_t$.) With the $t<h$ case
handled by the non-hop-bounded algorithm, we are only concerned with the
case~$t=h$.

We deal with \emph{one-sided} inputs first, where the source is the leftmost
input point. Let $\graph^{*}$ be the directed graph obtained by deleting
edges connecting points inside the same level of $\graph$, and orienting all
remaining edges from lower to higher levels. A \emph{Steiner arborescence} of
$\graph^{*}$ for the terminal set~$L_h$ is a directed tree rooted at~$s$ that
contains a (directed) path $\mypath_p$ from $s$ to $p$ for each $p\in L_h$.
From now on, whenever we speak of \emph{arborescence} we refer to a Steiner
arborescence in~$\graph^*$ for terminal set~$L_h$. We define the \emph{size}
of an arborescence to be the number of internal nodes of the arborescence.
Note that the leaves in a minimum-size arborescence are exactly the points in
$L_h$: these points must be in the arborescence by definition, they must be
leaves since they have out-degree~zero in~$\graph^*$, and leaves that are not
in $L_h$ can be removed.

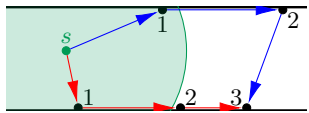
\begin{figure}
\begin{center}
\begin{tikzpicture}[x=0.8cm,y=0.8cm]
\clip (1,0) rectangle (6,1.73);
\draw [very thick] (0,0) -- (8,0);
\draw [very thick] (0,1.73) -- (8,1.73);
\tikzstyle{every node}=[draw,circle,fill=black,minimum size=3pt,
inner sep=0pt]
\node [ForestGreen,label={[ForestGreen]above:{$s$}}] (s) at (2,1) {};
\node [label=below:$1$](r1) at (3.6,1.68) {};
\node [label=above right:$1$](r1p) at (2.2,0.05) {};
\node [label=above right:$2$] (r2p) at (3.9,0.05) {};
\node [label=below right:$2$] (r2) at (5.6,1.68) {};
\node [label=above left:$3$] (r3) at (5,0.05) {};
\filldraw [ForestGreen, fill opacity=0.2] (s.center) circle (2);
\draw [blue,-{Latex[length=3mm, width=1.5mm]}] (s) -- (r1);
\draw [blue,-{Latex[length=3mm, width=1.5mm]}] (r1) -- (r2);
\draw [blue,-{Latex[length=3mm, width=1.5mm]}] (r2) -- (r3);
\draw [red,-{Latex[length=3mm, width=1.5mm]}] (s) -- (r1p);
\draw [red,-{Latex[length=3mm, width=1.5mm]}] (r1p) -- (r2p);
\draw [red,-{Latex[length=3mm, width=1.5mm]}] (r2p) -- (r3);
\end{tikzpicture}
\end{center}
\caption{Two different arborescences, with vertices labeled with their level.
The red arborescence does not define a feasible broadcast for $h=3$, since it would take four hops to reach the top right node.}
\label{fig:weirdarborescence}
\end{figure}

\skb{
\begin{remark}
In the minimum Steiner Set problem, we are given a graph $G$ and a vertex subset $T$ of terminals, and the goal is to find a minimum-size vertex subset $S$ such that $T\cup S$ induces a connected subgraph. This problem has a polynomial algorithm in co-comparability graphs~\cite{Breu96}, and therefore in narrow strip unit disk graphs. However, the broadcast set given by a solution does not fit our hop bound requirements. Hence, we have to work with a different graph (e.g.  the edges within each level $L_i$ have been removed), and this modified graph is not necessarily a co-comparability graph.
\end{remark}
}

Lemma~\ref{lem:arborescenceswork} below states that either we have a 
path-like solution---for the one-sided case a path-like solution is
a shortest $s \to Q^+$ path--- or any minimum-size arborescence
defines a minimum-size broadcast set. The latter solution
is obtained by activating all non-leaf nodes of the arborescence. We
denote the broadcast set obtained from an arborescence~$A$ by~$D_A$.

\begin{restatable}{lemma}{lemarborescenceswork}
\label{lem:arborescenceswork}
Any minimum-size Steiner arborescence for the terminal set $L_h$
defines a minimum broadcast set, or there is a path-like minimum broadcast
set.
\end{restatable}

\begin{proof}
Let $A$ be a minimum Steiner arborescence for the terminal set $L_h$. Suppose
that the broadcast set~$D_A$ defined by the internal vertices
of~$A$ is not an $h$-hop broadcast set. (If it is, it must also be minimum and we are done.)
By the properties
of the arborescence every point in $D_A$ can be reached
in at most $h-1$ hops. Hence, if there is a point~$p\in P$ that cannot be reached
within~$h$ hops via~$D_A$ then $p$ cannot be reached at all via~$D_A$.
Let $i$ be such that $p\in L_i$. Since $L_h\subset A$, we know that $i<h$.
Take any path from $s$ to any point in $L_{h-1}$ inside the arborescence.
By Observation~\ref{obs:corecover}(iii), this path covers all lower levels.
Hence, $i\ge h-2$, which implies $p \in L_{h-1}$.

Without loss of generality, suppose that $p$ has the highest $x$-coordinate
among points not covered by $A$. Let $q$ be the point in $P$ with the
largest $x$-coordinate. If $q\in L_{\leq h-1}$, then a shortest $s\to q$ path
is a feasible broadcast set of size at most $|A|$ that is path-like.
Therefore, we only need to deal with the case $q\in L_h$. Let $p'\in A$ be an
internal vertex of the arborescence whose disk covers $q$. The arborescence
contains an $s\to p'$ path, which, by
Observation~\ref{obs:corecover}(i), covers everything with $x$-coordinate up to
$x(p')+\frac12$. Since $p \not\in \disk(p')$, we have $x(p)>x(p')+\frac12 \geq
x(q)-\frac12$. Since $q$ has the maximum $x$ coordinate,
Observation~\ref{obs:corecover}(i) shows that the disks of a shortest $s\to p$
path form a feasible broadcast set, which is a path-like solution.
\qed \end{proof}

Notice that a path-like solution also corresponds to an arborescence.
However, it can happen that there are minimum-size arborescences that do not
define a feasible broadcast; see Fig.~\ref{fig:weirdarborescence}.
Lemma~\ref{lem:arborescenceswork} implies that if this happens, then there
must be an optimal path-like solution. The lemma also implies that for
non-path-like solutions we can use the Dreyfus-Wagner dynamic-programming
algorithm to compute a minimum Steiner tree~\cite{Dreyfus71}, and obtain an
optimal solution from this tree.\footnote{The Dreyfus-Wagner algorithm
minimizes the number of edges in the arborescence. In our setting the number
of edges equals the number of internal nodes plus $|L_h|-1$, so this also
minimizes the number of internal nodes.} Unfortunately the running time is
exponential in the number of terminals, which is $|L_h|$ in our case.
However, our setup has some special properties that we can use to get a
polynomial algorithm.

We define an arborescence~$A$ to be \emph{nice} if the following holds. For
any two arcs $uu'$ and $vv'$ of~$A$ that go between the same two levels, with
$u\neq v$, we have: $y(u')<y(v') \Rightarrow y(u) < y(v)$. Intuitively, a
nice arborescence is one consisting of paths that can be ordered vertically
in a consistent manner, see the left of Fig.~\ref{fig:nicearborpred}. We define an
arborescence~$A$ to be \emph{compatible} with a broadcast set~$D$ if $D=D_A$.
Note that there can be multiple arborescences---that is, arborescences with
the same node set but different edge sets---compatible with a given broadcast
set~$D$.

\begin{observation}\label{obs:levelwidth}
In a minimum broadcast set on the strip, the difference in $x$-coordinates
between active points from a given level $L_i$ ($i\leq h-1$) is at most
$\frac{1}{2}$.
\end{observation}
\begin{proof}
Let $p$ and $q$ be active points from $L_i$, and suppose for contradiction
that $x(p)>x(q)+\frac{1}{2}$. By Observation~\ref{obs:corecover}(i),
all points to the left of $p$ are covered by the active points, so we only
need to show that there are no points in $L_{i+1}$ whose hop distance becomes
longer by removing $\disk(q)$ from the solution. Indeed, consider a point
$v\in L_{i+1} \cap (\disk(q) \setminus \disk(p))$. Since $\disk(q)\setminus
\disk(p)$ lies to the left of $p$, $x(v)<x(p)$. So $v$ has a path of at most
$i+1$ hops. Hence we still have a feasible solution after removing
$\disk(q)$, which contradicts the optimality of the original solution.
\qed \end{proof}

\begin{lemma}
\label{lem:length_i_path}
Let $p\in L_i$ be a point in an optimal broadcast set $D$. Then there
is a path of length $i$ from $s$ to $p$ in $\graph[D]$, the graph induced by~$D$.
\end{lemma}
\begin{proof}
We say that a vertex $p\in L_i \cap D$ is \emph{bad} if the shortest path
in $\graph[D]$ has more than $i$ hops. Let $p$ be a bad vertex of highest
level among the bad vertices. If $i=h$, then the broadcast set is infeasible, thus $i \leq h-1$. If
$p\in L_{h-1}$, then the shortest $s \to p$ path in $\graph[D]$ must have length $h$,
consequently, $p$ cannot be used in an $h$-hop path to any other point.
Therefore, $p$ can be deactivated. (Note that $p$ itself remains covered
since it was reachable in the first place.)

If $p$ is on a lower level, then let $\pi_q$ be a shortest path in $\graph[D]$ going to the
last level, and let $q \in \pi_q \cap L_{h-1}$. Let $\pi_p$ be the shortest
$s\to p$ path in $\graph[D]$. Note that $\pi_q$ covers all lower levels $L_{\leq
h-2}$ using at most $h$ hops. Since $i$ is the highest level with a bad
point, all points $v \in D \cap L_{\geq i+1}$ have a shortest path in $\graph[D]$,
and such a path cannot pass through $p$.

Since $p$ is a necessary point in this broadcast, and it is already
covered by the disks of $\pi_q$ in at most $h$ hops, there must be a point
$p'$ to which all covering paths of length at most $h$ pass through $p$.
Since all points of $L_h$ are covered by $D\cap L_{h-1}$ and $L_{\leq h-2}$
is covered by $\pi_q$, the level of $p'$ has to be $h-1$. A covering path to
$p'$ has only bad vertices after $p$, so its point in $L_{h-2}$ is bad. By
the choice of $p$, we have $p\in L_{h-2}$, and since $p'$ is reached in
exactly $h$ hops, it also follows that $p' \in \disk(p)$.

Note that $p'$ cannot be to the left of $\disk(q)$, since
then $\pi_q$ would cover it in at most $h$ hops; therefore, $x(p') >
x(q)+\frac12$. It follows that $x(p) \geq x(q)-\frac12$, so $\disk(p)$ covers
$q$. Since $q$ is an arbitrary point in $D\cap L_{h-1}$, we have  $D\cap
L_{h-1} \subseteq \disk(p)$. Let $D'$ be the broadcast obtained by replacing $D
\cap L_{\leq h-2}$ with a shortest $s \to p$ path $\pi'_p$. We claim that $D'
= \pi'_p \cup (D\cap L_{h-1})$ is a feasible broadcast: it covers $L_h$ since
points of $L_h$ could only be covered by $D\cap L_{h-1}$, and it is easy to
check that all points are covered in at most $h$ hops. We arrived at a
contradiction since $D'$ is smaller than $\pi_p \cup (D\cap L_{h-1})
\subseteq D$.
\qed \end{proof}

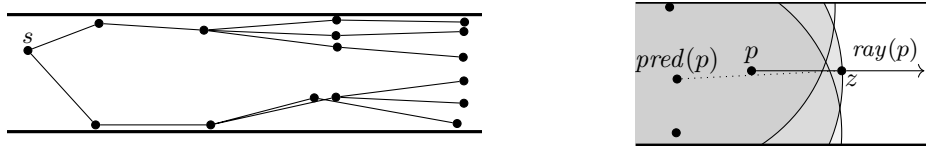
\begin{figure}
\begin{center}
\begin{tikzpicture}[x=0.9cm,y=0.9 cm,
	strip/.style={black, line width=1pt}
]
\clip (3,-0.2) rectangle (10,1.93);
\draw [very thick] (0,0) -- (14,0);
\draw [very thick] (0,1.73) -- (14,1.73);
\tikzstyle{every node}=[draw,circle,fill=black,minimum size=3pt,
inner sep=0pt]
\node [label={above:{$s$}}] (r1) at (3.3,1.2) {};
\node (r2) at (4.3,0.1) {};
\node (r3) at (4.35,1.6) {};
\node (r4) at (6,0.1) {};
\node (r5) at (5.9,1.5) {};
\node (r6) at (7.86,1.65) {};
\node (r7) at (7.85,1.42) {};
\node (r8) at (7.87,1.25) {};
\node (r9) at (7.53,0.5) {};
\node (r10) at (7.85,0.51) {};
\node (r11) at (9.75,1.62) {};
\node (r12) at (9.74,1.48) {};
\node (r13) at (9.73,1.1) {};
\node (r15) at (9.74,0.75) {};
\node (r16) at (9.74,0.42) {};
\node (r17) at (9.64,0.12) {};
\draw (r1) --(r2);
\draw (r2) --(r4);
\draw (r1) -- (r3);
\draw (r3) -- (r5) -- (r8) -- (r13);
\draw (r17) -- (r9) -- (r4) -- (r10);
\draw (r11) -- (r6) -- (r5) -- (r7) -- (r12);
\draw (r16) -- (r10) -- (r15);
\end{tikzpicture}
\hfill
\begin{tikzpicture}[x=1.1 cm,y=1.1 cm]
\clip (0,0) rectangle (3.5,1.73);
\draw [very thick] (0,0) -- (8,0);
\draw [very thick] (0,1.73) -- (8,1.73);
\begin{normalsize}
\tikzstyle{every node}=[draw,circle,fill=black,minimum size=3pt,
inner sep=0pt]
\node [label={[rectangle] $\mypred(p)$}] (lamq) at (0.5,0.8) {};
\node [label=above:$p$] (q) at (1.4,0.9) {};
\node [label=below right:$z$] (qpr) at (2.49,0.9) {};
\node (up) at (0.41,1.67) {};
\node (down) at (0.49,0.15) {};
\end{normalsize}
\draw [dotted] (lamq) -- (qpr);
\draw [->] (q) -- (qpr) -- node [above] {$\myray(p)$}  ++(1,0);
\filldraw [black, fill opacity=0.1] (lamq) circle (2);
\filldraw [black, fill opacity=0.05] (up) circle (2);
\filldraw [black, fill opacity=0.05] (down) circle (2);
\end{tikzpicture}
\end{center}
\caption{Left: A nice Steiner arborescence. Note that arc crossings are possible. Right: Defining the $\mypred$ function.}\label{fig:nicearborpred}
\end{figure}

\begin{restatable}{lemma}{lemnicearbor}
\label{lem:nicearbor}
Every optimal broadcast set~$D$ has a nice compatible arborescence.
\end{restatable}
\begin{proofsketch}
To find a nice compatible arborescence we will associate a unique
arborescence with~$D$. To this end,
we define for each $p\in (D \cup L_h) \setminus \{s\}$ a unique predecessor
$\mypred(p)$, as follows. Let $\partial^*_i$ be the boundary of $\bigcup
\left\{\disk(p) | p\in L_{i} \cap D \right\}$. It follows from
Observation~\ref{obs:levelwidth} that the two lines bounding
the strip~$\cS$ cut $\partial^*_i$ into four parts: a top and a bottom part
that lie outside the strip, and a left and a right part that lie inside the
strip. Let $\partial_i$ be the part on the right inside the strip.
We then define the function $\mypred: (D \cup L_h) \setminus \{s\}
\rightarrow D$ the following way. Consider a point $p \in (D \cup L_h)
\setminus \{s\}$ and let $i$ be its level. Let $\myray(q)$ be the horizontal
ray emanating from $q$ to the right; see the right of Fig.~\ref{fig:nicearborpred}. It follows
from Observation~\ref{obs:corecover}(iii) that $\myray(q)$ cannot enter any
disk from level~$i-1$. We can prove that any point $p\in D\cap L_h$ is
contained in a disk from $p$'s previous level, so $\mypred(p)$ is well
defined for these points. The edges $\mypred(p) p$ for $p\in D\cap L_h$ thus
define an arborescence. We can prove that it is nice by showing that the
$y$-order of the points in a level $L_i$ corresponds to the vertical order in
which the boundaries of their disks appear on $\bigcup \{ \disk(p) : p\in
L_i\cap D
\}$.
\end{proofsketch}
\begin{proof}
Recall that $\mypred(p)$, for $p\in L_i\cap D$, is the center of the level $i-1$ disk which has $z =
\myray(p) \cap \partial_{i-1}$ on its boundary. If there are multiple such
disks, we can break ties by choosing $\mypred(p)$ to be the point with the
highest $y$-coordinate in $L_i\cap D$ whose disk passes through $z$.

Let $A$ be the directed graph defined by the edges $\mypred(p)p$ for each $p\in
(D \cup L_h) \setminus \{s\}$. We show that $A$ is a nice arborescence.
By definition of the $\mypred$-function, each edge is between points
at distance at most~1 that are in subsequent levels. Hence, the edges we add
define an arborescence~$A$ on $\graph^{*}$ with terminal set~$L_h$. It
remains to prove that $A$ is nice.

Consider the edges of $A$ going between points in~$L_{i-1}$ and points
in~$L_i$. By drawing horizontal lines through each of the breakpoints of
$\partial_{i-1}$, the strip~$\cS$ is partitioned into horizontal sub-strips,
such that two points from $L_i$ are assigned the same predecessor iff they
lie in the same sub-strip. Number the sub-strips $\cS_1,\cS_2,\ldots$ in
vertical order, with $\cS_1$ being the bottommost sub-strip. Let $u_j\in D
\cap L_{i-1}$ be the point that is the predecessor of the points in the sub-
strip~$\cS_j$. To show that $A$ is nice, it is sufficient to demonstrate that
the sequence $u_1,u_2,\ldots$  is ordered by the $y$-coordinates of the
points.

Suppose for a contradiction that this is not the case. Then there are points
$u_j$ and $u_{j+1}$ such that $y(u_j) > y(u_{j+1})$. Let $z$ be the
breakpoint on $\partial_{i-1}$ between the arcs defined by $\disk(u_j)$ and
$\disk(u_{j+1})$. Since $z$ is in the right half circle of both
$\disk(u_j)$ and $\disk(u_{j+1})$, we have
$\max\{x(u_j),x(u_{j+1})\}<x(z)$. Since $|u_j z| = |u_{j+1} z| = 1$, the
point~$z$ lies on the perpendicular bisector of $u_j u_{j+1}$ to the right of
$u_j$ and $u_{j+1}$. Since $y(u_j) > y(u_{j+1})$, the outer circle below the
bisector is $\disk(u_{j+1})$ and the outer circle above the bisector is
$\disk(u_j)$. This contradicts the ordering of the sub-strips.
\qed \end{proof}

Let $q_1,q_2,\dots,q_m$ be the points of $L_h$ in increasing $y$-order. The
crucial property of a nice arborescence is that the descendant leaves of a
point $p$ in the arborescence form an interval of $q_1,q_2,\dots,q_m$. Using the above lemmas, we can adapt the Dreyfus-Wagner algorithm and get the following theorem.
\begin{restatable}{theorem}{thmonesidedhbr}
\label{thm:onesided_hbr}
The one-sided $h$-hop broadcast problem inside a strip of width at
most~$\sqrt{3}/2$ can be solved in $O(n^4)$ time.
\end{restatable}
\begin{proof}
By our lemmas, we know that our solution can be categorized as path-like or
as arborescence-based. We compute the best path-like solution by
invoking the second part of our narrow strip broadcast algorithm, which runs
in $O(n\log n)$ time. The output of this algorithm is a path with $t$ or
$t+1$ hops (where $t$ is the number of levels); thus, it is a minimum $h$-hop
broadcast set if $t<h$, or if $t=h$ and the path has length $h$. Otherwise
there is no path-like $h$-hop broadcast set, so an arborescence
defines a minimum $h$-hop broadcast set by Lemma~\ref{lem:arborescenceswork}.
By Lemma~\ref{lem:nicearbor}, it is sufficient to look for a nice Steiner
arborescence, and take the broadcast set defined by it.

The algorithm to find a nice Steiner arborescence is based on dynamic
programming. A subproblem is defined by a point $p\in P$ and an interval of
the last level (that is, an interval of the sequence $q_1, q_2, \dots, q_m$,
the points of $L_h$ ordered by $y$-coordinates). The solution of the
subproblem $M(p, \civ{i,j})$, for $1\leq i \leq j \leq m$, is the minimum
number of internal vertices in a nice arborescence which is rooted at $p$ and
contains $q_i,q_{i+1},\dots,q_j$ as leaves. Recall that $d_{\graph^{*}}(p,q)$
denotes the hop distance function in $\graph^{*}$, where
$d_{\graph^{*}}(p,q)=\infty$ if there is no path from $p$ to $q$. We claim
that the following recursion holds:

\begin{align}\label{eq:strip_h_1side}
M&(p,\civ{i,j}) =\nonumber \\
&
\begin{dcases}
d_{\graph^{*}}(p,q_i) -1\\[-0.5em]
 \hspace{9cm} \text{ if $i=j$,}\\[1em]
\, \min \Big( \! \min_{i \leq t \leq j-1} \!\!
                    \big(M(p,\civ{i,t}) + M(p,\civ{t+1,j})  \big),\; 1 +
                    \hspace{-3mm} \min_{\substack{p'\in P \cap \disk(p) \\
                    p' \neq p}} \hspace{-3mm} M(p',\civ{i,j}) \Big)\\[-0.5em]
\hspace{9cm} \text{ if $i<j$.}
\end{dcases}
\end{align}

The number of subproblems is $O(n^3)$, each of them requires computing the
minimum of at most $O(n)$ values. This results in an algorithm that runs in
$O(n^4)$ time. The minimum broadcast set size is $M(s,\civ{1,m})$; if we
keep track of a representing arborescence for each subproblem, we can also
return a minimum broadcast set without any extra runtime cost.

To prove correctness, we need to show that Equation~\eqref{eq:strip_h_1side}
is correct. The base case, $i=j$, is obviously correct, so now assume~$i<j$.
It is easily checked that $M(p,\civ{i,i})$ is at most the right-hand side of
the equation.
For the reverse direction, consider a nice optimal Steiner arborescence $A$
for $M\left(p,\civ{i,j}\right)$. If $p$ has exactly one outgoing arc in~$A$,
that arc must end in a point $p'\in P \cap \disk(p) \setminus \{p\}$. Then
$A\setminus \{p\}$ is an arborescence rooted at $p'$ that spans $\civ{i,j}$,
so it has at least $M(p',\civ{i,j})$ internal vertices. If $p$ has at least
two outgoing internal vertices, then let $p'$ be the child of $p$ with the
lowest $y$-coordinate. Since the arborescence is nice, the descendant leaves
of $p'$ in $A$ form a sub- interval of $\civ{i,j}$ that starts at~$i$. Let
$q_t$ be the leaf with the highest $y$-coordinate among the descendants of
$p'$. If $A$ had strictly less internal vertices than $M(p,\civ{i,t}) + M(p,
\civ{t+1,j})$, then it would need to include a nice sub-arborescence with
less internal vertices for at least one of the subproblems $M(p,\civ{i,t})$
or $M(p,\civ{t+1,j})$, but that would contradict the optimality in the
definition of the subproblems.
\qed \end{proof}

In the general (two-sided) case, we can have path-like solutions and 
arborescence-based solutions on both sides, and the two side solutions may or may not share points in $L_1$. We also need to handle
``small'' solutions---now these are 2-hop solutions---separately.
\begin{restatable}{theorem}{thmtwosidedhbr}
\label{thm:twosided_hbr}
The $h$-hop broadcast problem inside a strip of width at
most~$\sqrt{3}/2$ can be solved in $O(n^6)$ time.
\end{restatable}
\begin{proof}
We first analyze the possible structures of an optimal solution.
\begin{quotation}
\noindent \emph{Claim.}
For any input $P$ inside small strip that has a feasible $h$-hop broadcast set,
there is a minimum $h$-hop broadcast set $D$ that has one of the following structures:
\begin{itemize}
\item \emph{2-hop}: A solution $D$ that does not not contain any active points from $L_2$.
   (Note that such a solution might be optimal even if $h>2$.)
\item \emph{Path-like}: A solution $D$ that consists of two shortest paths,
    one from $s$ to $Q^+$ and one from $s$ to $Q^-$, possibly sharing their first vertex
    after~$s$.
\item \emph{Mixed}: A shortest path on one side, and a nice arborescence
    on the other side, where the shortest path may share its $L_1$-vertex with
    the arborescence.
\item \emph{Arborescence-based}: A single arborescence for $L_h$, which is nice on both sides.
\end{itemize}
\emph{Proof of claim.}
Suppose that there is no optimal 2-hop solution for $P$. Thus any optimal
solution has active points on $L_{\ge 2}$.
Let $\pi^+$ and $\pi^-$ be shortest paths to $Q^+$ and $Q^-$, respectively.
If both $\mypath^+$ and $\mypath^-$ have at most $h-1$ edges then
everything can be reached in $h$ hops. Hence, this is an optimal path-like
solution (since it is minimal even for the non-hop-bounded version).

If $\mypath^+$ has $h+1$ hops and $\mypath^-$ has at most $h$ hops, then
there is no path-like $h$-hop broadcast for the right side of the input, that is,
for the set $P^* := \{P \cap (\disk(s) \cup \cS_{\geq 0})\}$. 
Let $T$ be a minimum-size nice arborescence for~$P^*$. By Lemma~\ref{lem:arborescenceswork} 
and Lemma~\ref{lem:nicearbor}, $T$ gives a minimum $h$-broadcast
set for~$P^*$. Either there is a shortest $s\to Q^-$ path
whose $L_1$-vertex is also in a minimum-size arborescence, or there isn't.
In both cases, the resulting mixed solution must be optimal.
Thus, if exactly one of $\mypath^+$ and
$\mypath^-$ has $h+1$ hops and the other has fewer hops, then there is a mixed optimal solution.

Now suppose both paths have $h+1$ hops. We now now consider an optimal solution~$D$
and extend the definition of the $\mypred$ function (as described below) to conclude
that $D$ defines a nice arborescence. Let $\mypred^+$ be the previously defined function in
$L_{\leq 1} \cup L^+_i$, and let $\mypred^-$ be the same function for the left
side $L_{\leq 1} \cup L^-_i$. Note that points in $L_1$ belong to both sides,
but for a point $p\in L_1$ we have $\mypred^-(p)=\mypred^+(p)=s$, so this is
not an issue. The arborescence defined by this function is nice on both sides
by Lemma~\ref{lem:nicearbor}. In addition, since there is no path-like
$h$-hop broadcast set on either side, the active points corresponding to this
arborescence form a minimum $h$-hop broadcast set: by applying
Lemma~\ref{lem:arborescenceswork} on both sides, we see that the broadcast
set corresponding to this arborescence covers all points.
\hfill {\footnotesize $\Box$}
\end{quotation}
%
%
The best 2-hop solution can be found using our planar 2-hop broadcast algorithm from Theorem~\ref{thm:n4}. The best path-like solution can be found by invoking the
narrow-strip broadcast algorithm from Theorem~\ref{thm:narrowbroadcast},
and checking if it satisfies the hop-bound. It remains to describe how to find
the best mixed and arborescence-based solutions.
\begin{quotation}
\noindent \emph{Claim.} The best mixed solution can be found in $O(n^5)$ time.
\\[2mm]
\emph{Proof of claim.}
Suppose that $Q^-$ can be reached in $t\leq h$ hops. Recall from the one-sided case 
that $Q^-_i$ is the set of points~$p$ such that the shortest path from $p$ to $Q^-$ 
has $i-1$ hops. Thus the set $B^-$ of potential second points of a shortest
$s\to Q^-$ path is equal to $B^- := \disk(s) \cap Q^-_t$.
(This set can be computed using our algorithm from Theorem~\ref{thm:narrowbroadcast}.)
We need to be able to find the potential
second points of a nice arborescence. First, we run the one-sided
dynamic programming algorithm on the set $P^* := \{P \cap (\disk(s) \cup \cS_{\geq 0})\}$, 
which takes $O(n^4)$ time. Let $M(\cdot,[\cdot,\cdot])$ be the resulting
dynamic-programming table. We claim that $p\in L_1$ is a potential second point
if and only if there is an interval $\civ{i,j}$  such that
\begin{equation}\label{eq:claim}
M(s,\civ{1,m^+}) = M(s,\civ{1,i-1}) + M(p,\civ{i,j}) + M(s, \civ{j+1,m^+}) -1,
\end{equation}
where $m^+=|L^+_h|$. 

To prove the claim, first assume that Equality~\eqref{eq:claim} holds.
Then the arborescences corresponding to each $M$-value on the right side are nice
minimum arborescences rooted at $s$, $p$ and $s$ respectively---the fact
that $s$ is counted twice explains the -1 term---and so
their union together with the edge $sp$ is a minimum arborescence
that uses $p$ as as second point. 
On the other hand, if there is a minimum arborescence
using $p$, then there is a nice one and the set of ancestors of $p$ is an
subsequence $q_i,\ldots,q_j$ of~$L^+_h$. The points $q_1,\ldots,q_{i-1}$ 
and $q_{j+1},\ldots,q_{m^+}$ are covered by two nice arborescences rooted at~$s$, 
and the niceness implies that these subtrees only share $s$. 
Thus, Equality~\eqref{eq:claim} holds.

Hence, after filling in all entries in the table $M(\cdot,[\cdot,\cdot])$, we
can find all potential second points in $O(n^3)$ time by
checking all values $i,j$ for each point~$p\in L_1$. If there is such a point $p$ in
$B^-$, then the best mixed solution has size $A(s,\civ{1,m^+}) + t -1$, otherwise it has size
$A(s,\civ{1,m^+}) + t$. With standard techniques, an $h$-hop broadcast set
realizing this optimum can be computed within the same time bound.
\hfill {\footnotesize $\Box$}
\end{quotation}
%
%
\begin{quotation}
\noindent \emph{Claim.} The best arborescence-based solution can be found in $O(n^6)$ time.
\\[2mm]
\emph{Proof of claim.}
In order to find the best arborescence-based solution, we modify the one-sided
algorithm the following way. For all $p \in L_1 \cup \big( \bigcup_{i=2}^h
L^+_i\big)$ we define the subproblems $A^+ \left(p, \civ{i,j}\right)$ as
previously, where $\civ{i,j}$ refers to an interval in the last right side
level $L^+_h$. Similarly, we define an ordering on the last left level based
on $y$-coordinates, and define for all $p \in L_1 \cup \big( \bigcup_{i=2}^h
L^-_i\big)$ the subproblems $A^-\left(p, \civ{i,j}\right)$. We can compute
these values using the one-sided algorithm on both sides.

It will be convenient to generalize the definitions above as follows.
First of all, we extend the definition of $A^+ \left(p, \civ{i,j}\right)$
to include all points $p\in P$---not only the points 
in $L_1 \cup \big( \bigcup_{i=2}^h L^+_i)$---by setting  $A^+\left(p,\civ{i,j}\right):=\infty$
for $p\in L^-_{\geq 2}$. The definition of $A^- \left(p, \civ{i,j}\right)$
is extended similarly. Finally, we define 
$A^+\left(p,\civ{i,j}\right):= 0$ and $A^-\left(p,\civ{i,j}\right):=0$ for $j=i-1$.

We also need a third kind of subproblem. Define $A(p,\civ{i,j},\civ{k,\ell})$
as the number of internal vertices in an optimum arborescence rooted at
$p$ that has leaves $q^-_i,\ldots,q^-_j$ in the last left level and from
$q^-_k,\ldots,q^-_{\ell}$ on the last right level. If $p \neq s$, this can be easily
expressed:
\begin{equation}\label{eq:ADoubleGen}
A\big(p,\civ{i,j},\civ{k,\ell}\big)=A^-\big(p,\civ{i,j}\big) +
A^+\big(p,\civ{k,\ell}\big) -1.
\end{equation}
Note that on the right side of this formula, at least one of the summands is
$\infty$ if $p\in L_{\geq 2}$, and possibly for some points in $L_1$ as
well. Since the formula is so simple, we do not need to compute these values
explicitly. The only computation for this kind of subproblem is required at
the source, for which we require a new notation. Let
\begin{align*} 
\sep(&i,j,k,\ell) := \\ &\big\{(t,u)\; : \;i-1\leq t \leq j \mbox{ and } k-1\leq u \leq
\ell\big\}\setminus \big\{(i-1,k-1),(j,\ell)\big\}.
\end{align*}
The set $\sep(i,j,k,l)$ is a shorthand for the set of pairs $(t,u)$ that
separate the interval pair $\civ{i,j},\civ{k,l}$ into proper
sub-interval-pairs $\civ{i,t},\civ{k,u}$ and $\civ{t+1,j},\civ{u+1,l}$.
Our formula for the source is the following:
\begin{align*}
&A\big(s,\civ{i,j},\civ{k,\ell}\big) =\\
&\min
\begin{dcases}
\min_{(t,u)\in \sep(i,j,k,\ell)}
\Big(A\big(s,\civ{i,t},\
\civ{k,u}\big)+A\big(s,\civ{t+1,j},\civ{u+1,\ell}\big) - 1\Big)\\[-0.5em]
 \hspace{7cm} \text{ if branching at $s$,}\\[1em]
\, \min_{p \in L_1 } \Big(A^-\big(p,\civ{i,j}\big) +
A^+\big(p,\civ{k,\ell}\big)\Big)\\[-0.5em]
\hspace{7cm} \text{ otherwise.}
\end{dcases}
\end{align*}
The initialization of the values is straightforward:
\begin{align*}
A\big(s,\civ{i,i-1},\civ{k,k-1}\big)&=0\\
A\big(s,\civ{i,i},\civ{k,k-1}\big)&=d_{\graph^{*}}(s,q^-_i)\\
A\big(s,\civ{i,i-1},\civ{k,k}\big)&=d_{\graph^{*}}(s,q^+_k)
\end{align*}

Once the one-sided subproblem values are computed, the above dynamic program
can be initialized and computed in increasing order of $(j-i)+(\ell-k)$. The
number of subproblems that we need to compute is $O(n^4)$, each of which
require taking the minimum of $O(n^2)$ values. This enables a running time of
$O(n^6)$. To prove the correctness of the algorithm, we only need to show
that our formulas for $L_1$ and the source are correct. Again, the inequality
$A\big(s,\civ{i,j},\civ{k,\ell}\big) \leq \dots$ is trivial, so we only need
to show that there is an optimal solution which has the desired structure.

We start with an optimal arborescence that is nice when restricted to both
$L_1 \cup \big(\bigcup_{i=2}^h L^-_i\big)$ and $L_1 \cup \big(
\bigcup_{i=2}^h L^+_i \big)$. For a point $p \in L_1$, if the subproblem has
a non-empty interval on both sides, then there is a branching at $p$. The
arborescence can be partitioned into a left and right sub- arborescence, so
equation~\eqref{eq:ADoubleGen} holds.

At the source, we only need to explain the case when there is a branching at
$s$, the other case is trivial. Let $p\in L_1$ be the child of $s$ that has
the smallest $y$-coordinate. Since the left and right sub-arborescences are
nice, the descendant leaves of $p$ on the left form a starting slice
$\civ{i,t}$ of the last level on the left, and the descendant leaves on the
right form a starting slice $\civ{k,u}$ of the last level on the right. The
rest of the intervals are descendants of the other branches. This
demonstrates that the cost of the optimal arborescence can be written as
\[\Big(A\big
(s,\civ{i,t},\civ{k,u}\big)+A\big(s,\civ{t+1,j},\civ{u+1,\ell}\big)\Big)-1.
\tag*{\footnotesize $\Box$}\]
\end{quotation}
%
%
The overall algorithm computes the best feasible broadcast set of each
type, if it exists: 2-hop, path-like, mixed (for both sides), and arborescence-based
Since the minimum broadcast set must have one of these types, the minimum
among these is a minimum $h$-hop broadcast set. The overall running time is~$O(n^6)$.
\qed \end{proof}

\begin{figure}
\begin{center}
\begin{tikzpicture}[x=1.5cm,y=1.5cm,
smallcirc/.style={draw=black,fill=white,circle,inner sep=1pt}]

\draw[very thick, black,<->] (0.25,0.85)--(6.95,0.85);
\node at (3.8,1) {$h$ hops};

\foreach \y in {0.0,0.1,0.4,0.5,...,0.8}
   \draw (0.25,0.35) -- (1,\y);
\foreach \x in {1,2,...,5}
{
   \foreach \y in {0.0,0.1,0.4,0.5,...,0.8}
         \draw (\x,\y) -- (\x+1,\y);
}
\foreach \x in {1,2,...,5}
{
  \draw [red] (\x,0.2) -- (\x+1,0.2);
  \draw [red] (\x,0.3) -- (\x+1,0.3);
}
\draw [red] (0.25,0.35) -- (1,0.2);
\draw [red] (0.25,0.35) -- (1,0.3);
\draw [red] (6,0.2) -- (6.95,0.25);
\draw [red] (6,0.3) -- (6.95,0.25);

\foreach \y in {0.05,0.45,0.65}
{
   \draw (6,\y-0.05) -- (6.95,\y);
   \draw (6,\y+0.05) -- (6.95,\y);
}
\node[smallcirc,label=above:{$s$}] at (0.25,0.35) {};
\foreach \x in {1,2,...,6}
{
   \foreach \y in {0.0,0.1,...,0.8}
      \node[smallcirc] at (\x,\y) {};
}
\foreach \y in {0.05,0.25,0.45,0.65}
{
   \pgfmathtruncatemacro\i{round((\y+0.15)*5)}
   \node[smallcirc,label=right:{$x_\i$}] at (6.95,\y) {};
}
\end{tikzpicture}
\end{center}
\caption{The gadget representing the variables. The red paths form the $x_2$-string.}
\label{fig:bundle}
\end{figure}
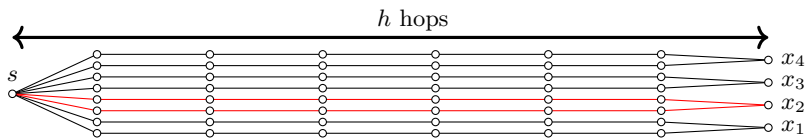

\section{A parameterized look at CDS-UDG}\label{sec:AppHardness}

In this section we prove that \cdsudg is \Wone-hard parameterized by the solution size; our proof heavily relies on the proof of the \Wone-hardness of \dsudg by Marx~\cite{Marx06}.

\mypara{The construction by Marx for DS-UDG.}
Marx uses a reduction from \textsc{Grid Tiling}
\cite{ParamAlg15} (although he does not explicitly state it this way). In a
grid-tiling problem we are given an integer $k$, an integer $n$, and a
collection $\cS$ of $k^2$ non-empty sets $U_{a,b} \subseteq \civ{n} \times
\civ{n}$ for $1 \le a,b \le k$. The goal is to select an element $u_{a,b}\in
U_{a,b}$ for each $1 \le a,b \le k$ such that
\begin{itemize}
\setlength\itemsep{0em}
\item If $u_{a,b}=(x,y)$ and $u_{a+1,b}=(x',y')$, then $x=x'$.
\item If $u_{a,b}=(x,y)$ and $u_{a,b+1}=(x',x')$, then $y=y'$.
\end{itemize}
One can picture these sets in a $k\times k$ matrix: in each cell $(a,b)$, we
need to select a representative from the set $U_{a,b}$ so that the
representatives selected from horizontally neighboring cells agree in the
first coordinate, and representatives from vertically neighboring sets agree
in the second coordinate.

Marx's reduction places $k^2$ gadgets, one for each $U_{a,b}$. A gadget contains
16 blocks of disks, labeled $X_1,Y_1,X_2,Y_2,\dots,X_8,Y_8$, that are
arranged along the edges of a square---see Fig.~\ref{fig:W1-hardness}(i).
Initially, each block $X_{\ell}$  contains $n^2$ disks, denoted by
$X_\ell(1), \dots, X_\ell(n^2)$ and each block $Y_{\ell}$ contains $n^2+1$
disks denoted by $Y_\ell(0), \dots, Y_\ell(n^2)$. The argument $j$ of
$X_\ell(j)$ can be thought of as a pair $(x,y)$ with $1\le x,y\le n$ for
which $f(x,y):=(x-1)n+y=j$. Let $f^{-1}(j)=\left(\iota_1(j),\iota_2(j)\right)
= \left(1+\lfloor j/n\rfloor,1+(j \mod n)\right)$.
\begin{figure}[bt]
    \begin{center}
    \includegraphics[scale=0.9]{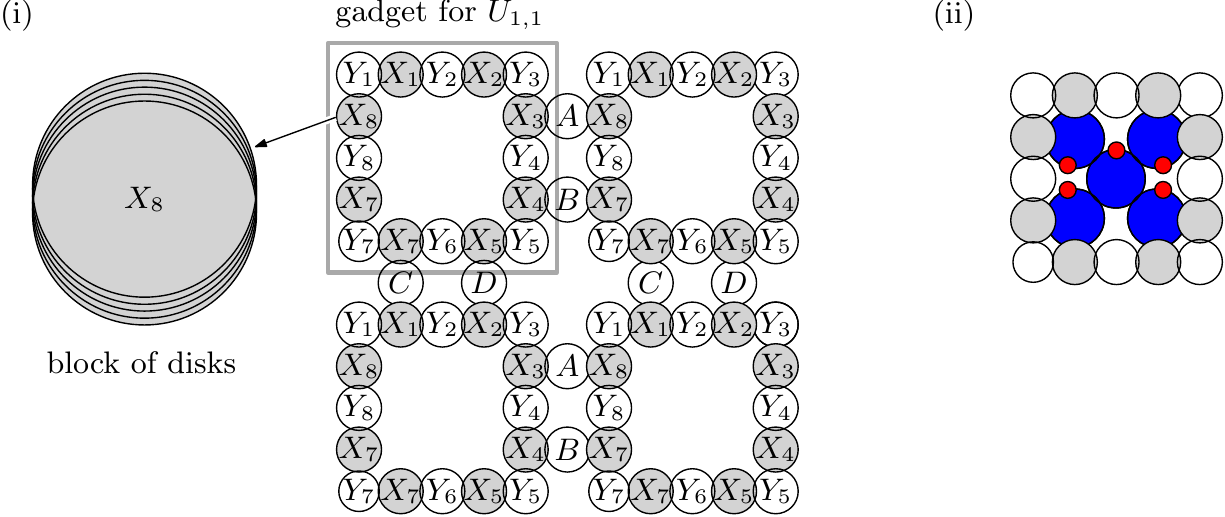}
    \end{center}
    \caption{(i) The construction by Marx. (ii) The idea behind our
    construction.}
    \label{fig:W1-hardness}
\end{figure}
For the final construction, in each gadget at position $(a,b)$, delete all
disks $X_\ell (j)$ for each $\ell=1,\dots,8$ and
$\left(\iota_1(j),\iota_2(j)\right) \not\in U_{a,b}$. This deletion ensures
that the gadgets represent the corresponding set $U_{a,b}$. The construction
is such that a minimum dominating set uses only disks in the $X$-blocks, and
that for each gadget $(a,b)$ the same disk $X_{\ell}(j)$ is chosen for each
$1\leq \ell\leq 8$. This choice signifies a specific choice $u_{a,b}=(x,y)$.
To ensure that the choice for $u_{a,b}$ in the same row and column agrees on
their first and second coordinate, respectively, there are special connector
blocks between neighboring gadgets. The connector blocks are denoted by
$A,B,C$ and $D$ in Fig.~\ref{fig:W1-hardness}(i), and they each contain $n+1$
disks---see Section~\ref{sec:AppHardness} for further details.

\mypara{Our construction for CDS-UDG.}
To extend the construction to \cdsudg, we have to make sure there is a
minimum-size dominating set that is connected. This requires two things.
First, we must add new disks inside the gadgets---that is, in the empty space
surrounded by the $X$- and $Y$-blocks---to guarantee a connection between all
chosen $X_\ell(j)$ disks without interfering with the disks in the
$Y$-blocks. Second, we need to connect all the different gadgets. This time,
in addition to avoiding the $Y$-blocks, we also need to avoid interfering
with the connector blocks.

The idea is as follows. Inside each gadget we add several pairs of disks,
consisting of a \emph{parent} disk and a \emph{leaf} disk. The parent disks
are placed such that, for any choice of one disk from each of the $X$-blocks,
the parent disks together with the eight chosen disks from the $X$-blocks
form a connected set. Moreover, the parent disks do not intersect any disk in
a $Y$-block. See Fig.~\ref{fig:W1-hardness}(ii) for an illustration; the
parent disks are blue in the figure. For each parent disk we add a leaf disk
---the red disks in the figure---that only intersects its parent disk. This
ensures there is a minimum dominating set containing all the parent disks,
which in turn implies that any minimum dominating set for the gadget is
connected.

In Fig.\ref{fig:W1-hardness}(ii) we used disks of different sizes.
Unfortunately this is not allowed, which makes the construction significantly
more tricky. To be able to place the pairs in a suitable way, we need to
create more space inside the gadget. To this end we use a gadget consisting
of 16 (instead of eight) $X$- and $Y$-blocks. This will also give us
sufficient space to put parent-leaf pairs in between the gadgets, so the
dominating sets from adjacent gadgets are connected through the parent disks;
see Section~\ref{sec:AppHardness} for details. Thus the size of a
minimum connected dominating set in the new construction is equal to the size
of a minimum dominating set in the old construction plus the number of parent
disks. Hence, we can decide if the \textsc{Grid Tiling} instance has a
solution by checking the size of the minimum connected dominating set in
our construction. Thus \cdsudg is \Wone-hard. Moreover, if the Exponential
Time Hypothesis (ETH) holds, then there is no algorithm for \textsc{Grid
Tiling} that runs in time $f(k)n^{o(k)}$ \cite{ParamAlg15}. We thus obtain
the following result.
\begin{theorem}\label{thm:W1-hardness}
The broadcast problem and \cdsudg are \Wone-hard when parameterized by the
solution size. Moreover, there is no $f(k)n^{o(\sqrt{k})}$ algorithm for
these problems, where $n$ is the number of input disks and $k$ is the size of
the solution, unless ETH fails.
\end{theorem}

\noindent \emph{Remark.}
Using a modified version of an algorithm by Marx and Pilipczuk~\cite{Marx15},
it is possible to construct an algorithm for \cdsudg with running time
$n^{O(\sqrt{k})}$.

\mypara{Some details of the construction in \cite{Marx06}.} In every block,
the place of each disk center is defined with regard to the midpoint of the
block, $(x(z),y(z))$. The center of each circle is of the form
$(x(z)+\alpha\epsilon,y(z)+\beta\epsilon)$ where $x(z),y(z), \alpha$ and
$\beta$ are integers, and $\epsilon>0$ a small constant. We say that the
\emph{offset} of the disk centered at $(x(z)+\alpha\epsilon,
y(z)+\beta\epsilon)$ is $(\alpha,\beta)$. Note that $|\alpha|,|\beta|\le n$,
and $\epsilon < n^{-2}$, so the disks in a block all intersect each other.
The offsets of $X$ and $Y$-blocks are defined as follows.

\begin{center}
\begin{tabular}{ll}
offset$(X_1(j))=( j,-\iota_2(j))$ & offset$(Y_1(j))=(j+0.5,j+0.5)$ \\

offset$(X_2(j))=( j, \iota_2(j))$ & offset$(Y_2(j))=(j+0.5,-n)$ \\

offset$(X_3(j))=(-\iota_1(j),-j)$ & offset$(Y_3(j))=(j+0.5,-j-0.5)$ \\

offset$(X_4(j))=( \iota_1(j),-j)$ & offset$(Y_4(j))=(-n,-j-0.5)$ \\

offset$(X_5(j))=(-j, \iota_2(j))$ & offset$(Y_5(j))=(-j-0.5,-j-0.5)$ \\

offset$(X_6(j))=(-j,-\iota_2(j))$ & offset$(Y_6(j))=(-j-0.5,n)$ \\

offset$(X_7(j))=( \iota_1(j), j)$ & offset$(Y_7(j))=(-j-0.5,j+0.5)$ \\

offset$(X_8(j))=(-\iota_1(j), j)$ & offset$(Y_8(j))=(n,j+0.5)$ \\
\end{tabular}
\end{center}

We remark some important properties. First, two disks can intersect only if
they are in the same or in neighboring blocks. Consequently, one needs at
least eight disks to dominate a gadget. The second important property is that
disk $X_\ell(j)$ dominates exactly $Y_{\ell}(j),\dots,Y_{\ell}(n^2)$ from the
``previous'' block $Y_\ell$, and $Y_{\ell+1}(0),\dots,Y_{\ell+1}(j-1)$ from
the ``next'' block $Y_{\ell+1}$. This property can be used to prove the
following key lemma.

\begin{lemma}[Lemma 1 of~\cite{Marx06}]\label{lem:Mgadget}
Assume that a gadget is part of an instance such that none of the blocks
$Y_i$ are intersected by disks outside the gadget. If there is a dominating
set $\Delta$ of the instance that contains exactly $8k^2$ disks, then there
is a \emph{canonical} dominating set $\Delta'$ with $|\Delta'| = |\Delta|$,
such that for each gadget $\graph$, there is an integer $1 \le j^G \le n$
such that $\Delta'$ contains exactly the disks $X_1(j^G), \dots, X_8(j^G)$
from $\graph$.
\end{lemma}

In the gadget $G_{a,b}$, the value $j$ defined in the above lemma represents
the choice of $s_{a,b}=\left(\iota_1(j),\iota_2(j)\right)$ in the grid tiling
problem. Our deletion of certain disks in $X$-blocks ensures that
$\left(\iota_1(j),\iota_2(j)\right) \in U_{a,b}$. Finally, in order to get a
feasible grid tiling, gadgets in the same row must agree on the first
coordinate, and gadgets in the same column must agree on the second
coordinate. These blocks have $n+1$ disks each, with indices $0,1,\dots,n$.
We define the offsets in the connector gadgets the following way.

\begin{center}
\begin{tabular}{ll}
offset$(A_j)=(-j-0.5,-n^2-1)$ & offset$(B_j)=(j+0.5,n^2+1)$ \\

offset$(C_j)=(n^2+1,-\iota_2(j))$ & offset$(D_j)=(-n^2-1,\iota_2(j))$ \\
\end{tabular}
\end{center}

Using this definition, it is easy to prove the following lemma.

\begin{lemma}\label{lem:Mconnect}
Let $\Delta$ be a canonical dominating set. For horizontally neighboring
gadgets $\graph$ and $H$ representing $j_G$ and $j_H$, the disks of the
connector block $A$ are dominated if and only if $\iota_1(j_G)\le
\iota_1(j_H)$; the disks of $B$ are dominated if and only if $\iota_1(j_G)\ge
\iota_1(j_H)$. Similarly, for vertically neighboring blocks $G'$ and $H'$,
the disks of block $C$ are dominated if and only if $\iota_2(j_{G'})\le
\iota_2(j_{H'})$; the disks of $D$ are dominated if and only if
$\iota_2(j_{G'})\ge \iota_2(j_{H'})$.
\end{lemma}

With the above lemmas, the correctness of the reduction follows. A feasible
grid tiling defines a dominating set of size $8k^2$: in gadget $G_{a,b}$, the
dominating disks are $X_\ell\left(f(s_{a,b})\right), \;\ell=1,\dots,8$. On
the other hand, if there is a dominating set of size $8k^2$, then there is a
canonical dominating set of the same size  that defines a feasible grid
tiling.

\mypara{Details of the CDS-UDG construction.} To extend the construction to
\cdsudg, we want to make sure that minimum-size dominating set is connected.
This requires two things. First, we must add new disks ``inside'' the gadgets
--- that is, in the empty space surrounded by the $X$ and $Y$-blocks --- such
that a canonical minimum dominating set includes some new disks that connect
the chosen $X_\ell(j)$ disks without interfering with disks in the
$Y$-blocks. Second, we need to connect all the different gadgets. This time
in addition to avoiding the $Y$-blocks, we also need to avoid interfering
with the connector blocks.

In order to have enough space, our gadgets contain 16 $X$-blocks and
16 $Y$-blocks instead of eight. The offsets of disks inside the blocks
are not modified: we use the same building blocks.
Fig.~\ref{fig:connectblocks} shows how we arrange these blocks, and depicts
the connector block placement.

The analogue of Lemma~\ref{lem:Mgadget} and Lemma~\ref{lem:Mconnect} are true
here; we have a construction that could be used to prove the \Wone-hardness of
\dsudg, with canonical sets of size $16k^2$, that contain one disk from each
$X$-block and $X'$-block. We extend this construction with parent-leaf pairs
so that we have canonical dominating sets that span a connected subgraph.

The most important property of the blocks that we use is that for a small
enough value $\epsilon$, the boundaries of the disks in a block all lie
inside a small width annulus - for this reason, the blocks in our pictures
are depicted with thick boundary disks. In order for a parent disk $p$ to
intersect every disk in a block it is sufficient if the boundary of $p$
crosses this annulus.

We are going to add 72 extra disks to every gadget, and 4 ``connector'' disks
between every pair of horizontally or vertically neighboring gadgets,
resulting in canonical dominating sets of size  $16k^2 + 36k^2 +
4k(k-1)=56k^2 - 4k$ (Note that only the parent disks are included in the
canonical set). In other words, the new construction has a connected
dominating set of size $56k^2 - 4k$ if and only if there is a feasible grid
tiling.

Inside any of the blocks, all offsets are in the rectangle with bottom left
$(-n^2-1,-n^2-1)$ and top right $(n^2+1,n^2+1)$. Consequently, every circle
in the block with center $(r,s)$ passes through the square with bottom left
$\big((-n^2-1)\epsilon, 1-(n^2+1)\epsilon\big)$ and top right
$\big((n^2+1)\epsilon, 1+(n^2+1)\epsilon\big)$. There are three similar
squares that also have this property, which we can get by rotating the square
around the midpoint of the block by $90$, $180$ and $270$ degrees.
Consequently, a unit disk that contains such a square intersects all the
disks in the given block. For an example with $n=3$ and $\epsilon = 0.02$ for
the block $X_2$, see Fig.~\ref{fig:smallsquares}.

\begin{figure}
\begin{center}
\begin{tikzpicture}[x=2.2 cm, y=2.2 cm]
\clip (-1.31,-1.31) rectangle (1.31,1.31);
\draw [step=1,dotted] (-1,-1) grid (1,1);
\pgfmathsetmacro{\k}{0.02}
\fill [draw=black, fill=black, fill opacity = 0.1] (\k,\k) circle (1);
\fill [draw=black, fill=black, fill opacity = 0.1] (2*\k,2*\k) circle (1);
\fill [draw=black, fill=black, fill opacity = 0.1] (3*\k,3*\k) circle (1);
\fill [draw=black, fill=black, fill opacity = 0.1] (4*\k,1*\k) circle (1);
\fill [draw=black, fill=black, fill opacity = 0.1] (5*\k,2*\k) circle (1);
\fill [draw=black, fill=black, fill opacity = 0.1] (6*\k,3*\k) circle (1);
\fill [draw=black, fill=black, fill opacity = 0.1] (7*\k,1*\k) circle (1);
\fill [draw=black, fill=black, fill opacity = 0.1] (8*\k,2*\k) circle (1);
\fill [draw=black, fill=black, fill opacity = 0.1] (9*\k,3*\k) circle (1);
\draw [very thick, red] (-10*\k,1-10*\k) rectangle (10*\k,1+10*\k);
\draw [very thick, red] (-10*\k,-1-10*\k) rectangle (10*\k,-1+10*\k);
\draw [very thick, red] (1-10*\k,-10*\k) rectangle (1+10*\k,10*\k);
\draw [very thick, red] (-1-10*\k,-10*\k) rectangle (-1+10*\k,10*\k);
\end{tikzpicture}
\end{center}
\caption{Circles in a block. The squares intersect every disk in the
block.}\label{fig:smallsquares}
\end{figure}
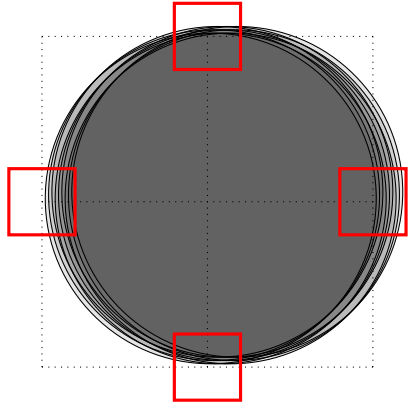

\medskip
\noindent\textit{Connecting neighboring gadgets.} For a pair of horizontally
neighboring gadgets, we add two pairs of disks that connect $X'_3$ from the
left gadget to $X'_8$ in the right gadget. This arrangement is depicted on
the left of Fig.~\ref{fig:connectside}. The parent disk with center $T_1$
intersects every disk in the block $X'_3$ of the left gadget, and the other
parent intersects every disk in the block $X'_8$. The two leaf disks (red
disks in the figure) only intersect their parent. Let the origin be the
center of the block $X'_3$ in the left gadget. The coordinates for the disk
centers are:

\begin{align*}
T_1 =& (1.3,0.4) \;\;\;\; U_1 = (2,1.55) \\
T_2 =& (2.7,-0.4) \;\;\;\; U_2 = (2,-1.55)
\end{align*}

We use a rotated version of these four disks for vertical connections, where
the parents connect $X'_5$ from the upper gadget and $X'_2$ from the lower
gadget.

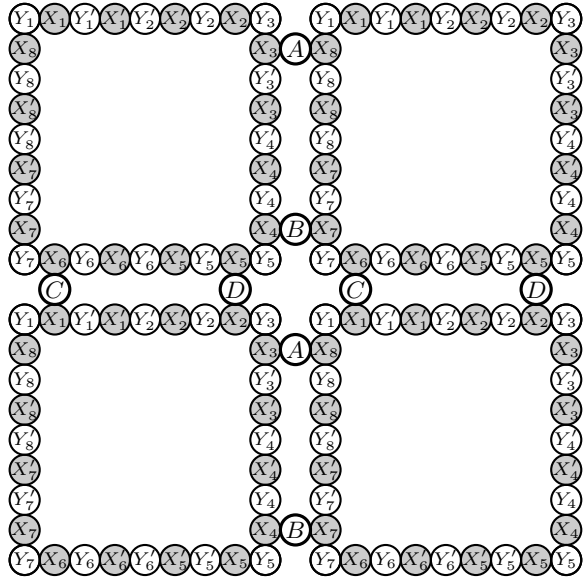
\begin{figure}
\begin{center}
\begin{tikzpicture}[x=0.2 cm,y=0.2 cm]
\begin{scriptsize}
\blocks
\begin{scope}[shift={(20,0)}]
\blocks
\end{scope}
\begin{scope}[shift={(20,20)}]
\blocks
\end{scope}
\begin{scope}[shift={(0,20)}]
\blocks
\end{scope}
\end{scriptsize}
\foreach \x in {0,4,8,12,16,20,24,28,32,36}
   \foreach \y in {0,16,20,36}
      {\draw [ybi] (\x,\y) circle (1);
      \draw [ybi] (\y,\x) circle (1);}
\foreach \x in {2,6,10,14,22,26,30,34}
   \foreach \y in {0,16,20,36}
      {\draw [xbi] (\x,\y) circle (1);
      \draw [xbi] (\y,\x) circle (1);}
\begin{footnotesize}
\draw [draw=black,fill=none,very thick] (18,2) circle (1);
\node [black] at (18,2) {$B$};
\draw [draw=black,fill=none,very thick] (18,14) circle (1);
\node [black] at (18,14) {$A$};
\draw [draw=black,fill=none,very thick] (18,22) circle (1);
\node [black] at (18,22) {$B$};
\draw [draw=black,fill=none,very thick] (18,34) circle (1);
\node [black] at (18,34) {$A$};

\draw [draw=black,fill=none,very thick] (2,18) circle (1);
\node [black] at (2,18) {$C$};
\draw [draw=black,fill=none,very thick] (14,18) circle (1);
\node [black] at (14,18) {$D$};
\draw [draw=black,fill=none,very thick] (22,18) circle (1);
\node [black] at (22,18) {$C$};
\draw [draw=black,fill=none,very thick] (34,18) circle (1);
\node [black] at (34,18) {$D$};
\end{footnotesize}
\end{tikzpicture}
\end{center}
\caption{Connecting neighboring gadgets}\label{fig:connectblocks}
\end{figure}

\begin{figure}
\begin{center}

\begin{tikzpicture}[x=0.8cm,y=0.8cm]
\pgfmathsetmacro{\de}{0.12}
\ccc{6,6}{6-\de,6}
\ccc{8,6}{8,6+4*\de}
\ccc{10,6}{10,6-\de}
\ccc{10,8}{10-4*\de,8}
\ccc{10,10}{10+\de,10}
\ccc{8,10}{8,10-4*\de}
\ccc{6,10}{6,10+\de}
\ccc{6,8}{6+4*\de,8}
\end{tikzpicture}
\hspace{1cm}
\begin{tikzpicture}[x=0.28cm, y=0.28cm]
\clip (-1.1,-1.1) rectangle (17.1,17.1);
\draw [shift={(1,1)},step=2,dotted] (-3,-3) grid (16,16);
\pgfmathsetmacro{\de}{0.12}
\pgfmathsetmacro{\shift}{2-\de*\de-\de*\de*\de*\de}
\pgfmathsetmacro{\base}{2-\de}
\ccc{6,6}{6-\de,6}
\ccc{8,6}{8,6+4*\de}
\ccc{10,6}{10,6-\de}
\ccc{10,8}{10-4*\de,8}
\ccc{10,10}{10,10+\de}
\ccc{8,10}{8,10-4*\de}
\ccc{6,10}{6,10+\de}
\ccc{6,8}{6+4*\de,8}
\foreach \x in {0,4,8,12,16}
   \foreach \y in {0,16}
   {
      \draw [yb,line width=2pt] (\x,\y) circle (1);
      \draw [yb,line width=2pt] (\y,\x) circle (1);
   }
\foreach \x in {2,6,10,14}
   \foreach \y in {0,16}
   {
      \draw [xb,line width=2pt] (\x,\y) circle (1);
      \draw [xb,line width=2pt] (\y,\x) circle (1);
   }
\end{tikzpicture}
\caption{Left: the circles in the center of every gadget; Right: placement
inside a gadget}\label{fig:center}
\end{center}
\end{figure}
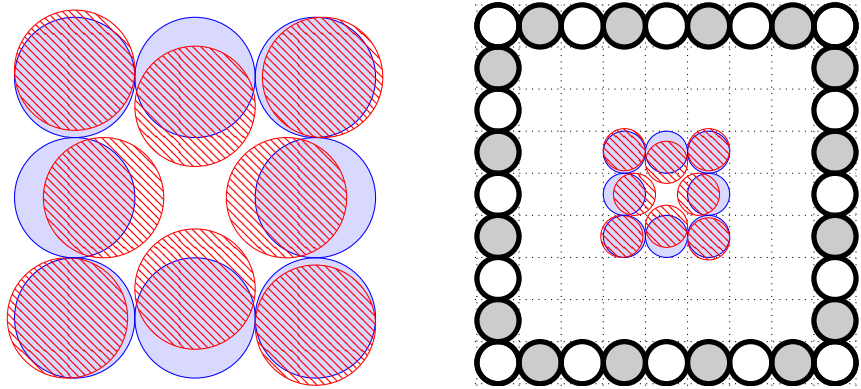

\begin{figure*}
\begin{center}
\begin{tikzpicture}[x=0.65cm,y=0.65cm]
\clip (-0.5,-3.2) rectangle (4.5,3.5);
\draw [shift={(1,1)},step=2,dotted] (-5,-5) grid (5,5);
\begin{Large}

\node at (0,-2) {$Y'_4$};
\node at (0,0) {$X'_3$};
\node at (0,2) {$Y'_3$};

\node at (4,2) {$Y_8$};
\node at (4,0) {$X'_8$};
\node at (4,-2) {$Y'_8$};

\node at (2,3.2) {$A$};
\end{Large}
\pgfmathsetmacro{\de}{0.05}
\draw [xb] (0,0) circle (1);
\draw [xb] (4,0) circle (1);
\draw [xb] (0,4) circle (1);
\draw [xb] (4,4) circle (1);
\draw [xb] (0,-4) circle (1);
\draw [xb] (4,-4) circle (1);
\draw [yb] (2,4) circle (1);
\draw [yb] (0,2) circle (1);
\draw [yb] (0,-2) circle (1);
\draw [yb] (4,2) circle (1);
\draw [yb] (4,-2) circle (1);
\cc{1.3,0.4}{2,1.55}
\cc{2.7,-0.4}{2,-1.55}

\node[draw=black,circle,inner sep=1pt,label=below:{$T_1$}] at (1.3,0.4) {};
\node[draw=red,circle,inner sep=1pt,label={[color=red]above:$U_1$}] at
(2,1.55) {};

\node[draw=black,circle,inner sep=1pt,label={$T_2$}] at (2.7,-0.4) {};
\node[draw=red,circle,inner sep=1pt,label={[color=red]below:$U_2$}] at
(2,-1.55) {};
\end{tikzpicture}
\hfill
\begin{tikzpicture}[x=0.6cm,y=0.6cm]
\clip (-1.1,-1.1) rectangle (12.5,6.3);
\draw [shift={(1,1)},step=2,dotted] (-2,-2) grid (15,9);
\begin{Large}
\blocks
\end{Large}
\pgfmathsetmacro{\de}{0.05}
\pgfmathsetmacro{\shift}{2-\de*\de-\de*\de*\de*\de}
\pgfmathsetmacro{\base}{2-\de}
\cc{\base,2-\de}{\base,3-\de}
\cc{\base + \shift,2+\de}{\base + \shift,2+2*\de}
\cc{\base + 2*\shift,2-\de}{\base + 2*\shift,2+6*\de}
\cc{\base + 3*\shift,2+\de}{\base + 3*\shift,2+2*\de}
\cc{\base + 4*\shift,2-\de}{11,2-\de}
\cc{\base + 4*\shift,4-2*\de}{11,4}
\cc{8,4+3*\de}{7,4+3*\de}
\ccc{6,6}{6-\de,6}
\ccc{8,6}{8,6+4*\de}
\ccc{10,6}{10,6-\de}
\ccc{10,8}{10-4*\de,8}
\ccc{10,10}{10,10+\de}
\ccc{8,10}{8,10-4*\de}
\ccc{6,10}{6,10+\de}
\ccc{6,8}{6+4*\de,8}
\foreach \x in {0,4,8,12,16}
   \foreach \y in {0,16}
   {
      \draw [yb,line width=2pt] (\x,\y) circle (1);
      \draw [yb,line width=2pt] (\y,\x) circle (1);
   }
\foreach \x in {2,6,10,14}
   \foreach \y in {0,16}
   {
      \draw [xb,line width=2pt] (\x,\y) circle (1);
      \draw [xb,line width=2pt] (\y,\x) circle (1);
   }
\end{tikzpicture}
\end{center}
\caption{Left: connecting horizontally; right: connecting one side to the
middle.}\label{fig:connectside}
\end{figure*}
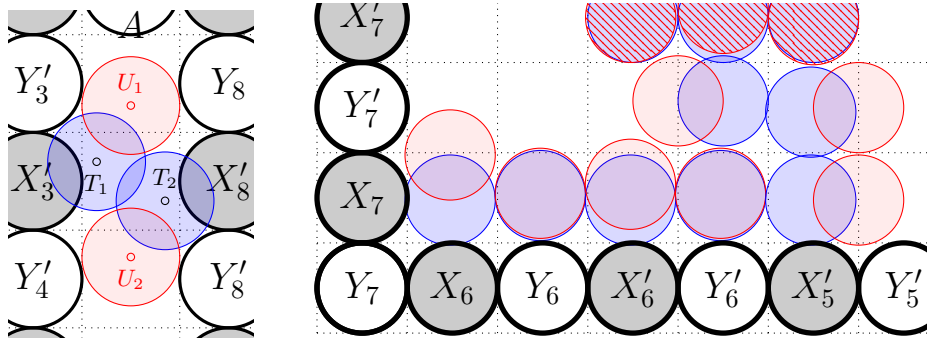

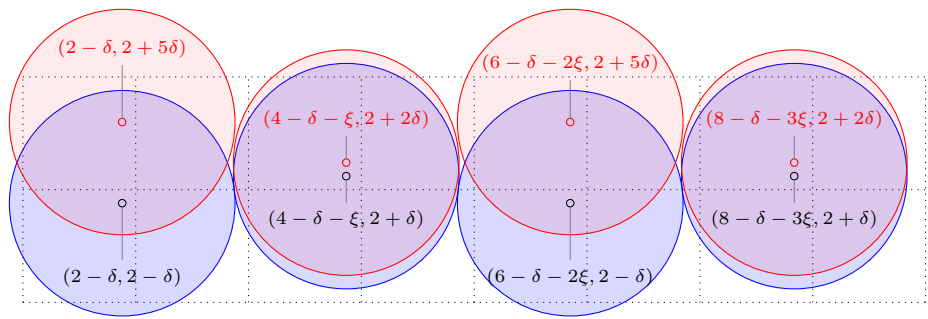
\begin{figure}
\begin{center}
\begin{tikzpicture}[x=1.5cm,y=1.5cm]
\draw [step=1,dotted] (1,1) grid (9,3);
\pgfmathsetmacro{\de}{0.12}
\pgfmathsetmacro{\shift}{2-\de*\de-\de*\de*\de*\de}
\pgfmathsetmacro{\base}{2-\de}
\cc{\base,2-\de}{\base,2+5*\de}
\cc{\base + \shift,2+\de}{\base + \shift  ,2+2*\de}
\cc{\base + 2*\shift,2-\de}{\base + 2*\shift,2+5*\de}
\cc{\base + 3*\shift,2+\de}{\base + 3*\shift,2+2*\de}
\begin{scriptsize}
\node[draw=black,circle,inner sep=1pt,pin={[pin
distance=7mm,black]-90:{$(2-\delta,2-\delta)$}}] at (\base,\base) {};
\node[draw=red,circle,inner sep=1pt,pin={[pin
distance=7mm,red]90:{$(2-\delta,2+5\delta)$}}] at (\base,\base+6*\de) {};

\node[draw=black,circle,inner sep=1pt,pin={[pin
distance=3mm,black]-90:{$(4-\delta-\xi,2+\delta)$}}] at (\base+\shift,2+\de)
{};
\node[draw=red,circle,inner sep=1pt,pin={[pin distance=3mm,
red]90:{$(4-\delta-\xi,2+2\delta)$}}] at (\base+\shift,2+2*\de) {};

\node[draw=black,circle,inner sep=1pt,pin={[pin
distance=7mm,black]-90:{$(6-\delta-2\xi,2-\delta)$}}] at
(\base+2*\shift,2-\de) {};
\node[draw=red,circle,inner sep=1pt,pin={[pin distance=5mm,
red]90:{$(6-\delta-2\xi,2+5\delta)$}}] at (\base+2*\shift,2+5*\de) {};

\node[draw=black,circle,inner sep=1pt,pin={[pin
distance=3mm,black]-90:{$(8-\delta-3\xi,2+\delta)$}}] at
(\base+3*\shift,2+\de) {};
\node[draw=red,circle,inner sep=1pt,pin={[pin distance=3mm,
red]90:{$(8-\delta-3\xi,2+2\delta)$}}] at (\base+3*\shift,2+2*\de) {};
\end{scriptsize}
\end{tikzpicture}
\end{center}
\caption{The zig-zag arrangement}\label{fig:zigzag}
\end{figure}

\noindent\textit{Disks inside gadgets.} We begin by adding eight disk pairs
to the center. The parents are arranged in a square, touching the neighbors,
and the leafs are placed so that it is possible to connect from the outside
on each side. See Fig.~\ref{fig:center} for a picture: the corresponding leaf
disks have parallel lines as a pattern.

Let $\delta>0$ be a small constant to be specified later. From now on, we fix
the origin in the center of the bottom left block, $Y_7$. The coordinates of
the disks centers are given below; in each pair we specify the coordinates of
a parent and its leaf.

\[\begin{matrix}
(6,6),(6-\delta,6) & (8,6),(8,6+4\delta) &
(10,6),(10,6-\delta) & (10,8),(10-4\delta,8) \\
(10,10),(10,10+\delta) & (8,10),(8,10-4\delta) &
(6,10),(6,10+\delta) & (6,8),(6+4\delta,8)\\
\end{matrix}\]

In order to connect the $X$-blocks, we need to connect the blocks of each
side to the central disks. For this purpose, we are going to use a zigzag
pattern of disks. The first parent disk intersects all disks in $X_6$ and
$X_7$ (i.e., it crosses the small squares of $X_6$ and $X_7$ that are facing
the inside of the gadget). The second parent is above the block $Y_6$, but it
is disjoint from it. The next with center $p_3$ intersects all disks in
$X'_6$ , and the disk around $p_4$ is disjoint from the disks in $Y'_6$.
Finally, the disk around $p_5$ intersects all disks in $X'_5$. See
Fig.~\ref{fig:zigzag} for an example. The leafs follow a more complicated
pattern. In our zigzag pattern, two neighboring parents touch each other. We
need them to have distance $2\delta$ along the $y$ axis, so the distance
along the $x$ axis is $\sqrt{4-4\delta^2}$. Let $\xi=2-\sqrt{4-4\delta^2}$.
Note that
\[2-\delta^2-\delta^4 < \sqrt{4-4\delta^2} < 2-\delta^2,\]
so $\delta^2 < \xi <\delta^2 + \delta^4$. We add two more disk pairs to this
pattern, and some modifications to the leafs. These seven disk pairs are
depicted on the right side of Fig.~\ref{fig:connectside}. We list the
coordinates of the disk centers below.

\begin{align*}
p_1 =& (2- \delta,2-\delta) \; & \; \ell_1 =& (2-\delta,3-\delta) \\
p_2 =& (4-\delta-\xi,2+\delta) \; & \; \ell_2 =& (4-\delta-\xi,2+2\delta) \\
p_3 =& (6-\delta-2\xi,2-\delta) \; & \; \ell_3 =& (6-\delta-2\xi,1+5\delta)\\
p_4 =& (8-\delta-3\xi,2+\delta) \; & \; \ell_4 =& (8-\delta-3\xi,2+2\delta)\\
p_5 =& (10-\delta-4\xi,2-\delta) \; & \; \ell_5 =& (11,2-\delta)\\
p_6 =& (10-\delta-4\xi,4-\delta) \; & \; \ell_6 =& (11,4)\\
p_7 =& (8,4+3\delta) \; & \; \ell_7 =& (7,4+3\delta)
\end{align*}

Our final gadget can be attained by rotating the above seven disk pairs
around the center $(8,8)$ by $90,180$ and $270$ degrees: see
Fig.~\ref{fig:gadget}. We added the spanned edges of a canonical dominating
set to this picture.

\begin{figure*}
\begin{center}
\begin{tikzpicture}[x=0.53cm,y=0.53cm]
\clip (-4.1,-4.1) rectangle (20.1,20.1);
\draw [shift={(1,1)},step=2,dotted] (-8,-8) grid (25,25);
\foreach \x in {-4,0,4,8,12,16,20}
   \foreach \y in {-4,0,16,20,24}
      {\draw [yb] (\x,\y) circle (1);
      \draw [yb] (\y,\x) circle (1);}
\foreach \x in {2,6,10,14}
   \foreach \y in {-4,0,16,20,24}
      {\draw [xb] (\x,\y) circle (1);
      \draw [xb] (\y,\x) circle (1);}
\pgfmathsetmacro{\de}{0.05}
\pgfmathsetmacro{\shift}{2-\de*\de-\de*\de*\de*\de}
\pgfmathsetmacro{\base}{2-\de}
\foreach \x in {0,90,180,270}
	{\sevencircles{\x}}
\ccc{6,6}{6-\de,6}
\ccc{8,6}{8,6+4*\de}
\ccc{10,6}{10,6-\de}
\ccc{10,8}{10-4*\de,8}
\ccc{10,10}{10,10+\de}
\ccc{8,10}{8,10-4*\de}
\ccc{6,10}{6,10+\de}
\ccc{6,8}{6+4*\de,8}

\draw [yb] (18,2) circle (1);
\draw [yb] (18,14) circle (1);
\draw [yb] (-2,2) circle (1);
\draw [yb] (-2,14) circle (1);
\draw [yb] (2,18) circle (1);
\draw [yb] (14,18) circle (1);
\draw [yb] (2,-2) circle (1);
\draw [yb] (14,-2) circle (1);
\begin{scope}
\tikzstyle{every node}=[draw,circle,fill=black,minimum size=4pt,inner
sep=0pt]
\foreach \x in {0,90,180,270}
	{\gadgetgraph{\x}}

\connectors{0,0}
\graphconnector{0,0}
\connectors{-20,0}
\graphconnector{-20,0}
\begin{scope}[rotate around={{-90}:(8,8)}]
   \connectors{0,0}
   \graphconnector{0,0}
   \connectors{-20,0}
   \graphconnector{-20,0}
\end{scope}
\end{scope}
\end{tikzpicture}
\end{center}
\caption{A gadget in the final construction. The dashed lines are spanned
edges of a canonical dominating set.}\label{fig:gadget}
\end{figure*}
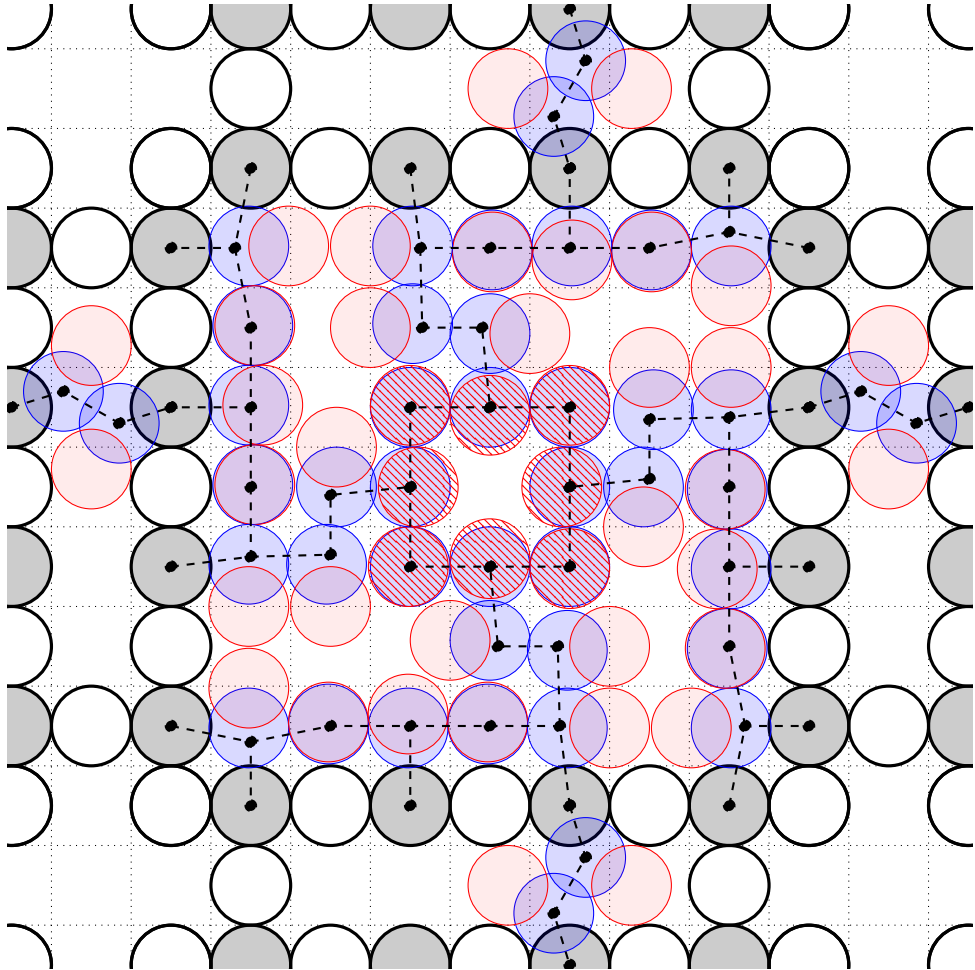

We can now turn to the proof of the following theorem.

\begin{theorem}\label{thm:cdsudg_w1h}
The \cdsudg problem is \Wone-hard.
\end{theorem}

\begin{proof}
A feasible grid tiling defines $16k^2$ disks: in gadget $(a,b)$, we include
the disks $X_\ell\left(f(s_{a,b})\right)$ and
$X'_\ell\left(f(s_{a,b})\right)$ for all $\;\ell=1,\dots,8$. We add all
parent disks of the construction, this results in a connected dominating set
of size $56k^2-4k$. In the other direction, if there is a connected
dominating set of size $56k^2-4k$, then there is a canonical dominating set
of the same size, whose disks inside $X$-blocks and $X'$-blocks define a
feasible grid tiling. Thus, it is sufficient to prove that the intersection
patterns are as described.

It can be verified using the coordinates that our final leaf disks only
intersect their parent disk, and also that the parent disks form a connected
subgraph both inside gadgets and at every connection. We need to show that
the parents inside a gadget connect all the $X$-blocks of a gadget, and that
the horizontal and vertical connectors intersect the two $X$-blocks that they
need to connect. In all of these cases, it is sufficient to show that the
parent disk contains one of the four squares that we associated with each
block. For connector disks, it is easy to see that the center of one of the
four squares is covered by the interior of the corresponding parent disk
(i.e., the square around $(1,0)$ is contained in the interior of
$\disk(T_1)$). By choosing a small enough value for $\epsilon$, the square
is contained in the parent disk.

For the inner connections of gadgets, it is sufficient to show that the inner
squares of $X_7, X_6, X'_6$ and $X'_5$ are contained in
$\disk(p_1),\disk(p_1),\disk(p_3)$ and $\disk(p_5)$ respectively: the
other sides have the same containments since the rotation around $(8,8)$ by
$90,180$ and $270$ degrees are automorphisms on the small squares. The
largest distance between parent disk and the corresponding small square is at
$\disk(p_5)$ and the inner small square of block $X'_5$. The farthest corner
of the square from $p_5$ is $\big(10+(n^2+1)\epsilon,1-(n^2+1)\epsilon\big)$.
Let $\epsilon<\frac{1}{2n^3}$ and $\delta<1$. The distance squared from $p_5$
has to be at most $1$:
\begin{align*}
 & \left(10-\delta-4\xi - (10+(n^2+1)\epsilon) \right)
  ^2+\left(2-\delta - (1-(n^2+1)\epsilon)   \right)^2\\
< &  \left(\delta+4\delta^2+ \delta^4+ \frac{1}{n} \right)^2
+\left(1-\delta + \frac{1}{n} \right)^2\\
= & 1 - 2\delta + \frac{4}{n} + O\left(\frac{\delta}{n}\right) + O(\delta^2)
\end{align*}

Let $\delta= \frac{1}{\sqrt{n}}$. For $n$ large enough,
\[1 - 2\delta + \frac{4}{n} + O\left(\frac{\delta}{n}\right) +
 O(\delta^2)\;=\;1 - \frac{2}{\sqrt{n}} + \frac{4}{n} +
 O\left(\frac{1}{n\sqrt{n}}\right) + O\left(\frac{1}{n}\right) < 1. \]

Note that the coordinates of each point can be represented with $O(\log
n)$ bits, since a precision of $c/n^4$ is sufficient for the construction.
\qed \end{proof}

We can let one of the blue parent disks be the source disk: in this way, the
minimum broadcast sets equal the minimum connected dominating
sets. We get the following corollary.

\begin{corollary}
The broadcast problem is \Wone-hard parameterized by the size of the broadcast
set.
\end{corollary}

\section{Broadcasting in a wide strip}
We show that the broadcast problem remains polynomial in a strip of any
constant width, or more precisely, it is in \xp for the parameter $w$ (the
width of the strip).

\begin{restatable}{theorem}{widebroadcast}\label{thm:widebroadcast}
The broadcast problem and \cdsudg can be solved in $n^{O(w)}$ time on a strip
of width~$w$. Moreover, there is no algorithm for \cdsudg or the broadcast problem with
runtime $f(w)n^{o(w)}$ unless ETH fails.
\end{restatable}

We begin by showing the following key lemma.
\begin{lemma}\label{lem:sparseDomSet}
Let $D$ be the disk centers of a minimum connected dominating set of a unit
disk graph on a strip of width $w$, and let $R$ be an axis parallel rectangle
of size $2\times w$. Then the number of points in $D\cap R$ is at most
$\frac{32w}{\sqrt{3}}+14$.
\end{lemma}
\begin{proof}
Let $R'$ be the $1$-neighborhood of $R$ inside the strip (so $R'$ is a
$4\times w$ rectangle). We subdivide $R'$ into cells of diameter $1$ by
introducing a rectangular grid with side lengths $1/2$ and $\sqrt{3}/2$.
Overall, we get $8\lceil\frac{w}{\sqrt{3}/2}\rceil<\frac{16w}{\sqrt{3}}+8$
cells in $R'$. Let $\graph$ be the unit disk graph spanned by the centers
that fall in $R'$. The points that fall into a grid cell form a clique in
$\graph$. Let $G'$ be the graph that we get if we contract the vertices of
$\graph$ in each cell. Let $T$ be a spanning tree of $G'$. We can represent
$T$ in the original graph in the following way. For each edge $uv\in E(T)$
select vertices $u',v'$ of distance at most $1$ from the cell of $u$ and $v$
respectively. We know that there are such points since otherwise $uv$ could
not be an edge in $G'$. Since $T$ has at most $(\frac{16w}{\sqrt{3}}+8)-1$
vertices, this selection gives us a point set $H$ of size at most
$2(\frac{16w}{\sqrt{3}}+7) = \frac{32w}{\sqrt{3}}+14$.

Suppose for contradiction that $R\cap D >\frac{32w}{\sqrt{3}}+14$. We argue
that $D' = (D\setminus R) \cup H$ defines a connected dominating set of
smaller cost. By our analysis above, we see that the cost is indeed smaller,
so we are left to argue that $D'$ is connected and dominating. Notice that
$D\cap R$ can only dominate vertices that are inside $R'$, so it is
sufficient to argue that all vertices of $\graph$ are dominated. This is easy
to see because $D'$ has at least one point in each non-empty cell, and the
points in each cell form cliques. It remains to argue that $D'$ is connected.
Notice that the set of points in $R'\cap D$ that had a neighbor in $D$ which
is outside $R'$ all lie in $R'\setminus R$, so these points are part of $D'$.
So it is sufficient to argue that $V(G) \cap D'$ is connected. This follows
from the fact that $T$ is connected and the points of each cell form a clique
in $\graph$.
\qed \end{proof}

\begin{proof}[Proof of Theorem \ref{thm:widebroadcast}]
For the sake of simplicity, we start with the one sided case. It is a dynamic
programming algorithm that has subproblems for certain $2\times w$
rectangles, and for each rectangle, all the possible dominating subsets with
various connectivity constraints will be considered. More specifically, let
$k \in \mathbb{N}$, let $U\subseteq P\cap \civ{k-1,k+1}\times \civ{0,w}$,
and let $\sim$ be a binary relation on $U$. The value of the subproblem
$A(k,U,\sim)$ is the minimum size of a set $D$ of active points inside
$\civ{0,k+1}\times \civ{0,w}$ for which
\begin{itemize}
\item $D\cap \civ{k-1,k+1}\times \civ{0,w}=U$
\item $D$ dominates $\civ{0,k}\times \civ{0,w}$
\item $u_1 \sim u_2$ if and only if they are connected in the graph spanned
by $D$
\item every equivalence class of $\sim$ has a representative in
$\civ{k,k+1}\times \civ{0,w}$
\end{itemize}

By Lemma~\ref{lem:sparseDomSet}, it is sufficient to consider subproblems
where $|U|\leq \frac{32w}{\sqrt{3}}+14$. Let $\mu =\big\lfloor
\frac{32w}{\sqrt{3}}+14 \big\rfloor$. For any value of $k$, there are at most
$\binom{n}{1}+\binom{n}{2}+\dots +\binom{n}{\mu}=O(n^{\mu +1})$ such subsets.
The relevant values of $k$ are integers between $0$ and $2n$. Finally, for
any subset $U$, the number of equivalence relations on $U$ is the number of
partitions of $U$, which is the Bell number $B_{|U|}$. This can be upper
bounded by $B_{\mu}<\mu^{\mu}=w^{O(w)}$. Thus, the total number of
subproblems is $O(n^{\mu+2}w^{O(w)})=n^{O(w)}$.

For all subsets $U$ of $P \cap \civ{0,1}\times \civ{0,w}$ with size at most
$\mu$, we can compute the equivalence relation $\sim_U$. For all such sets
$U$, we define $A(0,U,\sim_U)=|U|$. For higher values of $k$, we can compute
the subproblems the following way.

When computing $A(k,U,\sim)$ (for which there is a representative of each
equivalence class of $\sim$ in $\civ{k,k+1}\times \civ{0,w}$), we first need
to find the subproblems $A(k-1,U',\sim')$ for which $U' \cap
\civ{k-1,k}\times \civ{0,w} = U \cap \civ{k-1,k}\times \civ{0,w}$. We can
only extend this subproblem if $\sim'$ is compatible with $\sim$, i.e., is
$s_1,s_2 \in U\cap U'$, then $s_1\sim' s_2 \Rightarrow s_1 \sim s_2$. We can
find these potential subproblems by going through all subproblems
$A(k-1,.,.)$, and for each of these, we can decide in polynomial time whether
it is compatible with $A(k,U,\sim)$. Overall, computing the solution of a
single subproblem takes $n^{O(w)}$ time, so finding the optimal broadcast set
in the one sided case can be done in $n^{O(w)}$ time.

For the two sided case, we need to include in the subproblem description the
set of active points on both ends. Let $k \in \mathbb{N}$, let
$U^-\subseteq P\cap \civ{-k-1,-k+1}\times \civ{0,w}$, $U^+\subseteq P\cap
\civ{k-1,k+1}\times \civ{0,w}$, and let $\sim$ be a relation on $U^-\cup
U_+$. Let $B(k,U^-,U^+,\sim)$ be the minimum size of a set $D$ of active
points inside $\civ{-k-1,k+1}\times \civ{0,w}$ for which

\begin{itemize}
\item $D\cap \civ{-k-1,-k+1}\times \civ{0,w}=U_-$ and $D\cap
\civ{k-1,k+1}\times \civ{0,w}=U_+$
\item $D$ dominates $\civ{-k,k}\times \civ{0,w}$
\item $u_1 \sim u_2$ if and only if they are connected in the graph spanned
by $D$
\item every equivalence class of $\sim$ has a representative in $\big(
\civ{-k-1,-k} \cup \civ{k,k+1} \big) \times \civ{0,w}$.
\end{itemize}

The number of subproblems is still $n^{O(w)}$, so the running time is also
$n^{O(w)}$.
\qed \end{proof}

Surprisingly, the $h$-hop version has no $n^{O(w)}$ algorithm (unless \p=\np).

\section{The hardness of $h$-hop broadcast in wide strips}

The goal of this section is to prove the following theorem.
\begin{theorem}\label{thm:widestripNPC}
The $h$-hop broadcast problem is \np-complete in strips of width $40$.
\end{theorem}
(The theorem of course refers to the decision version of the problem:
given a point set $P$, a hop bound~$h$, and a value~$K$,
does $P$ admit an $h$-hop broadcast set of size at most~$K$?)
Our reduction is from $3$-SAT. Let $x_1,x_2, \dots x_n$ be the
variables and $C_1,\dots, C_m$ be the clauses of a $3$-CNF.

\subsection{Proof overview}

Fig.~\ref{fig:bundle} shows the structural idea for representing the
variables, which we call the \emph{base bundle}. It consists of $(2h-1)n+1$
points arranged as shown in the figure, where $h$ is an appropriate value.
The distances between the points are chosen such that the graph~$\graph$,
which connects two points if they are within distance~1, consists of the
edges in the figure plus all edges between points in the same level. Thus
(except for the intra-level edges, which we can ignore) $\graph$ consists of
$n$ pairs of paths, one path pair for each variable~$x_i$. The $i$-th pair of
paths represents the variable~$x_i$, and we call it the \emph{$x_i$-string}.
By setting the target size, $K$, of the problem appropriately, we can ensure
the following for each~$x_i$: any feasible solution must use either the top
path of the $x_i$-string or the bottom path, but it cannot use points from
both paths. Thus we can use the top path of the $x_i$-path to represent a
\mytrue setting of the variable~$x_i$, and the bottom path to represent a
\myfalse setting. A group of consecutive strings is called a \emph{bundle}.
We denote the bundle containing all $x_t$-strings with $t=i,i+1,\ldots,j$
by~$\bundle(i,j)$.

\begin{figure}
\begin{center}
\begin{tikzpicture}
[
  x=0.43cm, y=0.43cm,
  string/.style={black, line width=0.5pt, rounded corners=1mm},
  sstring/.style={black, line width=1.5pt, rounded corners=1mm},
  tape/.style={black, line width=1.5pt, densely dashed},
  strip/.style={black, line width=3pt},
  branch/.style={black,pattern = crosshatch dots, pattern color=black},
  cross/.style={black,fill=black!30},
  check/.style={black,pattern=north west lines, pattern color=black},
  smallcirc/.style={draw=black,fill=white,circle,inner sep=0pt}
]
\foreach \y in {0.0,0.1,...,0.8}
   \draw (0.25,0.35) -- (1,\y);
\foreach \x in {1,2,...,25}
{
   \foreach \y in {0.0,0.1,...,0.8}
         \draw (\x,\y) -- (\x+1,\y);
}

\foreach \y in {0.05,0.25,0.45,0.65}
{
   \draw (26,\y-0.05) -- (26.95,\y);
   \draw (26,\y+0.05) -- (26.95,\y);
}
\node[smallcirc,label=above:{\large $s$}] at (0.25,0.35) {};

\draw (26.95,0.35) ellipse (0.3 and 0.6);
\node at (26.95,1.5) {\large $L_h$};

\draw[strip,dashed,Orange,fill=white](2,-0.7) rectangle (6,2.7);
\draw[strip,dashed,Orange,fill=white](7,-0.7) rectangle (11,2.7);
\draw[strip,dashed,Orange,fill=white](21,-0.7) rectangle (25,2.7);

\begin{normalsize}
\node [Orange] at (4,2) {Clause $1$};
\node [Orange] at (9,2) {Clause $2$};
\filldraw [Orange] (15,1.5) circle (0.1);
\filldraw [Orange] (15.5,1.5) circle (0.1);
\filldraw [Orange] (16,1.5) circle (0.1);
\node [Orange] at (23,2) {Clause $m$};
\end{normalsize}

\draw[strip] (0,-0.5) -- (27.2,-0.5);
\draw[strip] (0,2.5) -- (27.2,2.5);

\filldraw [Orange] (8.7,-0.9) -- (9.3,-0.9) -- (9.5,-1.5) --
++(0.5,0) -- ++(-1,-0.5) -- ++(-1,0.5) -- (8.5,-1.5) -- cycle;

\begin{scope}[shift={(0,-8)}]
\draw[strip,dashed,Orange](-0.2,-0.7) rectangle (23,5.7);
\draw[strip] (0,-0.5) -- (22.8,-0.5);
\draw[strip] (0,5.5) -- (22.8,5.5);
\draw[string] (0,-0.1) -- (22.8,-0.1);
\draw[string] (0,0) -- (22.8,0);
\draw[string] (0,0.1) -- (22.8,0.1);
\draw[string] (0,0.2) -- (22.8,0.2);
\draw[string,Green] (0,0.3) -- (12.8,0.3) -- (12.8,3.3) -- (13.5,3.3) --
(13.5,2.5);
\draw[string] (0,0.4) -- (12.5,0.4) -- (12.5,3.4) -- (15.5,3.4) -- (15.5,2.5);
\draw[string,blue] (0,0.5) -- (5.8,0.5) -- (5.8,3.5) -- (6.5,3.5) --
(6.5,2.5);
\draw[string,red] (0,0.6) -- (1,0.6) -- (1,3.6) -- (2,3.6) -- (2,2.5);
\draw[branch] (1.5,1.5) rectangle (2.5,2.5);
\draw[string,red] (2,1.5) -- (2,0.6) -- (5.5,0.6) -- (5.5,3.6) -- (8.5,3.6) --
(8.5,2.5);
\draw[tape,red] (2.5,1.7) -- (3.5,1.7);
\draw[tape,red] (2.5,2.3) -- (3.5,2.3);
\draw[string] (0,0.7) -- (0.7,0.7) -- (0.7,3.7) -- (4,3.7) -- (4,2.5);
\draw[cross] (3.5,1.5) rectangle (4.5,2.5);
\draw[string] (4,1.5) -- (4,0.7) -- (5.2,0.7) -- (5.2,3.7) -- (10.5,3.7) --
(10.5,2.5);
\draw[sstring,red] (4.5,2) -- (4.8,2) -- (4.8,5) -- (22,5);
\draw[branch] (6,1.5) rectangle (7,2.5);
\draw[tape,blue] (7,1.7) -- (8,1.7);
\draw[tape,blue] (7,2.3) -- (8,2.3);
\draw[cross] (8,1.5) rectangle (9,2.5);
\draw[tape,blue] (9,1.7) -- (10,1.7);
\draw[tape,blue] (9,2.3) -- (10,2.3);
\draw[cross] (10,1.5) rectangle (11,2.5);
\draw[sstring,blue] (11,2) -- (11.3,2) -- (11.3,4.4) -- (22.2,4.4) --
(22.2,4.8);
\draw[string,blue] (6.5,1.5) -- (6.5,0.5) --
                    (12.2,0.5) -- (12.2,3.5) -- (17.5,3.5) -- (17.5,2.5);
\draw[branch] (13,1.5) rectangle (14,2.5);
\foreach \x in {14,16,18,20}
{
   \draw[tape,Green] (\x,1.7) -- (\x+1,1.7);
   \draw[tape,Green] (\x,2.3) -- (\x+1,2.3);
   \draw[cross] (\x+1,1.5) rectangle (\x+2,2.5);
}
\draw[sstring,Green] (22,2) -- (22.8,2) -- (22.8,5) -- (22.4,5);
\draw[check] (22,4.8) rectangle (22.4,5.2);
\draw[string,Green] (13.5,1.5) -- (13.5,0.3) -- (22.8,0.3);
\draw[string,red] (8.5,1.5) -- (8.5,0.6) --
                    (12,0.6) -- (12,3.6) -- (19.5,3.6) -- (19.5,2.5);
\draw[string] (10.5,1.5) -- (10.5,0.7) --
                    (11.8,0.7) -- (11.8,3.7) -- (21.5,3.7) -- (21.5,2.5);
\draw[string] (15.5,1.5) -- (15.5,0.4) -- (22.8,0.4);
\draw[string,blue] (17.5,1.5) -- (17.5,0.5) -- (22.8,0.5);
\draw[string,red] (19.5,1.5) -- (19.5,0.6) -- (22.8,0.6);
\draw[string] (21.5,1.5) -- (21.5,0.7) -- (22.8,0.7);
\node[draw=black,thick,rounded corners=2pt,below left=2mm]
 at (28.8,6.3) { \scriptsize
\begin{tabular}{@{}r@{ }l@{}}
 \raisebox{2pt}{\tikz{\draw [string] (0,0) -- (3mm,0);}}&string\\
 \raisebox{2pt}{\tikz{\draw [tape] (0,0) -- (3mm,0);}}&tape\\
 \raisebox{2pt}{\tikz{\draw [sstring] (0,0) -- (3mm,0);}}&side string\\
 \raisebox{2pt}{\tikz{\draw [branch,rounded corners=0mm]
 (0,0) rectangle (3mm,3mm);}} & branching\\
 \raisebox{2pt}{\tikz{\draw [cross,rounded corners=0mm]
 (0,0) rectangle (3mm,3mm);}} & crossing\\
 \raisebox{2pt}{\tikz{\draw [check,rounded corners=0mm]
 (0,0) rectangle (3mm,3mm);}} & clause check
\end{tabular}};
\end{scope}
\end{tikzpicture}
\end{center}
\caption{The overall construction, and the way a single clause is checked.
Note that in this figure each string (which actually consists of two paths)
is shown as a single curve.}\label{fig:clausecheck}
\end{figure}
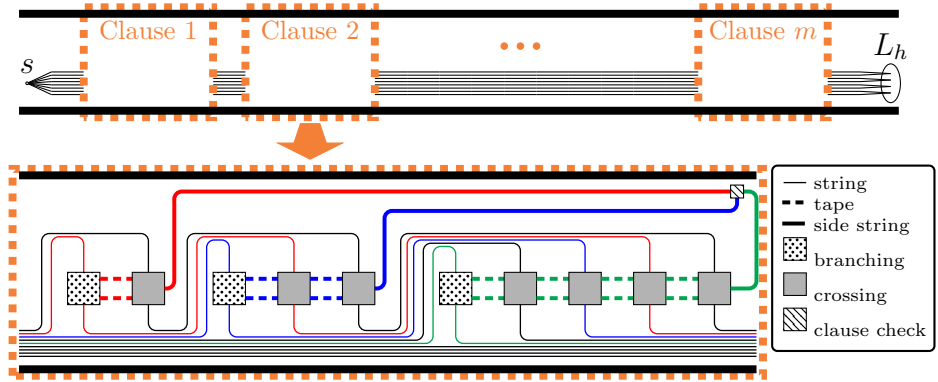

The clause gadgets all start and end in the base bundle, as shown in
Fig.~\ref{fig:clausecheck}. The gadget to check a clause involving variables
$x_i,x_j,x_k$, with $i<j<k$, roughly works as follows; see also the lower
part of Fig.~\ref{fig:clausecheck}, where the strings for $x_i$, $x_j$, and
$x_k$ are drawn in red, blue, and green respectively.

First we split off~$\bundle(1,i-1)$ from the base bundle, by letting the top
$i-1$ strings of the base bundle turn left. (In Fig.~\ref{fig:clausecheck}
this bundle consists of two strings.) We then separate the $x_i$-string from
the base bundle, and route the $x_i$-string into a \emph{branching gadget}.
The branching gadget creates a branch consisting of two \emph{tapes}---this
branch will eventually be routed to the \emph{clause-checking gadget}---and a
branch that returns to the base bundle. Before the tapes can be routed to the
clause-checking gadget, they have to cross each of the strings
in~$\bundle(1,i-1)$. For each string that must be crossed we introduce a
\emph{crossing gadget}. A crossing gadget lets the tapes continue to the
right, while the string being crossed can return to the base bundle. The
final crossing gadget turns the tapes into a \emph{side string} that can now
be routed to the clause-checking gadget. The construction guarantees that the
side string for $x_i$  still carries the truth value that was selected for
the $x_i$-string in the base bundle. Moreover, if the \mytrue path
(resp.~\myfalse path) of the $x_i$-string was selected to be part of the
broadcast set initially, then the \mytrue path (resp.~\myfalse path) of the
rest of the $x_i$-string that return to the base bundle must be in the
minimum broadcast set as well.

After we have created a side string for $x_i$, we create side strings for
$x_j$ and $x_k$ in a similar way. The three side strings are then fed into
the clause-checking gadget. The clause-checking gadget is a simple
construction of four points. Intuitively, if at least one side string
carries the correct truth value---\mytrue if the clause contains the positive
variable, \myfalse if it contains the negated variable---, then we
activate a single disk in the clause check gadget that corresponds to a true literal. Otherwise we need
to change truth value in at least one of the side strings, which requires an
extra disk.

The final construction contains~$\Theta(n^4m)$ points that all
fit into a strip of width~40.

In order to simplify our discussion and figures, we scale the input
such that $a$ can broadcast to $b$ if their unit disks intersect (or equivalently, if their
distance is at most $2$).

\subsection{Handling strings and bundles}
We start the initial bundle directly from the source, and end each string
with a disk that intersects the last true and false disk of the given
variable, as already seen in Fig.~\ref{fig:bundle}.
(A \mytrue disk is a disk on a \mytrue path, a \myfalse disk is a disk on a \myfalse path.)
A minimum-size solution
of this bundle for $h=7$ contains the source disk and true or false disks for
each of the 3 strings. In the final construction, once all the clause checks
are done and the strings have returned to the bottom bundle, we are going to
add some extra levels so that the $h$-hop restriction does not interfere with
the last side strings. (This can be done by for example doubling the maximum
distance from $s$.) The disks of a given level in a bundle lie on the same
vertical line, at distance $\frac{1}{2n}$ from each other, so for a bundle
containing all the variables, the disk centers on a given level fit on a
vertical segment of length $1$, and the whole bundle fits in width $3$.

\begin{figure}
\begin{center}
\begin{tikzpicture}[x=0.8cm,y=0.8cm]
\foreach \x in {0.0,1.95,...,7.0}
{
   \tf{\x,0.8}{\x,0}
}
\end{tikzpicture}
\end{center}
\caption{Disk pairs of a string.}\label{fig:string}
\end{figure}
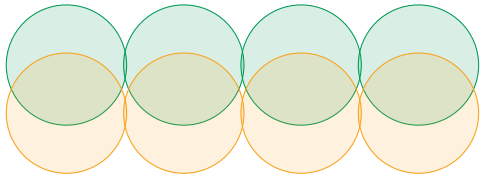

Bundled strings are in lockstep, i.e., a pair of intersecting disks in
the bundle that are not in the same string and truth value are on the same
level. We call this the \emph{lockstep condition}.

Next, we describe some important aspects of handling strings, bundles and
side strings. First, we show that we can do turns with strings in constant
horizontal space, and do turns in bundles in polynomial horizontal space. An
example of a string turn can be seen in Fig.~\ref{fig:stringturn}.

\begin{figure}
\begin{center}
\begin{tikzpicture}[x=0.8cm,y=0.8cm]
\turnleft{0,0}{0}
\end{tikzpicture}
\end{center}
\caption{A turn of $90\deg$ in a side string or string outside a bundle using
constant horizontal space.}\label{fig:stringturn}
\end{figure}
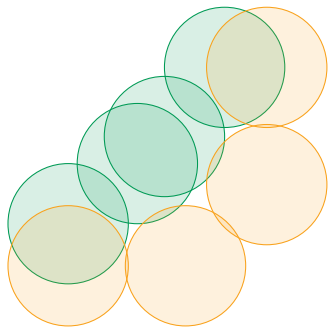

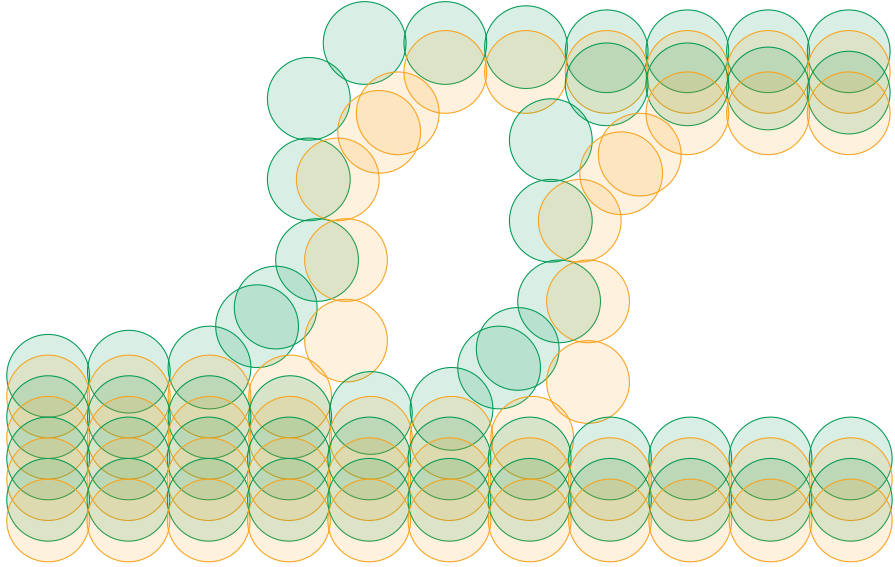
\begin{figure}
\begin{center}
\begin{tikzpicture}[x=0.55cm,y=0.55cm]

\tf{-3.9,0.5}{-3.9,0}
\tf{-1.95,0.6}{-1.95,0}
\turnleft{0,0}{0}
\turnright{5.7,8.55}{180}
\tf{7.65,8.45}{7.65,7.85}
\tf{9.6,8.35}{9.6,7.85}
\tf{11.55,8.35}{11.55,7.85}
\tf{13.5,8.35}{13.5,7.85}
\tf{15.45,8.35}{15.45,7.85}

\begin{scope}[shift={(5.85,-1)}]
\tf{-9.75,0.5}{-9.75,0}
\tf{-7.8,0.5}{-7.8,0}
\tf{-5.85,0.5}{-5.85,0}
\tf{-3.9,0.5}{-3.9,0}
\tf{-1.95,0.6}{-1.95,0}
\turnleft{0,0}{0}
\turnright{5.7,8.55}{180}
\tf{7.65,8.45}{7.65,7.85}
\tf{9.6,8.35}{9.6,7.85}
\end{scope}

\foreach \x in {-3.9,-1.96,...,15.6}
{
   \tf{\x,-1.5}{\x,-2}
   \tf{\x,-2.5}{\x,-3}
}
\end{tikzpicture}
\end{center}
\caption{Splitting the top 2 strings off a bundle of 4 strings: we peel the
top layers one by one.}\label{fig:bundleturn}
\end{figure}

This turning operation can be used on the top string of a bundle to ``peel''
off strings one by one and unify them later in a new bundle, see
Fig.~\ref{fig:bundleturn}. This is how we can split and turn a bundle: we
peel and turn the strings one by one. Notice that the lockstep condition is
upheld both in the bottom and top bundle. It requires $O(n)$ extra horizontal
space and $O(n^2)$ disks to split a bundle with this method.

If we were to return the strings one by one to the bottom bundle without
correction as depicted in Fig.~\ref{fig:clausecheck}, the returning strings
would be in a level disadvantage compared to the bottom bundle, so the new
bundle would violate the lockstep condition. To avoid this issue, we use a
correction mechanism. We have some room to squeeze bundle levels
horizontally. The largest horizontal distance between neighboring levels is
$2$; for the smallest distance, we need to make sure that a disk does
not intersect other disks from neighboring levels other than the disks in the
same string with the same truth value. So the horizontal distance has to be
at least $2\sqrt{1-\left(\frac{1}{4n}\right)^2}<2-\frac{1}{15n^2}$. Thus, if
we have $15n^2$ compressed levels in a bundle, then they take up the same
horizontal space as $15n^2-1$ maximum distance levels.

A detour of a string (peeling off, going through a gadget, returning to the
bottom bundle) requires a constant number of extra levels to achieve, we can
compensate for this with the addition of a polynomial number of extra disks.
Before a string peels off from the top bundle downward to rejoin the bottom
bundle, we add $15n^2k$ compressed levels to the top bundle and $(15n^2-1)k$
maximum distance levels to the bottom bundle, if the total number of extra
levels added by turning up, going through the gadget and turning down is $k$.
This ensures that the lockstep condition is upheld in the bottom bundle after
the return of this string. For each string that leaves the bottom bundle and
later returns, we use this correction mechanism. Overall, this correction
mechanism is invoked a polynomial number of times, so requires a polynomial
number of disks.

\subsection{Tapes}

Our tapes consist of tape blocks: a tape block is a collection of three
disks, the centers of which lie on a line at distance $\epsilon$ apart -- so
it is isometric to the old connector blocks $A,B,C$ and $D$ for the case
``$n$''$=2$ (see Section~\ref{sec:AppHardness}). Denote the three disks
inside a tape block $T^k$ by $\delta^k_1,\delta^k_2$ and $\delta^k_3$. We can
place multiple such blocks next to each other to form a tape. An example is
depicted in Fig.~\ref{fig:tape}.

\begin{figure}
\begin{center}
\begin{tikzpicture}[x=0.65cm,y=0.65cm]
\begin{scope}[rotate around={200:(0,0)}]
\tf{-2,0}{-1.9,0}
\tapeblock{0,0}{0}
\tapeblock{2,0}{0}
\tapeblock{3.87,0.7}{50}
\tapeblock{5.85,0.95}{-40}
\tapeblock{6.9,-0.76}{-90}
\tf{6.9,-2.86}{6.9,-2.96}
\node at (-1.95,0) {G};
\node at (6.9,-2.91) {F};
\node at (0,0){$T_p$};
\node at (3.87,0.7){$T_k$};
\node at (6.9,-0.76){$T_1$};
\end{scope}
\end{tikzpicture}
\hfill
\begin{tikzpicture}[x=0.9cm,y=0.9cm,
   every node/.style={draw=black,fill=white,circle,inner sep=2pt}]
\node [label={[label distance=0.1cm]above:True}, draw=ForestGreen,
fill=ForestGreen!15] (T) at (0,1.5) {};
\node [label={[label distance=0.1cm]below:False}, draw=YellowOrange,
fill=YellowOrange!15] (F) at (0,0.5) {};
\foreach \x in {1,2,...,5}
{
   \node [draw=RedViolet,fill=RedViolet!30] (V\x0) at (\x,0) {};
   \node [draw=RedViolet,fill=RedViolet!30] (V\x1) at (\x,1) {};
   \node [draw=RedViolet,fill=RedViolet!30] (V\x2) at (\x,2) {};
   \draw (V\x0) -- (V\x1) -- (V\x2);
   \draw (V\x0) to [bend left=20] (V\x2);
}
\draw (V11) -- (T) -- (V10);
\draw (T) -- (F) -- (V10);
\foreach \x in {2,3,...,5}
{
   \pgfmathtruncatemacro{\xp}{\x-1};
   \draw (V\x0) -- (V\xp0);
   \draw (V\x1) -- (V\xp1);
   \draw (V\x2) -- (V\xp2);
   \draw (V\x0) -- (V\xp1);
   \draw (V\x0) -- (V\xp2);
   \draw (V\x1) -- (V\xp2);
}
\node [label={[label distance=0.1cm]below:False},draw=YellowOrange,
fill=YellowOrange!15] (X0) at (6,0.5) {};
\node [label={[label distance=0.1cm]above:True},draw=ForestGreen,
fill=ForestGreen!15] (X1) at (6,1.5) {};
\draw (X0) -- (X1);
\draw (V52) -- (X1);
\draw (V52) -- (X0) -- (V51);
\draw [blue] (-0.3,0.2) rectangle (0.3,1.8);
\foreach \x in {1,2,...,5}
{
   \pgfmathsetmacro{\xm}{\x-0.3};
   \pgfmathsetmacro{\xp}{\x+0.3};
   \draw [RedViolet] (\xm,-0.3) rectangle (\xp,2.3);
}
\draw [blue] (5.7,0.2) rectangle (6.3,1.8);
\node [draw=none, fill=none] at (1,-0.8) {$1$};
\node [draw=none, fill=none] at (2,-0.8) {$2$};
\node [draw=none, fill=none] at (3,-0.8) {$3$};
\node [draw=none, fill=none] at (4,-0.8) {$4$};
\node [draw=none, fill=none] at (5,-0.8) {$5$};
\begin{scope}[every node/.style={draw=none,fill=none,rectangle}]
\node at (6.7,1) {G};
\node at (-0.7,1) {F};
\node at (5,2.7){$T_p$};
\node at (3,2.7){$T_k$};
\node at (1,2.7){$T_1$};
\end{scope}
\end{tikzpicture}
\end{center}
\caption{Tape blocks connecting two true-false disk pairs and the
corresponding subgraph.} \label{fig:tape}
\end{figure}
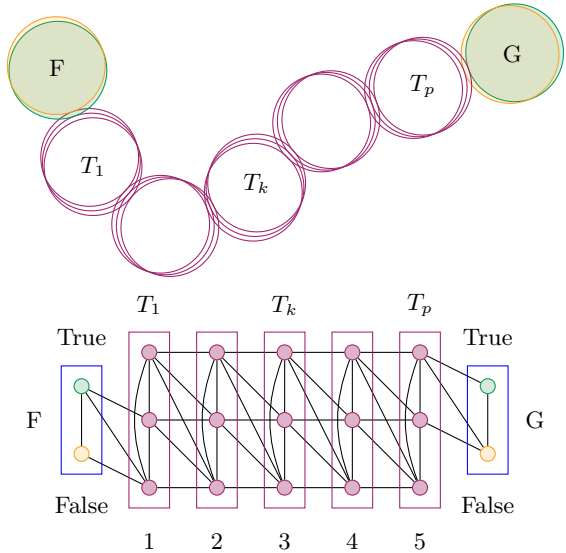

The tapes always connect blocks in which disks have truth values assigned,
e.g., the end of a string or disks of a gadget block. Denote the starting
true and false disks by $F$ and the ending true and false disks by $G$. We
say that a set of tape blocks $T^1,T^2,\dots, T^p$ forms a tape from $F$ to
$G$ if it satisfies the following conditions.
\begin{itemize}
\item In the first block, $\delta^1_1$ intersects both the true and false
disk(s) of $F$, $\delta^1_2$ intersects the true disk(s) of $F$, and
$\delta^1_3$ is disjoint from both the true and false disk(s).
\item $\delta^k_i$ intersects the disk $\delta^{k+1}_j$ if and only if $j\leq
i$ $(k=1,\dots,p-1)$.
\item In the last block, $\delta^p_1$ is disjoint from $G$, $\delta^p_2$
intersects the false disk(s) of $G$, and $\delta^p_3$ intersects both the
true and false disk(s).
\item Non-neighboring tape blocks are disjoint, $F$ is disjoint from all
blocks except the first, and $G$ is disjoint from all blocks except the last.
\end{itemize}

We would like to examine the set of disks that are used in a minimum
broadcast set from a tape.

\begin{lemma}\label{lem:tape}
Let $T$ be a tape from $F$ to $G$ that has $p$ tape blocks. Every $h$-hop
broadcast set contains at least $p-1$ disks from the tape. If a broadcast set
contains exactly $p-1$ disks, then the truth value of $F$ is is at least the
truth value of $G$, i.e., it cannot happen that the active disks in $F$ are
all false disks and the active disks in $G$ are all true disks.
\end{lemma}

\begin{proof}
Let the tape blocks be $T^1,T^2, \dots, T^p$. If there are at most $p-2$
active disks, then there are at least two empty blocks. These blocks have to
be neighboring, otherwise a point in between the two blocks is impossible to
reach from the source. Let these blocks be $T^k$ and $T^{k+1}$. All disks in
$T_k$ must be reached through the blocks $F,T^1,\dots,T^{k-1}$. Specifically,
$\delta^k_3$ has to be reached. The shortest path to this point from any
$F$-disk requires at least $k$ tape disks. Similarly, the shortest path from
any $G$-disk to $\delta^{k+1}_1$ requires at least $t-k$ disks. Overall, at
least $t$ active disks of the tape are required to reach these disks -- this
is a contradiction.

If the tape contains $t-1$ active disks, then both $F$ and $G$ must contain
an active disk, otherwise there would be a component inside the tape that is
not connected to the source. Suppose for contradiction that the active disks
of $F$ are false and the active disks of $G$ are true. There is at least one
tape block that has no active disk; let $T^k$ be such a block, where $k$ is
as small as possible. Since $\delta^k_2$ has to be covered, it has to be
reached either from $F$ or $G$.

Suppose that $\delta^k_2$ is reached through $F$; this requires $k$ active
disks from the tape blocks. We have only $p-1-k$ active disks for the rest of
the $p-k$ blocks $T^{k+1},\dots, T^p$, so there has to be another empty tape
block, $T^{\ell} \;(\ell > k)$. As previously demonstrated, we cannot have
non-neighboring empty blocks, so the other empty block is $T^{k+1}$. This
means that $\delta^k_3$ also has to be dominated from the left side, the
shortest path to which requires $k+1$ active tape disks from any false disk
of $F$. This leads to an additional empty block among $T^{k+1},\dots, T^p$.
But as shown above, there can be at most one such block ($T^{k+1}$) -- we
arrived at a contradiction. The same argument works for the case when
$\delta^k_2$ is reached from $G$.
\qed \end{proof}

\subsection{Gadgets and their connection to tapes and strings}

\mypara{Crossing and branching gadgets.} Our crossing gadget and our
branching gadget are almost identical to the one used in the \Wone-hardness
proof of \cdsudg. This gadget can be used to transmit information both
horizontally and vertically -- this is exactly what we need. Since we only
need to transmit truth values, we take the gadget for ``$n$''$=2$, resulting
in $X$ blocks with $2\cdot 2$ and $Y$ -blocks with $2\cdot 2 + 1$ disks. The
only change we make in the crossing gadget is that we swap the $X_1$ and
$X_2$ blocks.

For the branching gadget, we modify some offsets so that we can transmit the
vertical truth value on the right side of our gadget. For this purpose, we
redefine the offsets in the following right side $X$-blocks.

\[\text{offset}(X_3(j))=(-\iota_2(j),-j) \qquad
\text{offset}(X_4(j))=(\iota_2(j),-j)\]

In case of these horizontal connections, we say that a disk $X_k(j)$ from the
block $X_k$ is a true disk if $\iota_1(j)=2$ and a false disk if
$\iota_1(j)=1$. Similarly, for vertical connections, a disk $X_{\ell}(j)$ is
a true disk if $\iota_2(j)=2$ and a false disk if $\iota_2(j)=1$.

\mypara{Connecting gadgets with tapes and strings.}
When connecting branching and crossing gadgets or two crossing gadgets with
tapes horizontally, we are going to add a tape that goes from the $X_4$ block
of the left gadget to the $X_7$ block of the right gadget, and a tape that
goes from the $X_8$ block of the right gadget to the $X_3$ block of the left
gadget. Note that in the \Wone-hardness proofs, we used the same strategy with
tapes consisting of only one block. In this case, we place the first and last
block of each tape at the same location as the connector block in the proof
of Theorem~\ref{thm:W1-hardness}, and use some tape blocks in between these,
the number of which will be specified later. Note that this placement gives
us a tape that is consistent with the definition of true and false disks in
the $X$-blocks.

In order to connect strings and side strings to the gadgets, we use both
tapes and parent-leaf pairs.  Fig.~\ref{fig:gadgetwithear} depicts a
connection to a crossing gadget from the top and bottom.

\begin{figure}
\begin{center}
\begin{tikzpicture}[x=0.3cm,y=0.3cm,
tape/.style={RedViolet,pattern color=RedViolet,pattern=dots,very thick}]
\begin{footnotesize}
\blocksswap
\end{footnotesize}
\foreach \x in {0,4,8,12,16}
   \foreach \y in {0,16}
      {\draw [ybi] (\x,\y) circle (1);
      \draw [ybi] (\y,\x) circle (1);}
\foreach \x in {2,6,10,14}
   \foreach \y in {0,16}
      {\draw [xbi] (\x,\y) circle (1);
      \draw [xbi] (\y,\x) circle (1);}
\begin{scope}[rotate around={{-90}:(8,8)},draw=RedViolet]
\draw [tape] (-2,2) circle (1);
\draw [tape] (-2,14) circle (1);
\draw [tape] (-4,2) circle (1);
\draw [tape] (-4,14) circle (1);
\draw [tape] (-5.95,2.2) circle (1);
\draw [tape] (-5.95,13.8) circle (1);
\foreach \y in {4.2,7.8,9.8,11.8}
   \draw [tape] (-6,\y) circle (1);
\turnright{-11.9,6.3}{180}
\tf{-9.9,6.3}{-9.9,5.6}
\tf{-8,6.2}{-8,5.75}
\tf{-6,6.1}{-6,5.9}
\cc{-2,6}{-2,7}
\cc{-4,6}{-3.8,4.7}
\end{scope}
\begin{scope}[rotate around={{90}:(8,8)},draw=RedViolet]
\draw [tape] (-2,2) circle (1);
\draw [tape] (-2,14) circle (1);
\draw [tape] (-4,2) circle (1);
\draw [tape] (-4,14) circle (1);
\draw [tape] (-5.95,2.2) circle
(1);
\draw [tape] (-5.95,13.8) circle
(1);
\foreach \y in {4.2,7.8,9.8,11.8}
   \draw [tape] (-6,\y)
   circle (1);
\turnleft{-11.9,6.3}{180}
\tf{-9.9,5.6}{-9.9,6.3}
\tf{-8,5.75}{-8,6.2}
\tf{-6,5.9}{-6,6.1}
\cc{-2,6}{-2,7}
\cc{-4,6}{-3.8,4.7}
\end{scope}
\pgfmathsetmacro{\de}{0.05}
\pgfmathsetmacro{\shift}{2-\de*\de-\de*\de*\de*\de}
\pgfmathsetmacro{\base}{2-\de}
\foreach \x in {0,90,180,270}
	{\sevencircles{\x}}
\ccc{6,6}{6-\de,6}
\ccc{8,6}{8,6+4*\de}
\ccc{10,6}{10,6-\de}
\ccc{10,8}{10-4*\de,8}
\ccc{10,10}{10,10+\de}
\ccc{8,10}{8,10-4*\de}
\ccc{6,10}{6,10+\de}
\ccc{6,8}{6+4*\de,8}
\node[draw=black,thick,rounded corners=2pt,below left=2mm]
 at (34,35) {
\begin{tabular}{@{}c@{ }l@{}}
 \raisebox{2pt}{\tikz{\cc{0,0}{0.4,0} \ccc{2.7,0}{3.1,0}}}&parent-leaf pair\\
 \raisebox{2pt}{\tikz{\draw [tape, RedViolet, pattern
 color=RedViolet,pattern=dots] (0,0) circle (1);}}&tape block\\
 \raisebox{2pt}{\tikz{\tf{0.4,0}{0,0}}}&string block\\
 \raisebox{2pt}{\tikz{\draw [xbi] (0,0) circle (1);
 \draw [ybi] (2.9,0) circle (1);}}&$X$- and $Y$-block
\end{tabular}};
\end{tikzpicture}
\end{center}
\caption{Connecting to a crossing or branching gadget from the top and
bottom.} \label{fig:gadgetwithear}
\end{figure}
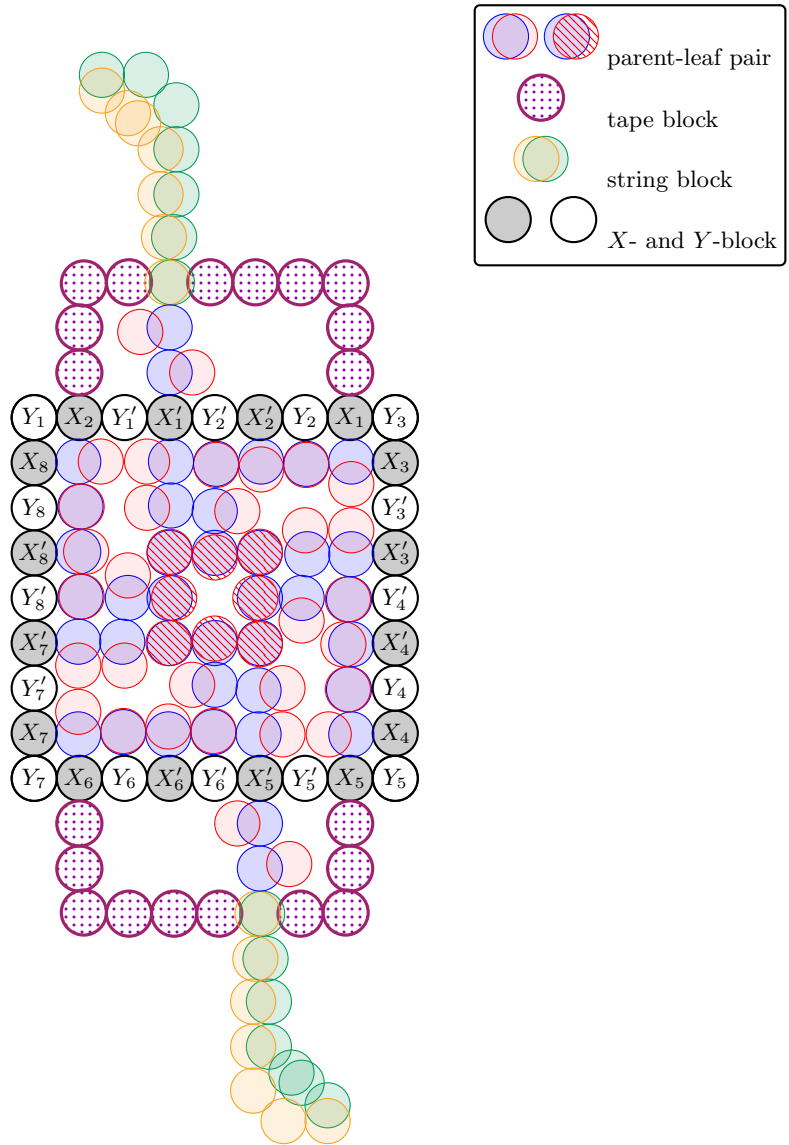

We need to connect both ``sides'' of the string: on the top, we use a tape
from $X_2$ to the last string block, and a tape from the last string block to
$X_1$. Moreover, in order to make sure that all the disk pairs of the strings
are in use, and the connection is not maintained through the tapes, we create
a  short path to the gadget with some disks that are guaranteed to be in the
solution. This path consists of parent and leaf disk pairs, where all the
parents will be inside a canonical solution -- we used this technique before
inside the gadgets to ensure gadget connectivity. The string exits the gadget
similarly. Note that the shortest path through the gadget from the string end
on the top to the string end on the bottom has length 18, and its internal
vertices are all parent disks, a disk from $X'_1$ and a disk from $X'_5$; the
paths using any of these tapes are longer.

We use the same type of connection to connect side strings to the right side
of the last crossing gadget (or to the branching gadget, if the current
clause contains the first variable). The complete gadget together with the
connections and string turns fits in $50$ units of vertical space. (Recall
that all distances have been scaled by a factor of two, so that we have
unit radius disks.)

We briefly return to the tape pairs that connect neighboring blocks. We need
to make sure that the tapes do not provide a shortcut -- we want the shortest
path from source to the last level $h$ to be through string blocks, and to go
through gadgets as discussed above. When choosing a tape length, we also need
to bridge the distance between neighboring gadgets. Note that this amount can
be polynomial in $n$ because of the correction mechanism for strings. We add
a small detour to make sure that the shortest path to a gadget that uses a
tape is longer than the shortest path that uses only the string that enters
the gadget. It is easy to see that there is enough place for such a detour:
taking twice the amount of blocks that would be necessary to cover the
distance is enough. A tape connection between neighboring blocks is depicted
in Fig.~\ref{fig:twotapes}. (Note that these tapes need no additional
vertical space: they fit easily in the $18$ units of vertical space between
the gadgets.)

\begin{figure}
\begin{center}
\begin{tikzpicture}[x=0.16cm,y=0.16cm,
tape/.style={RedViolet,pattern color=RedViolet,pattern=dots,very thick}]
\begin{tiny}
\blocksswap
\foreach \x in {0,4,8,12,16}
   \foreach \y in {0,16}
   {
      \draw [ybi] (\x,\y) circle (1);
      \draw [ybi] (\y,\x) circle (1);
   }
\foreach \x in {2,6,10,14}
   \foreach \y in {0,16}
   {
      \draw [xbi] (\x,\y) circle (1);
      \draw [xbi] (\y,\x) circle (1);
   }
\begin{scope}[shift={(56,0)}]
  \blocksswap
  \foreach \x in {0,4,8,12,16}
     \foreach \y in {0,16}
     {
        \draw [ybi] (\x,\y) circle (1);
        \draw [ybi] (\y,\x) circle (1);
     }
  \foreach \x in {2,6,10,14}
     \foreach \y in {0,16}
     {
        \draw [xbi] (\x,\y) circle (1);
        \draw [xbi] (\y,\x) circle (1);
     }
\end{scope}
\foreach \x in {18,20,...,30,34,36,...,54}
   \draw [tape] (\x,14) circle (1);
\foreach \x in {20,22,...,30}
   \draw [tape] (\x,10) circle (1);
\foreach \x in {20,22,...,34}
   \draw [tape] (\x,6) circle (1);
\foreach \y in {8,10,12}
   \draw [tape] (34,\y) circle (1);
\draw [tape] (20,8) circle (1);
\draw [tape] (30,12) circle (1);
\begin{scope}[rotate around={{180}:(36,8)}]
\foreach \x in {18,20,...,30,34,36,...,54}
   \draw [tape] (\x,14) circle (1);
\foreach \x in {20,22,...,30}
   \draw [tape] (\x,10) circle (1);
\foreach \x in {20,22,...,34}
   \draw [tape] (\x,6) circle (1);
\foreach \y in {8,10,12}
   \draw [tape] (34,\y) circle (1);
\draw [tape] (20,8) circle (1);
\draw [tape] (30,12) circle (1);
\end{scope}
\end{tiny}
\end{tikzpicture}
\end{center}
\caption{Connecting neighboring gadgets with tapes.} \label{fig:twotapes}
\end{figure}
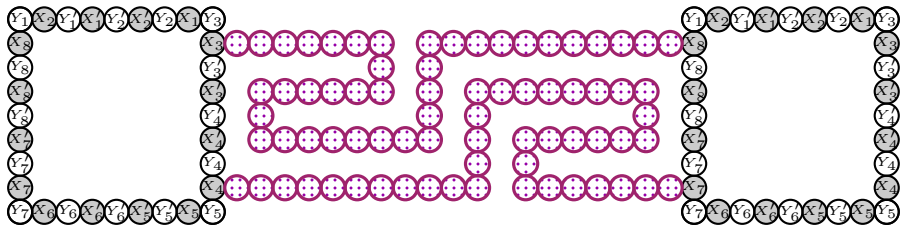

\mypara{The clause check gadget.} The clause check gadget is very simple, it
contains four well-placed disks: one at the end of each of the three side
strings, and one disk that only intersects the three other clause check
disks. We turn the three side strings towards their corresponding disks so
that the side strings do not interfere with each other. Among the six last
disks at the end of the three side strings only the ones corresponding to
the literals of this clause intersect the gadget. The rest of the side string
disks are disjoint from the gadget. See Fig.~\ref{fig:clausecheckgadget}
for an example of checking $(x_2 \vee x_3 \vee \bar{x}_5)$. The vertical
space required is less than 20 units.

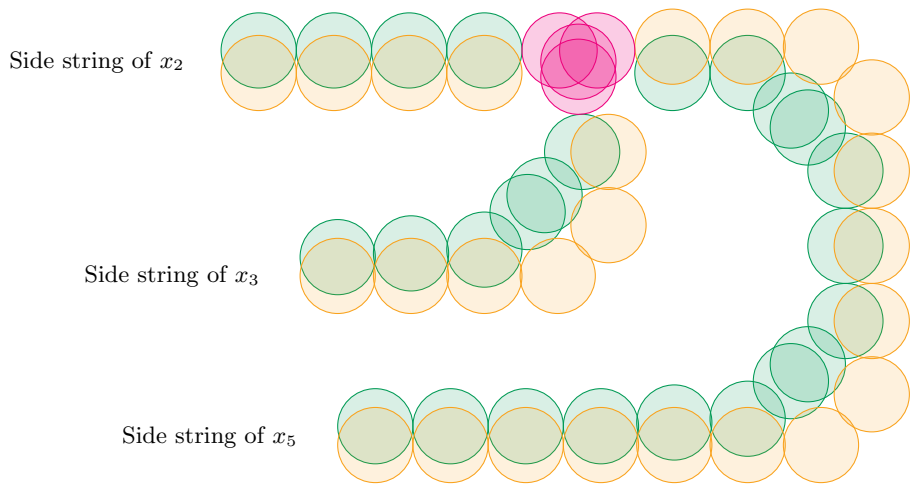
\begin{figure}
\begin{center}
\begin{tikzpicture}[x=0.5cm,y=0.5cm,
   every node/.style={draw=black,fill=white,circle,inner sep=2pt}]
\tf{-9.9,0.5}{-9.9,0}
\tf{-7.9,0.5}{-7.9,0}
\tf{-5.9,0.5}{-5.9,0}
\tf{-3.9,0.5}{-3.9,0}
\tf{-1.95,0.6}{-1.95,0}
\turnleft{0,0}{0}
\tf{2.6,5.3}{3.3,5.3}
\turnleft{3.3,7.3}{90}
\tf{-2,9.9}{-2,10.6}

\begin{scope}[shift={(-7,4.5)}]
\tf{-3.9,0.5}{-3.9,0}
\tf{-1.95,0.6}{-1.95,0}
\turnleft{0,0}{0}
\end{scope}

\begin{scope}[shift={(-13,10.5)}]
\tf{0,0}{0,-0.6}
\tf{2,0}{2,-0.6}
\tf{4,0}{4,-0.6}
\tf{6,0}{6,-0.6}
\end{scope}
\draw [RubineRed, fill=RubineRed, fill opacity=0.2] (-4,10.5) circle (1);
\draw [RubineRed, fill=RubineRed, fill opacity=0.2] (-5,10.5) circle (1);
\draw [RubineRed, fill=RubineRed, fill opacity=0.2] (-4.5,9.8) circle (1);
\draw [RubineRed, fill=RubineRed, fill opacity=0.2] (-4.5,10.2) circle (1);
\begin{small}
\node [draw=none, fill=none] at (-14.3,0.2) {Side string of $x_5$};
\node [draw=none, fill=none] at (-15.3,4.5) {Side string of $x_3$};
\node [draw=none, fill=none] at (-17.3,10.2) {Side string of $x_2$};
\end{small}
\end{tikzpicture}
\end{center}
\caption{Clause check gadget for the clause $(x_2 \vee x_3 \vee \bar{x}_5)$.}
\label{fig:clausecheckgadget}
\end{figure}

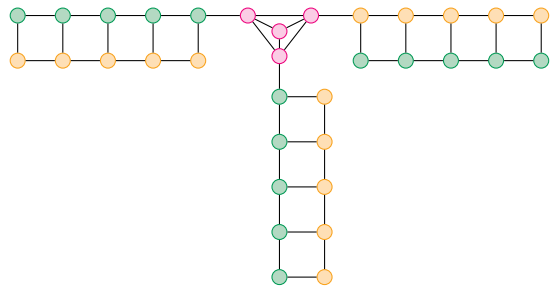
\begin{figure}
\begin{center}
\begin{tikzpicture}[x=0.6cm,y=0.6cm,
   every node/.style={draw=black,fill=white,circle,inner sep=2pt}]

\node [RubineRed, fill=RubineRed, fill opacity=0.2] (A) at (-0.2,0) {};
\node [RubineRed, fill=RubineRed, fill opacity=0.2] (C) at (1.2,0) {};
\node [RubineRed, fill=RubineRed, fill opacity=0.2] (B) at (0.5,-0.9) {};
\node [RubineRed, fill=RubineRed, fill opacity=0.2] (D) at (0.5,-0.35) {};
\draw (A) -- (B) -- (C) -- (A) -- (D) -- (B);
\draw (C) -- (D);

\foreach \x in {1,2,...,5}
{
   \node [draw=ForestGreen,fill=ForestGreen!30] (AT\x) at (-\x-0.3,0) {};
   \node [draw=YellowOrange,fill=YellowOrange!30] (AF\x) at (-\x-0.3,-1) {};
   \draw (AT\x) -- (AF\x);
   \node [draw=ForestGreen,fill=ForestGreen!30] (CT\x) at (\x+1.3,-1) {};
   \node [draw=YellowOrange,fill=YellowOrange!30] (CF\x) at (\x+1.3,0) {};
   \draw (CT\x) -- (CF\x);
   \node [draw=ForestGreen,fill=ForestGreen!30] (BT\x) at (0.5,-\x-0.8) {};
   \node [draw=YellowOrange,fill=YellowOrange!30] (BF\x) at (1.5,-\x-0.8) {};
   \draw (BT\x) -- (BF\x);
}
\draw (A)--(AT1);
\draw (B)--(BT1);
\draw (C)--(CF1);
\foreach \x in {2,3,...,5}
{
   \pgfmathtruncatemacro{\xp}{\x-1};
   \draw (AT\x) -- (AT\xp);
   \draw (BT\x) -- (BT\xp);
   \draw (CT\x) -- (CT\xp);
   \draw (AF\x) -- (AF\xp);
   \draw (BF\x) -- (BF\xp);
   \draw (CF\x) -- (CF\xp);
}
\end{tikzpicture}
\end{center}
\caption{The subgraph spanned by the disks of a clause check gadget and its
surroundings.} \label{fig:clausecheckgadgetUDG}
\end{figure}

Our complete construction can fit in $80$ units of vertical space. Ten units
can accommodate the lower bundle and turning strings up and down from it;
$50$ units of vertical space can accommodate the branching and crossing
gadgets, along with their connections and tapes. We need ten units for the
bundle that goes above the gadgets (along with the string turns), and finally
$20$ more for the side strings and the clause check gadget. Recall that we
did a scaling by two to switch to the intersection model of broadcasting. In
the original model of broadcasting, the construction occupies $40$ units of
vertical space.

In case of a satisfiable formula, we can choose the disks in each side
string that correspond to the value of the variable, and choose a disk from
the clause check gadget that intersects a true literal (at least one of the
literals is true in the clause).

This lemma describes the usage of the clause check gadgets and the side strings.

\begin{lemma}\label{lem:clausecheck}
Let $\varrho$ be the number of true-false disk pairs (blocks) in the three
side strings that correspond to a particular clause checking gadget. An
$h$-hop broadcast set contains at least $\varrho+1$ disks from the three side
strings and the clause check gadget. Moreover, if an  $h$-hop broadcast set
has exactly $\varrho+1$ actives among these disks, then the truth values
chosen at the beginning of the side strings satisfy the clause.
\end{lemma}

\begin{proof}
We prove the following claim first.

\myclaim{A side string cannot contain two empty blocks.}
{Suppose that $U_k$ and $U_{\ell}$ are two empty side string blocks. If they
are not neighboring, then a disk between them is unreachable from the source.
So $\ell=k+1$. Consequently, both disks of $U_k$ are dominated from the start
of the side string, and both disks of $U_{k+1}$ are dominated from the end.
Since the side string has length at more than four, either $k>2$ or $k<p-1$.
Suppose $k>2$, the other case is similar. The only way to reach both disks in
$U_k$ is to have both of the disks in $U_{k-1}$ active. Since $U_{k-1}$ is
also reached from the left, there is an active disk in $U_{k-2}$; let its
truth value be $v$. So we can deactivate the disk in $U_{k-1}$ of value $\neg
v$ and activate the disk in $U_k$ of value $v$. This way every disk that has
been dominated remains dominated, and the number of active disks does not
increase. (Note that we do not need to worry about exceeding $h$ hops since
$h$ will be chosen large enough to not interfere with side strings.)}

Now we show that every  $h$-hop broadcast set includes at least $\varrho+1$
disks from these side strings and the clause check gadget. Suppose there is a
side string of $p$ blocks that contains an empty block $U_k$, and let $v$ be
the truth value of the disk in the last block (the one intersecting the
clause check gadget). Suppose $2 \leq k \leq p-1$; a small variation of the
argument applies to the cases $k=1$ and $k=p$. In $U_{k-1}$, the disk of
value $\neg v$ has to be reached from the beginning of the string --- the
shortest path requires at least $k-2$ active disks in $U_1,\dots,U_{k-2}$. In
$U_{k+1}$, the disk of value $\neg v$ has to be reached through the clause
check gadget; this requires that the clause check disk corresponding to this
side string is active, and there are at least $p-k$ side string actives from
$U_{k+1},\dots, U_p$, since we also need to change truth value along the way.
Additionally, the disk of value $\neg v$ in $U_k$ has to be reached from one
of the neighboring blocks, requiring $U_{k-1}(\neg v)$ or $U_{k+1}(\neg v)$
to be active. Overall, either a side string does not contain an empty block
(so it has at least $p$ disks), or we needed $k-2 + (p-k) + 1 + 1 = p$ active
disks from the side string and the corresponding clause check disk. Moreover,
at least one of the three side strings needs to connect the middle point of
the clause check gadget to the source: the shortest path through a side
string of $p$ blocks has $p+1$ inner vertices, since it has to include one
disk from each block of the side string and the clause check disk
corresponding to this side string. Consequently, we need at least $\varrho+1$
active disks.

Finally, we need to show that if the assignment at the beginning of the side
strings does not satisfy the clause (all literals are false), then we need at
least $\varrho+2$ active disks. A similar argument shows that the string
that reaches the clause check gadget must have one extra active disk.
\qed \end{proof}

\subsection{Reduction from 3-SAT}

Let $\beta$ be the number of branchings, let $\gamma$ be the number of
crossings, let $\xi$ be the number of disk pairs inside strings and side
strings, and let $\tau$ be the number of tape blocks in our construction. We
examine the disks that are necessarily part of a minimum broadcast set if the
formula is satisfiable. It will be apparent that a solution of the same size
cannot exist if the formula is not satisfiable.

We include all the disks from the strings and side strings that correspond to
the value given to the variable, altogether $\xi$ disks. Add the disks from
the gadgets: the blue parent disks inside and one disk from each $X$-block,
altogether $52(\beta+\gamma)$ disks. The branching and crossing gadget
connections require four blue parent disks outside the gadget at the top and
bottom connection, and two more disks on the last gadget (one per branching),
so we require $4(\beta+\gamma)+ 2\beta$ for connections. In each tape we
include one disk from all of its blocks except one. The number of tapes that
connect neighboring gadgets is $2\gamma$, and we also use 2 tapes per string-
gadget connection, so we have $4(\beta+\gamma)+ 2\beta$ such tapes. Thus, the
number of tape disks in a solution is $\tau - (2\gamma + 4(\beta+\gamma)+
2\beta)$. Finally, we use one disk to cover each clause check gadget, overall
$m$ disks. (Recall that $m$ is the number of clauses.)

In case of a satisfiable formula, the total number of disks required for a
canonical broadcast set is
\begin{align*}
&\xi + 52(\beta+\gamma) + 4(\beta+\gamma)+ 2\beta + \tau - \big(2\gamma +
4(\beta+\gamma)+ 2\beta\big) + m\\ =\; &\xi + 52 \beta + 50 \gamma + \tau +
m.
\end{align*}

\begin{theorem}
There is a minimum $h$-hop broadcast set of cost $C=\xi + 52 \beta + 50
\gamma + \tau + m$ if and only if the original 3-CNF formula is satisfiable.
\end{theorem}

\begin{proof}
As we demonstrated previously, if the formula is satisfiable, then there is
an $h$-hop broadcast set of the given size. We need to show that if there is
an $h$-hop broadcast set of this size, then the formula is satisfiable. Take
a minimum  $h$-hop broadcast set. First, we know that the shortest path to
the string ending disks requires exactly $h$ hops, and the only path of this
length includes all blocks of the string in question, plus the shortest way
through the gadgets in which this string is involved. It is easy to check
that the shortest way through a gadget from the string end on the top to the
string end on the bottom uses only blue disks, and one disk from $X'_1$ and
$X'_5$ each. Without loss of generality we can suppose that the  $h$-hop
broadcast set restricted to each gadget is canonical by the analogue of
Lemma~\ref{lem:Mgadget}. Let $t=2\gamma + 4(\beta+\gamma)+ 2\beta$ be the
number of tapes in the construction. A minimum $h$-hop broadcast set must
include at least $\tau - t$ tape disks, and for each clause $i$, at least
$\varrho_i+1$ active disks as shown by Lemmas~\ref{lem:tape}
and~\ref{lem:clausecheck}, where $\varrho_{i}$ is the number of blocks in the
three side strings that correspond to a clause. So a minimum $h$-hop
broadcast set does indeed require at least $C$ disks. An $h$-hop broadcast
set of this size that is canonical when restricted to each gadget also means
that the truth value carried by a string before entering a gadget is the same
as the truth value carried after exiting the gadget. Similarly, the truth
values are transferred between neighboring gadgets connected by a tape pair:
this can be seen by applying Lemma~\ref{lem:tape} for both tapes. And
finally, all clauses must have a true literal at the beginning of at least
one of the corresponding side strings by Lemma~\ref{lem:clausecheck}. Since
the disk choice at the beginning of a side string is forced to comply with
the corresponding string, it follows that the truth values defined by the
strings satisfy the formula.
\qed \end{proof}

Our construction can be built in polynomial time -- note that the
coordinates of each point can be represented with $O(\log n)$ bits, since a
precision of $c/n^4$ is sufficient. We have successfully reduced 3-SAT to
the $h$-hop broadcast problem in a strip of width $40$. Since the problem is
trivially in \np, this concludes the proof of Theorem~\ref{thm:widestripNPC}.

\section{Conclusion}
We studied the complexity of the broadcast problem in narrow and wider
strips. For narrow strips we obtained efficient polynomial algorithms, both
for the non-hop-bounded and for the $h$-hop version, thanks to the special
structure of the problem inside such strips. On wider strips, the broadcast
problem has an $n^{O(w)}$ algorithm, while the $h$-hop broadcast becomes
\np-complete on strips of width $40$. With the exception of a constant width
range (between $\sqrt{3}/2$ and $40$) we characterized the complexity when
parameterized by strip width. We have also proved that the planar problem
(and, similarly, \cdsudg) is \Wone-hard when parameterized by the solution
size. The problem of finding a planar $h$-hop broadcast set seems even
harder: we can solve it in polynomial time for $h=2$ (see
Appendix~\ref{sec:planaralgs}) but already for $h=3$ we know no better
algorithm than brute force. Interesting open problems include:

- What is the complexity of planar $3$-hop broadcast? In particular, is
there a constant value $t$ such that $t$-hop broadcast is \np-complete?

- What is the complexity of $h$-hop broadcast in planar graphs?

\bibliography{broadcast}

\newpage

\appendix

\section{Planar 2-hop broadcast}\label{sec:planaralgs}
To compute a minimum-size broadcast set inside a narrow strip in the
hop-bounded case, we will need a subroutine for the special case of two hops.
For this we provide an algorithm that does not need that the points are
inside a strip of width at most $\sqrt{3}/2$. Since this result is of
independent interest, we provide it in a separate subsection.

Our algorithm is a modification of the $O(n^7)$ algorithm by
Amb\"uhl~et~al.~\cite{Ambuhl04}. Their algorithm works for the case
where one can use different radii for the disks around the points.
For the homogeneous case that we consider (where a point is either active with unit radius
or inactive) we obtain a better bound.
%
%
\begin{theorem}\label{thm:n4}
There is an algorithm that finds a minimum planar 2-hop broadcast set in
$O(n^4)$ time.
\end{theorem}
\begin{proof}
We start by testing if there is a solution consisting of a single disk
(namely~$\disk(s)$) or two disks ($\disk(s)$ and $\disk(p)$ for some $p\neq s$).
This takes $O(n^2)$ time. If we do not find a solution of size one or two, we proceed as
follows.

Let $Q := \{ q_1,\ldots,q_m\}$ be the subset of points in $P$ that are
not covered by~$\disk(s)$, where the points are numbered in counterclockwise
order around~$s$. We define $[i,j]$ to be the set of indices $\{i,\ldots,j\}$
if $i\leq j$, and we define $[i,j]$ to be the set of indices
$\{i,\ldots,m,1,\ldots,j\}$ if $i> j$. Furthermore, we define $Q[i,j]$ to
be the set of points with indices in~$[i,j]$. Let $\Delta$ be the set of
disks (excluding the source disk~$\disk(s)$) that may be useful in a minimum
2-hop broadcast. Obviously any point~$p\in P$ with $\disk(p)\in\Delta$ must
lie inside $\disk(s)$, because the broadcast is 2-hop. Moreover, $\disk(p)$
must contain at least one point $q_i\in Q$ to be useful.

We start by making sure that there is a feasible solution, so by checking
that $Q \subseteq \bigcup \Delta$. The rest of the algorithm is a dynamic
program, but we need several notations to describe it. The values $A[i,j]$ of
our subproblems are defined as follows:
\begin{quotation}
\noindent $A(i,j)$ := the minimum number of disks from $\Delta$ needed to
                      cover all points in $Q[i,j]$.
\end{quotation}
We will prove later that the size of an optimal broadcast set (not counting
the source disk, and assuming that we need at least two disks in addition to
the source disks) is given by
\begin{equation}\label{eq:opt}
\opt = \min_{i,j} \big(A(i,j) + A(j+1,i-1) \big).
\end{equation}

Define $\Delta_i$ to be the set of
disks that can be used to cover a point $q_i\in Q$, that is,
\[
\Delta_i := \{ \delta \in \Delta : q_i \in \delta\}.
\]
Let $\mynext{i}$ be the first index in the sequence~$[i,i-1]$ such that
$Q[i,\mynext{i}]$ cannot be covered by a single disk from $\Delta_i$. (Such
an index must exist since the solution size is at least three.) Furthermore,
for a disk $\delta\in\Delta$, let $\mynext{i,\delta}$ be the first index in
$[i,i-1]$ such that $Q[i,\mynext{i,\delta}]$ cannot be covered by $\delta$.
Thus $\mynext{i} = \max_{\delta\in\Delta} \mynext{i,\delta}$.

We now wish to set up a recurrence for $A(i,j)$. To this end, consider a disk
$\delta\in\Delta$ and the point set $\delta \cap Q[i,j]$. The points in
$\delta\cap Q[i,j]$ need not be consecutive in angular order around~$s$: the
disk $\delta$ may first cover a few points from $Q[i,j]$ (until
$q_{\mynext{i,\delta}-1}$), then there may be some points not covered, then
it may cover some points again, and so on; see Fig.~\ref{fig:2hop-broadcast}
where the angular ranges containing covered points are indicated in gray.
\begin{figure}[bt]
    \begin{center}
    \includegraphics{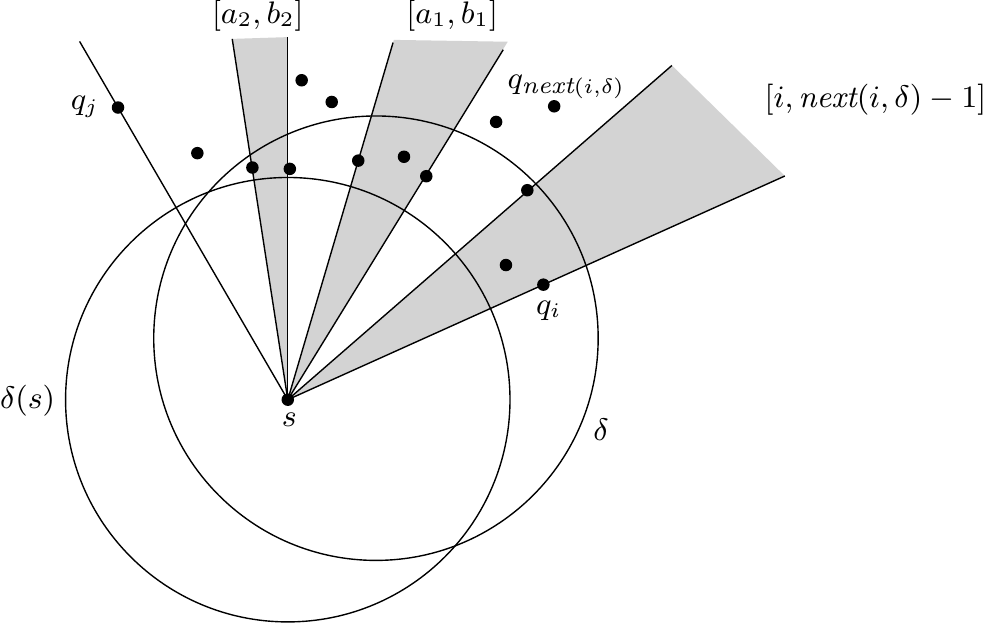}
    \end{center}
    \caption{Definition of the intervals~$[a_i,b_i]$}
    \label{fig:2hop-broadcast}
\end{figure}
We can thus define a set of maximal intervals that together form $\delta\cap Q[i,j]$:
\[
\delta \cap Q[i,j]
    = Q[i,\mynext{i,\delta}-1] \cup Q[a_1,b_1] \cup Q[a_2,b_2] \dots \cup Q[a_t,b_t].
\]
Now define $\mathfrak{I}(i,j,\delta)$ as
\[
\mathfrak{I}(i,j,\delta) :=  [a_1-1,b_1+1] \cup [a_2-1,b_2+1] \dots \cup [a_t-1,b_t+1].
\]
We claim that we now have the following recurrence:
\begin{align}\label{eq:dpAn4}
A&(i,j)=\nonumber\\
&\begin{cases}
1 & \text{ if } i=j \\[1ex]
1 \!+\! \displaystyle \min \Big\lbrace A(\mynext{i},j),
\displaystyle \mkern-25mu \min_{\substack{\delta \in \Delta_i\\ (a,b) \in \mathfrak{I}(i,j,\delta)} } \mkern-25mu
\left(A(\mynext{i,\delta},a) \!+\! A(b,j) \right)\Big\rbrace & \text{ otherwise}
\end{cases}
\end{align}
We need to establish some key properties to prove the correctness of this
recurrence. Let $D$ be the set of active points in a minimum-size $2$-hop
broadcast.  We call a disk $\disk(p)$ of an active point~$p$
an \emph{active disk}. Let $\cu(D) := \bigcup \{\disk(p)\;:\;p\in D\}$ be
the union of the active disks.
\begin{observation}
The region $\cu(D)$ is star-shaped with respect to the source point $s$, that is,
for any point $z$ in $\cu$, the segment $sz$ is inside $\cu(D)$.
\end{observation}
\begin{proof}
Let $p \in D$ be a point such that $z\in\disk(p)$. Suppose for contradiction
that there is a point $t \in sz$ that lies outside $\cu(D)$, and let $\ell$ be
the perpendicular bisector of $tz$. Since $t \not\in \disk(p)$, point $p$
lies on the same side of $\ell$ as $z$. Note that since $t \not\in
\disk(s)$, the disk $\disk(s)$ is entirely covered by the other half plane
of $\ell$. Thus $p \not\in \disk(s)$, which is a contradiction since in a
2-hop broadcast set we have $D\subset \disk(s)$.
\qed \end{proof}
Let $\partial \cu(D)$ be the boundary of $\cu(D)$. By the previous
observation, $\partial \cu(D)$ is connected for 2-hop broadcast sets.
Note that a point~$q\in Q$ can be covered by multiple active disks.
We will assign a unique point $\mypred(q)\in D$ whose disk covers $q$ to
each $q\in Q$, as follows. We call $\mypred(q)$ the \emph{predecessor}
of $q$ (in the given solution~$D$) because $\mypred(q)$ can be thought
of as the predecessor of $q$ in a broadcast tree induced by~$D$.
Let $\myray(q)$ be the ray emanating from
$s$ and passing through~$q$, and consider the point~$z$ where $\myray(q)$ exists
$\cu(D)$. Then we define $\mypred(q)$ to be the point that is the center
of the active disk~$\delta$ on whose boundary~$z$ lies (with ties
broken arbitrarily, but consistently).

Recall that the points in $Q$ are numbered in angular order around~$s$,
and consider the circular sequence
$\sigma(D) := \langle \mypred(q_1),\ldots,\mypred(q_m) \rangle$.
We modify $\sigma(D)$ by replacing any consecutive subsequence consisting
of the same point by a single occurrence of that point. For example,
we would modify $\langle p,p,p,q,q,p,p,r,r,r,p\rangle$ to obtain $\langle p,q,p,r,p\rangle$.
\begin{observation}\label{obs:ijij}
In a 2-hop broadcast set~$D$, the boundary sequence $\sigma(D)$ has no cyclic
subsequence $\cdots p \cdots p' \cdots p \cdots p'$ with $p \neq p'$.
\end{observation}
\begin{proof}
Between two adjacent occurrences of $p$ and $p'$ on the boundary, there must
be an intersection between $p$ and $p'$. Since there can be at most two
intersections between two circles, the sequence $\cdots p \cdots p' \cdots p \cdots p'$
cannot occur in~$\sigma$.
\qed \end{proof}
\begin{lemma} \label{lem:atmost-twice}
In a 2-hop broadcast set~$D$, any point~$p\in D$ can appear in $\sigma(D)$ at most
twice.
\end{lemma}
\begin{proof}
Consider the part of the boundary $\partial \disk(p)$ lying outside
the source disk~$\disk(s)$. This boundary part, which we denote by $\gamma$,
can be partitioned into arcs where $\partial\disk(p)$ defines $\partial \cu(D)$
and arcs where it does not. Assume for a contradiction that there
are three arcs where $\partial\disk(p)$ defines $\partial \cu(D)$---obviously
this is necessary for $p$ to appear three times in~$\sigma(D)$.
Then there must be two arcs, $\gamma_1$ and $\gamma_2$, where
$\partial\disk(p)$ does not define $\partial \cu(D)$ and such that
$\gamma_1$ and $\gamma_2$ lie fully in the interior of~$\gamma$.
Let $\alpha(\gamma)$ denote the opening angle of the cone with apex $p$
defined by $\gamma$, and define $\alpha(\gamma_1)$ and $\alpha(\gamma_2)$
similarly; see Fig.~\ref{fig:cones}.
\begin{figure}[bt]
    \begin{center}
    \includegraphics{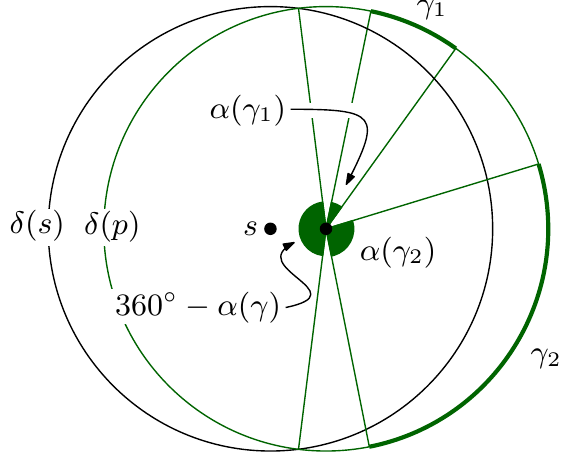}
    \end{center}
    \caption{Illustration for the proof of
    Lemma~\protect\ref{lem:atmost-twice}.}
    \label{fig:cones}
\end{figure}
It is easy to see that $\alpha(\gamma)\leq 240^{\circ}$.
Since $\gamma_1$ and $\gamma_2$ do not cover $\gamma$ completely
then one of them, say $\gamma_1$, must be less than $120^{\circ}$.
We will show that this leads to a contradiction, thus proving the lemma.

Let $\disk(p')$ be a disk covering (part of) $\gamma_1$. Since $\disk(p')$
covers less than $120^{\circ}$ of $\gamma_1$, its center $p'$ must lie
outside $\disk(p)$. On the other hand, $p'$ must lie inside $\disk(s)$,
since we have a 2-hop broadcast and $p'\in D$.
Now observe that $p'$ lies on the ray $\rho$ starting at $p$
that goes through the midpoint of the arc $\gamma_1\cap \disk(p')$. This is
a contradiction because $\rho$ is disjoint from $\disk(s) \setminus
\disk(p)$.  
\qed \end{proof}
We are now ready prove the correctness of our algorithm.
\medskip

First consider Equation~\eqref{eq:opt}. It is clear that
\[
\opt \le \min_{i,j} \left(A(i,j) + A(j+1,i-1) \right)
\]
since the union of the best covering of $Q[i,j]$ and
$Q[j+1,i-1]$ is a feasible covering.

To prove the reverse, let~$D$ be a minimum-size 2-hop broadcast set.
Suppose some point, $p$, appears only once in $\sigma(D)$.
Let $i,j$ be such that $\{q\in Q: \mypred(q) = p\} = Q[i,j]$.
Then $A(i,j)=1$ and there is a covering of $Q[j+1,i-1]$ with
$|D|-1$~disks. Hence, $\min_{i,j} \left(A(i,j) + A(j+1,i-1) \right) \leq \opt$
in this case.
If all points appear twice in $\sigma(D)$ then we can argue as follows.
Consider a point $p\in D$, and let $i_1,j_1$ and $i_2,j_2$ be such that
$\{q\in Q: \mypred(q) = p\} = Q[i_1,j_1] \cup Q[i_2,j_2]$. Then the set
of disks used by $D$ in the covering of $Q[i_1,j_2]$ is disjoint from the
set of disks used by $D$ in the covering of $Q[j_2+1,i_1-1]$ by
Observation~\ref{obs:ijij}. Hence, $\left(A(i_1,j_2) + A(j_2+1,i_1-1) \right)
\leq \opt$.
\medskip

Next, we prove that the recursive formula~\eqref{eq:dpAn4} holds. We prove
this by induction on the length of~$[i,j]$. If $i=j$, then $A(i,j)=1$ is
correct since our initial feasibility check implies that there is at least
one disk $\delta \in \Delta$ that can cover $q_i$.
Now consider the case $i\neq j$. First we note that
\[
A(i,j) \leq 1+ \displaystyle \min \Big\lbrace A(\mynext{i},j),
\displaystyle \min_{\substack{\delta \in \Delta_i\\ (a,b) \in \mathfrak{I}(i,j,\delta)} }
\left(A(\mynext{i,\delta},a) + A(b,j) \right)\Big\rbrace.
\]
Indeed, there is a disk covering $Q[i,\mynext{i}-1]$ by definition of $\mynext{i}$
and we can cover $Q[\mynext{i},j]$ by $A(\mynext{i},j)$ disks by induction.
Similarly, the definition of $\mathfrak{I}(i,j,\delta)$ implies that
any disk $\delta\in\Delta_i$ covers $Q[i,\mynext{i,\delta}-1]$
and $Q[a+1,b-1]$. By induction we can thus cover $Q[i,j]$
by $1+A(\mynext{i,\delta},a) + A(b,j)$ disks.

To prove the reverse, let $D$ be a minimum-size 2-hop broadcast for
$Q[i,j]$ and let $p := \mypred(q_i)$. If $p$ appears in the covering of
$Q[i,j]$ only once, then $A(i,j)=1+A(\mynext{i},j)$. Otherwise $p$ appears
twice by Lemma~\ref{lem:atmost-twice}. Let $q_a$ be the last point before the
second appearance of $p$ in $\sigma(D)$, and let $q_b$ be the first point
after the second appearance of $p$ in $\sigma$. By
Observation~\ref{obs:ijij}, the coverings of $Q[\mynext{i,\delta},a]$ and
$Q[b,j]$ are disjoint in $D$. Hence, $1+\left(A(\mynext{i,\delta},a) +
A(b,j) \right) \leq |D|$. We conclude that
\[
A(i,j) \geq 1+ \displaystyle \min \Big\lbrace A(\mynext{i},j),
\displaystyle \min_{\substack{\delta \in \Delta_i\\ (a,b) \in
\mathfrak{I}(i,j,\delta)} } \left(A(\mynext{i,\delta},a) + A(b,j)
\right)\Big\rbrace. \]
\medskip

It remains to analyze the running time. The algorithm works by first
computing $\mynext{i}$, $\mynext{i,\delta}$ and $\mathfrak{I}(i,j,\delta)$
for each $i,j$ and $\delta \in \Delta$. This can easily be done in $O(n^4)$
time. Running the dynamic program using the recursive
formula~\eqref{eq:dpAn4} then takes $O(n^4)$ time, as we have $O(n^2)$
entries $A(i,j)$ that each can be computed in $O(n^2)$ time. Finally,
computing the optimal solution using Equation~\eqref{eq:opt} takes $O(n^2)$m
time. Hence, the overall time requirement is $O(n^4)$, while the space
required is $O(n^2)$. Computing an optimal solution itself, rather than just
the value of $\opt$, can be done in a standard manner, without increasing the
time or space bounds.
\qed \end{proof}

\end{document}